\documentclass[10pt,conference,letterpaper]{IEEEtran}

\usepackage{array}
\usepackage{amsfonts,amsmath,amsopn,amssymb,amstext,amsthm,mathrsfs}
\usepackage{cite}
\usepackage{dsfont}
\usepackage{multirow}
\usepackage[tight,footnotesize]{subfigure}
\usepackage[usenames]{color}
\usepackage[utf8]{inputenc}
\usepackage{color}
\usepackage{relsize}
\ifCLASSINFOpdf
\usepackage[pdftex]{graphicx}
\else
\usepackage[dvips]{graphicx}
\fi

\newcommand{\ee}{{\mathrm e}}
\newcommand{\natu}{{\mathbb N}}

\newcommand{\Erl}{\mathrm{Erl}}
\newcommand{\LRE}{\mathrm{LRE}}

\newtheorem{theorem}{Theorem}

\newtheorem{corollary}{Corollary}

\newtheorem{remark}{Remark}

\newcommand{\pr}[1]{{\mathbb P}\left\{{\displaystyle#1}\right\}}

\begin{document}

\title{Delay versus Stickiness Violation Trade-offs \\
for Load Balancing in Large-Scale Data Centers}

\author{Qingkai Liang$^\dagger$, Sem Borst$^{\star}$ \\
$^\dagger$Laboratory for Information and Decision Systems, MIT, Cambridge, MA \\
$^\star$Nokia Bell Labs, 
Murray Hill, NJ}

\maketitle

\begin{abstract}
Most load balancing techniques implemented in current data centers tend to rely on a mapping from packets to server IP addresses through a hash value calculated from the flow five-tuple. The hash calculation allows extremely fast packet forwarding and provides flow `stickiness', meaning that all packets belonging to the same flow get dispatched to the same server. Unfortunately, such static hashing may not yield an optimal degree of load balancing, e.g.~due to variations in server processing speeds or traffic patterns. On the other hand, dynamic schemes, such as the Join-the-Shortest-Queue (JSQ) scheme, provide a natural way to mitigate load imbalances, but at the expense of stickiness violation.

In the present paper we examine the fundamental trade-off between stickiness violation and packet-level latency performance in large-scale data centers. We establish that stringent flow stickiness carries a significant performance penalty in terms of packet-level delay. Moreover, relaxing the stickiness requirement by a minuscule amount is highly effective in clipping the tail of the latency distribution. We further propose a bin-based load balancing scheme that achieves a good balance among scalability, stickiness violation and packet-level delay performance. Extensive simulation experiments corroborate the analytical results and validate the effectiveness of the bin-based load balancing scheme.
\end{abstract}

\section{Introduction}
\label{intr}

Load balancing is a key mechanism for achieving efficient
resource allocation in data centers, ensuring high levels of server
utilization and robust application performance.
%Efficient server utilization is a particularly critical objective
%as server resources tend to account for the dominant share of the
%capital and operating expenses in the majority of data centers.
%Reliable performance in terms of low latency and high throughput is
%equally crucial, especially in the context of network virtualization
%and migration of vital applications to cloud platforms and telco
%data centers.
Unfortunately, most load balancing techniques implemented in current
web data centers are prone to delivering poor server utilization,
uneven performance, or both.
Standard load balancing mechanisms map incoming packets to server
IP addresses via a hash value calculated from the flow five-tuple
in the packet header.
The hash calculation allows extremely fast packet forwarding
(at line speed) and guarantees flow `stickiness', in the sense that all
packets belonging to the same flow get dispatched to the same server,
which is essential for smooth execution of stateful applications
(where relevant state information must be stored in server memory
for the duration of the flow, e.g. for transaction, accounting
or authentication purposes).

For uniform hashing functions, the above-described mechanisms map
packets across servers with equal probability.
This translates into equal long-term utilization levels in case all
servers have identical processing capacities.
In practice, however, processing speeds and traffic characteristics
of flows show notable variation, resulting in substantial imbalances
between traffic loads and server capacities.
Such a mismatch manifests itself in underutilization of some servers,
and serious performance degradation at others, which can only be
mitigated through costly overprovisioning. 

%The structural load imbalances may be further exacerbated due to
%non-uniformity of the hash function.
%Also, server utilization levels over shorter time intervals will
%exhibit even greater discrepancies than the long-term averages,
%especially when the flow rates or flow durations show significant
%variability or possibly heavy-tailed characteristics. \\

In the literature, several sophisticated dynamic load balancing algorithms
have been considered in order to address the above-mentioned issues
by exploiting state information in various ways.
%The merits of these approaches in terms of server utilization
%and performance must be weighed with other relevant criteria, such as
%scalability, load balancing granularity and communication overhead.
For example, in the Join-the-Shortest Queue (JSQ) scheme,
each arriving packet is routed to the server with the minimum
current queue length.
Such a strategy minimizes the average packet latency in symmetric systems
with exponential service time distributions~\cite{Winston77}.
Unfortunately, a fundamental limitation of the JSQ scheme is its
poor scalability when implemented at packet level, since the
required amount of information exchange is proportional to the
number of servers as well as the packet arrival rate,
which could be prohibitive for even a moderate-size data center.

The scalability issue has motivated a strong interest in so-called
power-of-$d$ policies, where an incoming packet is routed to the
server with the minimum current queue length among $d$~randomly
selected servers.
Even for values as low as $d = 2$, these policies significantly
outperform randomized splitting (which corresponds to $d = 1$)
\cite{Mitzenmacher01,VDK96}.
In case of batch arrivals, the communication burden can be further
amortized over multiple packets~\cite{YSK15}.
Power-of-$d$ policies also extend to heterogeneous scenarios
and loss systems (rather than single-server queueing settings)
\cite{MKM16,MKMG15,MMG15}.
Despite their better scalability, power-of-$d$ policies still suffer
from high communication overhead, since the required amount
of information exchange is proportional to the packet arrival rate.

Besides the above `push-based' schemes, where a dispatcher directs
traffic to servers, significant attention has recently 
focused on alternative `pull-based' schemes where idle servers solicit traffic
by advertising their availability to the dispatcher
\cite{BB08,LXKGLG11,Mitzenmacher16,Stolyar15}.
These pull-based schemes provide effective solutions for assigning
jobs to servers for sequential processing,
with even lower communication overheads and better performance than power-of-$d$ policies.
%but do not seem suitable when servers handle several active packet flows concurrently.

In spite of their superior performance and improved scalability, the aforementioned dynamic load balancing schemes (i.e., JSQ, power-of-$d$ and pull-based schemes) are `flow-agnostic' in the sense that they ignore the flow stickiness requirement that all packets belonging to the same flow should be dispatched to the same server.
Flow stickiness could be preserved by implementing the above dynamic load balancing schemes at flow level, where each flow (instead of each individual packet) gets dispatched to servers according to some pre-defined rule such as JSQ. 
However, such a flow-level implementation requires us to maintain a flow table that records which flow is dispatched to which server. 
In large-scale deployments with massive numbers
of flows, the associated flow table may quickly become unmanageable.  More importantly, the flow stickiness requirement inevitably degrades the packet-level delay performance achieved by a load balancing scheme, since a coarser load balancing granularity has to be used (load balancing has to be performed at flow level instead of packet level). On the other hand, if some degree of stickiness violation is allowed, the packet-level delay performance should improve. Thus there is a trade-off between stickiness violation and delay performance.

In the present paper we examine the above-mentioned trade-offs
between flow stickiness violation and packet-level latency performance, which is governed by complex interactions between load balancing dynamics
and both flow and packet-level traffic activity patterns.
To the best of our knowledge, this paper is the {\bf first to explore
load balancing from a joint flow and packet-level perspective}.
We establish that stringent stickiness carries a significant
performance penalty in terms of packet-level delay.
In particular, even the simplest packet-level load balancing policy
outperforms the best flow-level load balancing policy that maintains
stringent stickiness.
Moreover, we show that a minor level of stickiness violation tolerance
is highly effective in clipping the tail of the latency distribution.
We further propose an efficient bin-based load balancing scheme
that achieves a good balance among scalability, flow stickiness
violation and packet-level delay performance.
Extensive simulation experiments corroborate the analytical results
and validate the effectiveness of the bin-based load balancing scheme.

The remainder of the paper is organized as follows.
In Section~\ref{sec:prelin} we present a detailed description of the
system model and relevant performance metrics.
In Sections~\ref{sec:flow-level-lb}--\ref{sec:simulation-flow-level}
we examine the trade-off between flow stickiness violation
and packet-level latency for various flow-level load balancing
schemes through a combination of analytical methods based
on mean-field limits and simulation experiments.
In Section~\ref{sec:bin-level} we propose a bin-based load
balancing scheme and investigate its performance through simulation.
Conclusions are presented in Section~\ref{conc}.

\section{Preliminaries}
\label{sec:prelin}

\subsection{Model Description}
\label{mode}

We consider a data center with $n$~parallel servers (virtual machines)
and a single dispatcher.
Flows are initiated as a Poisson process of rate $n \lambda$, and have
independent and exponentially distributed durations with mean~$\beta$. Most of the results below extend to phase-type distributions however,
at the expense of unwieldy notation.
Denote by $\rho = \lambda \beta$ the average load per server (in terms
of number of active flows).
Note that the durations of the various flows do not depend on the
number of flows contending for service, but the perceived performance
(e.g.~in terms of packet-level delay) does deteriorate
with an increasing number of concurrent flows.
These characteristics pertain for instance to video streaming sessions
or basic network processing functions.

Over the duration of a flow, packets are generated according to some
stochastic process with rate~$\nu$.
Each packet has a random processing time with mean $\frac{1}{\mu_j}$
at server~$j$. 
To facilitate a transparent analysis, we focus on homogeneous scenarios
where $\mu_j=\mu$ for any $j=1,\cdots,n$.
It is worth recalling that load balancing schemes are particularly
targeted for heterogeneous scenarios, where server processing
capacities may be different or unknown.
However, homogeneous scenarios provide pertinent insights in the
stationary behavior for heterogeneous scenarios after convergence
of the load balancing process to an equilibrium where the structural
imbalances have been remedied. 

%For each packet a hashing value is computed at the dispatcher
%based on the flow five-tuple in the packet header.
%For conciseness, the $m$~different hashing values will be referred
%to as bins.
%Each of the bins is assigned to a particular server as governed
%by a bin table represented in terms of an $m \times n$ binary matrix
%$x_{jk}$, with $\sum_{k = 1}^{n} x_{jk} = 1$ for all $j = 1, \dots, m$.
%When a packet is hashed into bin~$j$ and $x_{jk} = 1$, the dispatcher
%will forward that packet to server~$k$ for processing.
%Note that all packets belonging to the same flow are hashed to the same
%bin, and thus forwarded to the same server, ensuring flow stickiness,
%as long as the corresponding entry in the bin table does not change.
%Through the packet hashing, each active flow is also uniquely associated
%with a particular bin.
%When a particular flow is associated with bin~$j$ and $x_{jk} = 1$,
%that flow is said to be assigned to server~$k$.

For many stateful applications, flow stickiness is required, i.e.,
packets belonging to the same flow should be dispatched to the same server,
otherwise the application may not function properly or experience
severe performance degradation.
To preserve perfect flow stickiness, load balancing has to be performed
at the granularity of flows as opposed to packets, i.e., each newly
initiated flow is dispatched to some server and all packets in that
flow are directed to the same server.
Consequently, the flow stickiness requirement degrades the packet-level
delay performance as compared to packet-level load balancing
(as will be quantified in Sections
\ref{flow-perfect-stickiness}--\ref{sec:simulation-flow-level}).
One way to alleviate such performance degradation is to introduce
some degree of stickiness violation, which leads to a crucial trade-off
between stickiness violation and packet-level delay performance.

\subsection{Performance Metrics}

Next we introduce the performance metrics used to measure the
stickiness-delay trade-off.

\subsubsection{Metric for Stickiness Violation}
\label{sec:metric-stickiness}

The degree of stickiness violation is measured by the stickiness
violation probability~$\epsilon$, i.e., the probability that
a flow gets dispatched to more than one server.
It measures the long-term fraction of flows that do not preserve
stickiness.

\subsubsection{Metric for Packet-level Performance}
\label{sec:metric-packet}

As mentioned earlier, the duration of a flow does not depend on the
number of concurrent flows contending for service,
but the perceived performance (e.g.~packet-level latency) does
strongly vary with the number of concurrent flows at the same server.
In order to capture that dependence, we adopt the usual time scale
separation assumption between packet-level dynamics and flow-level
dynamics.
%That assumption yields a quasi-stationary approximation, implying
%that the distribution of the packet-level state for a given number
%of $i$~flows is the same as the stationary distribution
%in a scenario with a constant number of $i$~flows.
In particular, we suppose that the relevant packet-level performance
metric at each individual server can be described as a function $G(i)$
of the number of concurrent flows~$i$ at that server.

%Let $Q(t) = (Q_1(t), \dots, Q_n(t))$ represent the flow population
%at time~$t$, with $Q_k(t)$ denoting the number of flows assigned
%to server~$k$ at time~$t$.
For any vector $Q \in \natu^n$,
let $\pi(Q) = \lim_{t \to \infty} \pr{Q(t) = Q}$ be the probability
that the flow population is~$Q$ in stationarity.
Then the relevant average packet-level performance in stationarity is
\begin{equation}
\widetilde{G} = \frac{\sum_{Q \in \natu^n} \pi(Q) \sum_{k = 1}^{n} Q_k G(Q_k)}
{\sum_{Q \in \natu^n} \pi(Q) \sum_{k = 1}^{n} Q_k},
\label{perf1}
\end{equation}
%It is easily verified that the function $\sum_{k = 1}^{n} Q_k F(Q_k)$
%is Schur-convex, and hence the more balanced the flow population
%vector~$Q$ (in the sense of weak majorization) for a given value
%of $\sum_{k = 1}^{n} Q_k$, the smaller $\sum_{k = 1}^{n} Q_k F(Q_k)$.
which may equivalently be expressed as
\begin{eqnarray}
\widetilde{G}
%&=&
%\frac{\sum_{Q \in \natu^n} \pi(Q) \sum_{k = 1}^{n} \sum_{i = 1}^{\infty} i F(i) \indi{Q_k = i}}
%{\sum_{Q \in \natu^n} \pi(Q) \sum_{k = 1}^{n} \sum_{i = 1}^{\infty} i \indi{Q_k = i}} \nonumber \\
%&=&
%\frac{\sum_{i = 1}^{\infty} i F(i) \sum_{Q \in \natu^n} \pi(Q) \sum_{k = 1}^{n} \indi{Q_k = i}}
%{\sum_{i = 1}^{\infty} i \sum_{Q \in \natu^n} \pi(Q) \sum_{k = 1}^{n} \indi{Q_k = i}} \nonumber \\
&=&
\frac{\sum_{i = 1}^{\infty} i G(i) p_i}{\sum_{i = 1}^{\infty} i p_i},
\label{perf2}
\end{eqnarray}
with $p_i$ being the expected fraction of servers with exactly $i$~flows.

It is worth emphasizing that the expressions~(\ref{perf1})
and~(\ref{perf2}) for the packet-level performance only entail
a generic time scale separation assumption.
\textbf{They do not rely on a particular performance criterion or specific
properties of the packet-level traffic characteristics,
which are entirely encapsulated by the function $G(\cdot)$.}
A prototypical performance criterion would be the fraction of packets
in stationarity for which the perceived delay exceeds a threshold
$\tau_\chi$ equal to $\chi \geq 1$ times the mean processing time.
The derivation of the function $G(\cdot)$ then involves a separate
queueing analysis, which is fairly case-specific and mostly orthogonal
to the central theme of the present paper, and therefore not pursued
in any detail or generality.
As a brief illustrative example, which will be used in the numerical
experiments, suppose that each active flow generates packets
as a Poisson process of rate $\nu \gg 1 / \beta$,
and that the packets have independent and exponentially distributed
processing times with mean $\frac{1}{\mu}$.
We assume that $\lceil \rho \rceil \nu < \mu$ to ensure that the
system in its entirety is not overloaded.
Thus, when there are~$i$ active flows at a particular server,
the number of packets evolves as the queue length in an M/M/1 system
with arrival rate $i \nu$ and service rate~$\mu$.
The probability that the packet delay exceeds some value~$t$ in that
case equals $\ee^{- (\mu - i \nu) t}$, assuming $i \leq \mu / \nu$.
In particular, the probability that the packet delay exceeds the mean
processing time by a factor~$\chi$ is $\ee^{- \chi (1 - i \nu / \mu)}$.
Thus we obtain
\begin{equation}\label{poiss-pkt}
G_\chi(i) =
\left\{\begin{array}{ll} \ee^{- \chi (1 - i \nu / \mu)} & i \leq \mu / \nu \\
1 & i > \mu / \nu \end{array} \right.
\end{equation}
Plugging into \eqref{perf2}, we derive a specific packet-level
metric $\widetilde{G}_\chi$:
\begin{equation}
\widetilde{G}_\chi =
\frac{\sum_{i = 1}^{\infty} i G_\chi(i) p_i}{\sum_{i = 1}^{\infty} i p_i}.
\label{perf3}
\end{equation}
This metric will be referred to as the $\chi$-delay tail probability,
i.e., the probability that a packet experiences a delay exceeding
$\chi$~times the mean processing time.
In the rest of the paper, the $\chi$-delay tail probability
$\widetilde{G}_\chi$ will be frequently used as an illustrative
example of the generic packet-level performance metric $\widetilde{G}$,
in both numerical experiments and analytical results.

Finally, note that once the function $G(i)$ has been determined,
the packet-level performance metric $\widetilde{G}$ only depends
on the stationary distribution of the flow population, namely~$p_i$.
Thus, in the analysis of packet-level performance, it suffices
to derive the expression for~$p_i$ and we omit the complete
expression for $\widetilde{G}$.

\section{Flow-Level Load Balancing}
\label{sec:flow-level-lb}

In the next few sections, we investigate the fundamental trade-off
between stickiness violation and packet-level latency under various
flow-level load balancing schemes.
%The previous section showed that evaluating the packet-level
%performance comes down to deriving the stationary distribution
%of the flow population in terms of the probabilities $\pi(\cdot)$
%or just the fractions~$p_i$.
We begin by discussing two flow-level load balancing schemes that
maintain perfect flow stickiness in Section~\ref{flow-perfect-stickiness}.
The analysis of these schemes demonstrates that stringent stickiness
carries a significant performance penalty in terms of packet-level delay.
Then in Section~\ref{flow-level-stickiness-violation} we investigate
several flow-level load balancing schemes that allow stickiness
violation to some extent.
It will be shown that relaxing the stickiness requirement
by a minuscule amount is highly effective in clipping the tail of the
latency distribution.

As it turns out, for virtually any state-dependent flow assignment scheme,
an exact analysis of the stationary distribution of the flow population
does not appear to be tractable.
For the sake of tractability, we therefore pursue mean-field limits
in an asymptotic regime where the number of servers~$n$ grows large.
While load balancing at flow level may not be feasible at that scale
(since we need to maintain a flow assignment table whose size is
proportional to the number of active flows in the system),
such a scenario allows derivation of explicit mean-field limits
and sheds light on the fundamental trade-off between flow stickiness
violation and packet-level latency.
Note that the mean-field regime is not only convenient
from a theoretical perspective, but also relevant from a practical
viewpoint given the large number of servers in typical data centers.

For convenience, we henceforth assume that the decisions of the load balancing scheme only depend
on the history of the process through the current flow population~$Q$,
so that the process $\{Q(t)\}_{t \geq 0}$ behaves as a Markov process.
Most of the results below extend to phase-type distributions however,
at the expense of unwieldy notation.
Introduce $S^n(t) = (S_i^n(t))_{i \geq 0}$, with $S_i^n(t)$
representing the fraction of servers with $i$ or more flows
at time~$t$ (when there are $n$ servers),
with $S^n_0(t) \equiv 1$, and observe that the process
$\{S^n(t)\}_{t \geq 0}$ also evolves as a Markov process.
Then any weak limit $s(t)$ of the sequence $\{S^n(t)\}_{t \geq 0}$
as $n \to \infty$ is called a mean-field limit.
Rigorous proofs to establish weak convergence to the mean-field
limit are beyond the scope of the present paper,
but can be constructed along similar lines as in~\cite{HK94}.
The function $s(t)$ can typically be described as a dynamical system,
and any stationary point $s^*$ of $s(t)$ is referred to as a fixed point.
Denote by $p_i^* = s_i^* - s_{i+1}^*$ the corresponding fraction
of servers with exactly $i$~flows.
Throughput we will suppose (without formal proof) that the
mean-field ($n \to \infty$) and steady-state ($t \to \infty$)
limits may be interchanged, and thus interpret $p_i^*$ as the limit
of the stationary probability that a particular server has exactly
$i$~flows when the total number of servers grows large.

%In this appendix we analyze mean-field limits for various flow-level
%load balancing schemes mentioned in Section \ref{sec:flow-level-lb}.
We focus the attention on schemes that dispatch arriving flows based
on the numbers of flows at the various servers only,
and not the identities of the individual servers,
as is quite natural in homogeneous scenarios.
The number of flows at the server to which an arriving flow is assigned,
is then a random variable that depends on the current flow
population~$Q$ only through the vector~$S^n$, with $S_i^n$
representing the fraction of servers with $i$ or more flows.

The evolution of the mean-field limit may then in general be
described in terms of differential equations of the form
\[
\frac{{\rm d} s_i(t)}{{\rm d} t} =
\lambda q_{i-1}(s(t)) - i (s_i(t) - s_{i+1}(t)) / \beta,
\hspace*{.4in} i \geq 1.
\]
The terms $q_{i-1}(s)$ may be interpreted as the probabilities that
an arriving flow is assigned to a server with exactly $i - 1$ flows
when the mean-field state is~$s$, and depend on the specific flow
assignment scheme under consideration.
We note that there are many different sequences $S^n$ with the same
limit~$s$, for which the limiting probabilities $q_{i-1}(S^n)$
could be different.
For the schemes that we will consider, however, the limiting
probabilities $q_{i-1}(S^n)$ are uniquely determined by the
mean-field state~$s$, and hence we write $q_{i-1}(s)$
with minor abuse of notation.

The fixed point of the above system of differential equations is
in general determined by
\begin{equation}
\rho q_{i-1}(s) = i (s_i - s_{i+1}),
\hspace*{.4in} i \geq 1,
\label{fp1}
\end{equation}
or equivalently,
\begin{equation}
\rho q_{i-1}(p) = i p_i,
\hspace*{.4in} i \geq 1,
\label{fp2}
\end{equation}
where the probabilities $q_{i-1}(s)$ depend on the specific flow
assignment scheme under consideration.
In the next two sections we will derive the probabilities $q_{i-1}(s)$
and solve the above fixed-point equations for various specific schemes.

\section{Flow-level Load Balancing with Perfect Stickiness}
\label{flow-perfect-stickiness}

In this section we consider the following two flow-level load balancing
schemes which both provide perfect flow stickiness:
(I) power-of-$d$ flow assignment;
(II) pull-based flow assignment.\\

\subsubsection*{Scheme (I): power-of-$\mathlarger{\mathlarger{d}}$ flow assignment}

In this scheme, an arriving flow is assigned to a server with the
minimum number of active flows among $d$~randomly selected servers,
where $1 \leq d \leq n$.
Thus the probabilities $q_{i-1}(s)$ in Equations~(\ref{fp1}) are given by
\begin{equation}\label{eq:jsq-q}
q_{i-1}(s) = s_{i-1}^d - s_i^d,
\hspace*{.4in} i \geq 1.
\end{equation}
This scheme covers several existing flow-level load balancing schemes
as special cases, such as randomized flow assignment ($d = 1$)
and the flow-level Join-the-Shortest-Queue (JSQ) policy ($d = n$).

It is difficult to derive an explicit expression for the fixed point(s)
of Equations~(\ref{fp1}) in this case.
We present a useful upper bound in the next theorem.

\begin{theorem}\label{thm:pod}
Denote $k^* = \lfloor \rho \rfloor$, and let $s^*$ be any fixed point under the flow-level power-of-$d$ policy.
Then $s^*_i \leq \hat{s}_i$ for all $i \geq 0$, where
\[
\hat{s}_i =
\begin{cases}
1, & 0 \leq i \leq k^* \cr
\big(\frac{\rho}{k^*+1}\big)^{\frac{d^{i - k^*} - 1}{d - 1}}, & i \geq k^* + 1.
\end{cases}
\]
\end{theorem}

\begin{proof}
We prove the theorem by induction.

\textbf{Base Case:}
For any $0\le i\le k^*$, it is clear that $s_i \leq \hat{s}_i = 1$.

\textbf{Inductive Step:}
Suppose $s_{k^*+j}\le \hat{s}_{k^*+j}$ for some $j \geq 1$.
Then we prove that $s_{k^*+j+1} \leq \hat{s}_{k^*+j+1}$.
Summing both sides of Equation~\eqref{fp1} over $i \ge k^*+j+1$, we obtain
\[
\sum_{i=k^*+j+1}^{\infty} \rho q_{i-1}(s) =
\sum_{i=k^*+j+1}^{\infty} i (s_i-s_{i+1}).
\]
Substituting~\eqref{eq:jsq-q} into the above equation, we derive
\[
\begin{split}
\rho s_{k^*+j+1}^d = &(k^*+j+1) s_{k^*+j+1} + s_{k^*+j+2} + \cdots \\
\geq &
(k^*+j+1) s_{k^*+j+1},
\end{split}
\]
which implies that
\[
\begin{split}
s_{k^*+j+1} \leq &\frac{\rho s_{k^*+j}^d}{k^*+j+1} \\
\leq &\frac{\rho \hat{s}_{k^*+j}^d}{k^*+j+1} \\
=& \frac{\rho}{k^*+j+1} \big(\frac{\rho}{k^*+1}\big)^{\frac{d^{j+1} - d}{d - 1}} \\
\leq & \big(\frac{\rho}{k^*+1}\big)^{\frac{d^{j+1} - 1}{d - 1}} = \hat{s}_{k^*+j+1},
\end{split}
\]
where the second inequality is due the inductive assumption and the last
inequality holds because $\frac{\rho}{k^*+j+1} \leq \frac{\rho}{k^*+1}$.
This completes the proof.
\end{proof}

Note that $\hat{s}_i \downarrow 0$ for all $i \geq k^* + 2$ if
$d \to \infty$.
In addition, we have $\sum_{i = 0}^{\infty} i p^*_i =
\sum_{i = 0}^{\infty} i (s^*_i - s^*_{i+1}) = \rho$, which then yields
an exact formula for the fixed point in case $d \to \infty$ as $n\rightarrow \infty$. In particular when $d = n$ in the pre-limit system, which corresponds to the flow-level JSQ policy, we have the next corollary.
\begin{corollary}\label{col:JSQ}
The unique fixed point under the flow-level JSQ policy is given by
\begin{equation}
p_i^* =
\begin{cases}
0, & 0 \leq i \leq k^* - 1 \cr
k^* + 1 - \rho, & i = k^* \cr
\rho - k^*, & i = k^* + 1 \cr
0, & i \geq k^* + 2 \cr
\end{cases}
\label{fpjsq2}
\end{equation}
\end{corollary}
Plugging \eqref{fpjsq2} into~(\ref{perf2}) gives the packet-level performance under the flow-level JSQ policy. In particular, taking \eqref{fpjsq2} into~(\ref{perf3}) yields the $\chi$-delay tail probability of the flow-level JSQ policy:
\[
\small
\widetilde{G}_\chi^{~\text{flow-JSQ}} = \frac{(k^*+1-\rho)}{\rho}G_\chi(k^*)+\Big[1-\frac{(k^*+1-\rho)}{\rho}\Big]G_\chi(k^* + 1),
\]
which is a convex combination of $G_\chi(k^*)$ and $G_\chi(k^* + 1)$. Optimizing  over $k^* \in [0, \rho]$, we can further deduce that $\widetilde{G}_\chi^{~\text{flow-JSQ}} \ge G_\chi(\rho)$.\\

\noindent \textbf{Price of Perfect Flow Stickiness.}
Note that the fixed point in~(\ref{fpjsq2}) is the most balanced flow distribution through flow-level load balancing if  strict flow stickiness is maintained. Thus, $\widetilde{G}_\chi^{~\text{flow-JSQ}}$ is the best $\chi$-delay tail probability that can be achieved by any flow-level load balancing scheme that preserves perfect stickiness. 

%\noindent We deduce that a significant improvement in the $\chi$-delay tail
%probability requires that only a negligible fraction of the servers
%have $k^* + 1$ flows, i.e., at most $k$*~flows can be assigned
%to all but a negligible fraction of the servers.
%Since the mean number of active flows per server is~$\rho$,
%this would imply that at least a fraction $1 - k^* / \rho$ of the
%flows must be discarded, in the absence of any advance knowledge
%of flow durations.
On the other hand, in the absence of any stickiness requirement, one could implement
load balancing at the packet level (leaving aside any practical
feasibility constraints).
Consider the simplest randomized packet-level load balancing policy, where a packet is routed to any of the~$n$ servers with equal probability. Under this policy and those packet-level assumptions adopted for \eqref{poiss-pkt}, the queues
at the various servers behave as independent M/M/1 systems
with arrival rate $\rho \nu$ and service rate~$\mu$, so that the $\chi$-delay tail probability is
\[
\widetilde{G}_\chi^{~\text{pkt-rand}}
=
\ee^{- \chi (1 - \frac{\rho \nu}{\mu})} = G_{\chi}(\rho)\le \widetilde{G}_\chi^{~\text{flow-JSQ}}.
\]
Thus, even the simplest packet-level load balancing policy
outperforms the best flow-level load balancing policy that maintains perfect
stickiness, indicating that perfect stickiness carries
a significant performance penalty in terms of packet-level delay.\\

\subsubsection*{Scheme (II): pull-based flow assignment}

This scheme involves two threshold values, a low threshold~$l$
and a high threshold $h > l$.
When the number of flows at a server reaches the threshold~$h$,
a disinvite message is sent from the server to the dispatcher;
the disinvite message is revoked as soon as the number of flows
drops below the level~$h$ again.
When the number of flows at a server falls below the threshold~$l$,
an invite message is sent from the server to the dispatcher;
the invite message is retracted as soon as the number of flows
at the server reaches the level~$l$ again.

When a new flow arrives, the dispatcher assigns it to an arbitrary server
with an outstanding invite message, if any.
Otherwise, the flow is assigned to an arbitrary server without
an outstanding disinvite message, if any.
If all servers have outstanding disinvite messages,
then the flow is assigned to a randomly selected server.
Note that this scheme subsumes various existing load balancing
schemes as special cases.
For example, this scheme reduces to random flow assignment if $l = 0$
and $h = \infty$, and corresponds to the flow-level
``Join-the-Idle-Queue'' (JIQ) policy in \cite{LXKGLG11,Stolyar15}
for $l = 1$ and $h = \infty$.

In order to determine the probabilities $q_{i-1}(s)$ in~(\ref{fp1}),
we need to distinguish three cases. \\

\noindent $\bullet$ \textbf{Case~(i):} $s_l < 1$.

Then
\[
q_{i-1}(s) = \frac{s_{i-1} - s_i}{1 - s_l},
\hspace*{.4in} i = 1, \dots, l,
\]
and $q_{i-1}(s) = 0$ for all $i \geq l + 1$.\\

\noindent $\bullet$ \textbf{Case~(ii):} $s_l = 1$, but $s_h < 1$.

In this case, flows at servers with a total of $l$~flows complete
at rate $l (1 - s_{l+1}) / \beta$, generating invite messages
at that rate.
These invite messages will be used before any arriving flow can be
assigned to a server with $l$ or more flows.

We need to distinguish two sub-cases, depending on whether
$\lambda$ is larger than $(1 - s_{l+1}) / \beta$ or not.
If $\lambda \leq (1 - s_{l+1}) / \beta$, i.e.,
$\rho \leq l (1 - s_{l+1})$, then $q_{l-1}(s) = 1$
and $q_{i-1}(s) = 0$ for all $i \neq l$.
If $\rho > l (1 - s_{l+1})$, then
\[
\lambda q_{l-1}(s) = l (1 - s_{l+1}) / \beta, \mbox{ i.e., }
\rho q_{l-1}(s) = l (1 - s_{l+1}),
\]
and
\[
\lambda q_{i-1}(s) = \tilde\lambda \tilde{q}_{i-1}(s),
\hspace*{.4in} i = l + 1, \dots, h,
\]
with $\tilde\lambda = \lambda - l (1 - s_{l+1}) / \beta$,
\[
\tilde{q}_{i-1}(s) = \frac{s_{i-1} - s_i}{1 - s_h},
\hspace*{.4in} i = l + 1, \dots, h,
\]
and $\tilde{q}_{i-1}(s) = 0$ for all $i = 1, \dots, l - 1$ and $i \geq h + 1$.

In particular, if $s_{l+1} = 1$, but $s_h < 1$, then
\[
q_{i-1}(s) = \frac{s_{i-1} - s_i}{1 - s_h},
\hspace*{.4in} i = l + 1, \dots, h,
\]
and $q_{i-1}(s) = 0$ for all $i = 1, \dots, l$ and $i \geq h + 1$.\\

\noindent $\bullet$ \textbf{Case~(iii):} $s_h = 1$.

In this case, flows at servers with a total of $h$~flows complete
at rate $h (1 - s_{h+1}) / \beta$.
Arriving flows will be dispatched to these servers before any flow
can be assigned to a server with $h$ or more flows
(and hence an outstanding disinvite message).

We need to distinguish two sub-cases, depending on whether
$\lambda$ is larger than $h (1 - s_{h+1}) / \beta$ or not.
If $\lambda \leq h (1 - s_{h+1}) / \beta$, i.e.,
$\rho \leq h (1 - s_{h+1})$, then $q_{h-1}(s) = 1$
and $q_{i-1}(s) = 0$ for all $i \neq h$.
If $\rho > h (1 - s_{h+1})$, then
\[
\lambda q_{h-1}(s) =  h (1 - s_{h+1}) / \beta, \mbox{ i.e., }
\rho q_{h-1}(s) = h (1 - s_{h+1}),
\]
and
\[
\lambda q_{i-1}(s) = \tilde\lambda \tilde{q}_{i-1}(s),
\]
with $\tilde\lambda = \lambda - h (1 - s_{h+1}) / \beta$,
%\lambda \facone - h \factwo (1 - s_{h+1})
\[
\tilde{q}_{i-1}(s) = s_{i-1} - s_i,
\hspace*{.4in} i \geq h + 1,
\]
and $\tilde{q}_{i-1}(s) = 0$ for all $i \leq h$.

In particular, if $s_{h+1} = 1$, then
\[
q_{i-1}(s) = s_{i-1} - s_i,
\hspace*{.4in} i \geq h + 1,
\]
and $q_{i-1}(s) = 0$ for all $i \leq h$.\\

Having determined the expressions for $q_{i-1}(s)$, we can solve the
fixed point(s) of Equations (\ref{fp1}) explicitly.

\begin{theorem}\label{thm:pull}
Assume $l \leq \rho < h$.
The unique fixed point under pull-based flow assignment is given by
\begin{equation}\label{eq:pull-perfect}
p_i^*  =
\begin{cases}
0, & i<l~\text{or}~i>h \cr
\frac{f_{\sigma}(i)}{F_{\sigma}(h)-F_{\sigma}(l-1)}, & i = l, \dots, h,\cr
\end{cases}
\end{equation}
where $f_{\sigma}(\cdot)$ and $F_{\sigma}(\cdot)$ denote the PMF
and CDF of a Poisson random variable with parameter~$\sigma$,
respectively, and $\sigma$ is the unique root of the equation
\[
\sigma [1 - p^*_h(\sigma)]+lp^*_l(\sigma) = \rho.
\]
\end{theorem}

\begin{proof}
See Appendix \ref{ap:II}.
\end{proof}

The fixed point \eqref{eq:pull-perfect} shows that the fraction
of servers with less than $l$~flows or more than $h$~flows is
negligible in the mean-field limit.
In other words, through the invite/disinvite messages,
the variation in the number of flows at each of the servers is
effectively constrained to the range $\{l, \dots, h\}$.

In particular, if $l \leq \rho < h = l + 1$, so that $l = k^*$
and $h = k^* + 1$, then $p^*_l = h - \rho = k^* + 1 - \rho$
and $p^*_h = \rho - l = \rho - k^*$, implying that the fixed point
coincides with that in~(\ref{fpjsq2}) for the flow-level JSQ scheme.
Of course, setting the parameters~$l$ and~$h$ to $\lfloor \rho \rfloor$
and $\lfloor \rho \rfloor + 1$, respectively, requires knowledge
of the value of~$\rho$, which may be difficult to obtain.
When $l$ is higher than $\lfloor \rho \rfloor$, or $h$ is lower
than $\lfloor \rho \rfloor + 1$, the invite/disinvite messages lose
their effectiveness, yielding an unbalanced flow population,
and occasional overload at individual servers.
It may thus be advantageous to set $l$ lower and $h$ somewhat higher
to reduce the risk of overload in case the estimate for~$\rho$ is
inaccurate, at the expense of variation in the number of flows at each
of the servers across a somewhat greater range.

%Note that (*) is the most balanced flow distribution that we can obtain by using flow-level load balancing, where each server roughly has $\lambda\beta$ active flows (in steady state). At the packet-level, such a flow distribution yields almost the same performance as packet-level randomized load balancing. However, it is well known that randomized packet assignment could have a bad performance in practice. For example, the packet-level JSQ($d$) or ``Joint-the-Idle-Queue" policy has a much better latency performance. In other words, the ``best" flow-level load balancing scheme performs similarly to the ``naivest"  packet-level load balancing scheme, and such a performance degradation is due to flow stickiness.

\section{Flow-level Load Balancing with Stickiness Violation}
\label{flow-level-stickiness-violation}

In the previous section we considered two flow-level load balancing
schemes that guarantee perfect stickiness.
In this section we turn to the following three schemes which sacrifice
some flow stickiness for improvement of the packet-level delay performance:
(III) random flow assignment with load shedding;
(IV) random flow assignment with threshold-based flow transfer
to a server with an invite message;
(V) random flow assignment with threshold-based transfer to the
least-loaded server.\\

\subsubsection*{Scheme (III): random flow assignment with load shedding}

In this scheme, an arriving flow is assigned to a randomly selected server.
However, if the selected server already has $h$~flows,
the flow is immediately terminated and discarded.

The numbers of flows at the various servers are then independent,
and the number of flows at each individual server evolves as the number
of jobs in an Erlang loss system with load~$\rho$ and capacity~$h$.
%assuming no flows are reassigned ($\alpha = 0)$.
Thus the number of active flows at a server in stationarity follows
a Poisson distribution with parameter~$\rho$ truncated at level~$h$.
%\[
%\pi(Q_1, \dots, Q_n) = \prod\limits_{k = 1}^{n} \pi_k(Q_k),
%\]
%with
%\[
%p_i = \frac{\frac{\rho^{i}}{i!}}{\sum_{j = 0}^{h} \frac{\rho^j}{j!}} =
%\frac{f_{\rho}(i)}{F_\rho(h)},~\forall i \in \mathbb{N},
%\]
%where $f_{\rho}(\cdot)$ and $F_{\rho}(\cdot)$ denote the PMF and CDF
%of a Poisson random variable with parameter~$\rho$, respectively.
%\[
%\pi_k(0) = \left[\sum_{i = 0}^{h} \frac{\rho^i}{i!}\right]^{- 1},
%\]
In particular, the stickiness violation probability is the blocking probability, i.e.,
$\epsilon_h = p_h^* = f_\rho(h) \slash F_\rho(h)$,
where $f_{\rho}(\cdot)$ and $F_{\rho}(\cdot)$ denote the PMF and CDF
of a Poisson random variable with parameter~$\rho$.
%which corresponds to the blocking probability in the above-mentioned Erlang loss system.

%Adopting the Jagerman approximation for the latter
%probability~\cite{Jagerman74}, we obtain
%\[
%\pi(h) \approx \frac{\phi\left(\frac{h - \rho}{\sqrt{h}}\right)}
%{\sqrt{h} \Phi\left(\frac{h - \rho}{\sqrt{h}}\right)},
%\]
%with $\phi(\cdot)$ and $\Phi(\cdot)$ denoting the standard normal
%density and cumulative distribution function, respectively.
%The Jagerman approximation is asymptotically exact, in the sense that
%if $h$ grows large and $(h - \rho) / \sqrt{h} \to \gamma$
%as $h \to \infty$, then
%\[
%\lim_{h \to \infty} \pi(h) \sqrt{h} = \frac{\phi(\gamma)}{\Phi(\gamma)}.
%\]

%If $\rho / h \to \theta$ as $h \to \infty$,
%\[
%\lim_{h \to \infty} \pi(h) = \max\{\theta - 1, 0\}.
%\]
%(We can further specify the decay rate in case $\theta < 1$.)

We now examine the $\chi$-delay tail probability by applying~\eqref{perf3}.
In case $h < \mu / \nu$, it can be shown that
\[
\widetilde{G}_\chi^h =
%\frac{\sum_{i = 0}^{h} i p_i F(i)}{\sum_{i = 0}^h i p_i}
\frac{\sum_{i = 0}^{h} i \frac{\rho^i}{i!} \ee^{- \chi (1 - i \nu / \mu)}}
{\sum_{i = 0}^{h} i \frac{\rho^i}{i!}} =
\frac{\ee^{- \chi (1 - \nu / \mu)} \sum_{i = 0}^{h - 1}
\frac{(\rho \ee^{\chi \nu / \mu})^i}{i!}}{\sum_{i = 0}^{h - 1} \frac{\rho^i}{i!}}.
\]
Defining $a_\chi \triangleq \rho \ee^{\chi \nu / \mu}$
and $b_\chi \triangleq \exp\{- \chi (1-\nu / \mu) + a_\chi - \rho\}$,
we can rewrite $\widetilde{G}_\chi^h$ as
\begin{equation}
\label{eq:shedding-h-small}
\widetilde{G}_\chi^h = \frac{b_\chi F_{a_\chi}(h-1)}{F_\rho(h-1)}.
\end{equation}
In case $h \geq \mu / \nu$, it can be similarly shown that
\begin{small}
\[
\widetilde{G}_\chi^h=\frac{b_\chi F_{a_\chi}\big(\lfloor\mu\slash \nu\rfloor-1\big)+F_\rho(h-1)-F_\rho\big(\lfloor\mu\slash \nu\rfloor-1\big)}{F_\rho(h-1)}.
\]
\end{small}

\noindent In particular, when $h = \infty$ (i.e., when perfect stickiness is
required), we obtain
\begin{equation}
\label{eq:shedding-h-infty}
\widetilde{G}_\chi^\infty =
b_\chi F_{a_\chi}\big(\lfloor \mu \slash \nu \rfloor - 1 \big) + 1 -
F_\rho\big(\lfloor \mu \slash \nu \rfloor - 1\big).
\end{equation}
Thus, for any $h < \mu / \nu$, a stickiness violation probability
$\epsilon_h = f_\rho(h) \slash F_\rho(h)$ improves the $\chi$-delay tail
probability by a factor of $\widetilde{G}_\chi^\infty \slash \widetilde{G}_\chi^h$,
where $\widetilde{G}_\chi^h$ and $\widetilde{G}_\chi^\infty$ are given
in~\eqref{eq:shedding-h-small} and~\eqref{eq:shedding-h-infty},
respectively.
This trade-off will be numerically examined
in Section~\ref{sec:simulation-flow-level}.
It is observed that relaxing the strict stickiness requirement by even
a minimal amount can significantly improve the packet-level latency. \\

\subsubsection*{Scheme (IV): random flow assignment with threshold-based flow transfer to a server with an invite message}

In this scheme, an arriving flow is assigned to a randomly selected server.
However, if the selected server already has $h$~flows,
then the flow (or a randomly selected flow associated with that server) is instantly diverted to an arbitrary server
with less than $l$~flows, if possible.
Otherwise, the flow is instantly redirected to an arbitrary server
with less than $h$~flows, if possible, or entirely discarded otherwise.

In order to determine the probabilities $q_{i-1}(s)$ in~(\ref{fp1}),
we need to distinguish three cases.\\

\noindent $\bullet$ \textbf{Case~(i):} $s_l < 1$.

In this case, we have
\[
\begin{split}
q_{i-1}(s) &= s_{i-1} - s_i + \frac{s_{i-1} - s_i}{1 - s_l} s_h \\
&=\frac{(s_{i-1} - s_i) (1 - s_l + s_h)}{1 - s_l},~i = 1, \dots, l,
\end{split}
\]
\[
q_{i-1}(s) = s_{i-1} - s_i,
~i = l + 1, \dots, h,
\]
and $q_{i-1}(s) = 0$ for all $i \geq h + 1$.\\

\noindent $\bullet$ \textbf{Case~(ii):} $s_l = 1$, but $s_h < 1$.

In this case, flows at servers with a total of $l$~flows complete
at rate $l (s_l - s_{l+1}) / \beta$.
Flows that are initiated at servers that have already $h$~flows,
will be dispatched to these servers before any flow can be assigned
to a server with $l$ or more flows.

We need to distinguish two sub-cases, depending on whether
$\lambda s_h$ is larger than $l (1 - s_{l+1}) / \beta$ or not.
If $\lambda s_h \leq l (1 - s_{l+1}) / \beta$, i.e.,
$\rho s_h \leq l (1 - s_{l+1})$, then
then $q_{l-1}(s) = s_h$, $q_{i-1}(s) = s_{i-1} - s_i$ for all
$i = l + 1, \dots, h$, and $q_{i-1}(s) = 0$ for all
$i = 1, \dots, l - 1$ and $i \geq h + 1$.
If $\rho s_h > l (s_l - s_{l+1})$, then
\[
\lambda q_{l-1}(s) = l (1 - s_{l+1}) / \beta, \mbox{ i.e.,}
\rho q_{l-1}(s) = l (1 - s_{l+1}),
\]
and
\[
\lambda q_{i-1}(s) =
\lambda (s_{i-1} - s_i) + \tilde\lambda \tilde{q}_{i-1}(s),
~i = l + 1, \dots, h,
\]
with $\tilde\lambda = \lambda s_h - l (1 - s_{l+1}) / \beta$,
\[
\tilde{q}_{i-1}(s) = \frac{s_{i-1} - s_i}{1 - s_h},
~i = l + 1, \dots, h,
\]
and $\tilde{q}_{i-1}(s) = 0$ for all $i = 1, \dots, l - 1$ and $i \geq h + 1$.

In particular, in case $s_{l+1} = 1$, but $s_h < 1$, then
\[
\begin{split}
q_{i-1}(s) &= s_{i-1} - s_i + \frac{s_{i-1} - s_i}{1 - s_h} s_h\\
& =\frac{s_{i-1} - s_i}{1 - s_h},
~i = l + 1, \dots, h,
\end{split}
\]
and $q_{i-1}(s) = 0$ for all $i = 1, \dots, l$ and $i \geq h + 1$.\\

\noindent $\bullet$ \textbf{Case~(iii):} $s_h = 1$.

In this case, flows at servers with a total of $h$~flows complete
at rate $h (1 - s_{h+1}) / \beta$.
Arriving flows will be dispatched to these servers before any flow
can be assigned to a server with $h$ or more flows.

We need to distinguish two sub-cases, depending on whether
$\lambda$ is larger than $h (1 - s_{h+1}) / \beta$ or not.
If $\lambda \leq h (1 - s_{h+1}) / \beta$, i.e.,
$\rho \leq h (1 - s_{h+1})$, then $q_{h-1}(s) = 1$
and $q_{i-1}(s) = 0$ for all $i \neq h$.
If $\rho > h (1 - s_{h+1})$, then
\[
\lambda q_{h-1}(s) = h (1 - s_{h+1}) / \beta, \mbox{ i.e.,}
\rho q_{h-1}(s) = h (1 - s_{h+1}),
\]
and
\[
\lambda q_{i-1}(s) = \tilde\lambda \tilde{q}_{i-1}(s),
~i \geq h + 1,
\]
with $\tilde\lambda = \lambda - h (1 - s_{h+1}) / \beta$,
\[
\tilde{q}_{i-1}(s) = s_{i-1} - s_i,
~i \geq h + 1,
\]
and $\tilde{q}_{i-1}(s) = 0$ for all $i \leq h$.

In particular, in case $s_{h+1} = 1$, then
\[
q_{i-1}(s) = s_{i-1} - s_i,
~i \geq h + 1,
\]
and $q_{i-1}(s) = 0$ for all $i \leq h$.\\

Having determined the expressions for $q_{i-1}(s)$, we can solve the
fixed point(s) of Equations (\ref{fp1}) explicitly.

\begin{theorem}\label{thm:scheme-IV}
Assume $l \leq \rho < h$.
The unique fixed point under Scheme (IV) is given by
\begin{equation}\label{eq:fix-iv}
p_i^*  =
\begin{cases}
0, & i<l~\text{or}~i>h \cr
\frac{f_{\sigma}(i)}{F_{\sigma}(h)-F_{\sigma}(l-1)}, & i = l, \dots, h,\cr
\end{cases}
\end{equation}
where $\sigma$ is the unique root of the equation
\[
\sigma [1 - p^*_h(\sigma)]+lp^*_l(\sigma) = \rho.
\]
\end{theorem}

\begin{proof}
See Appendix \ref{ap:IV}.
\end{proof}

Note that the fixed point \eqref{eq:fix-iv} is identical to the fixed
point in Scheme (II), yet the proof is somewhat different.
The above fixed point shows that the fraction of servers with less than
$l$~flows or more than $h$~flows is negligible in the mean-field limit.
In other words, through the threshold-based flow transfer,
the variation in the number of flows at each of the servers is
effectively limited to the range $\{l, \dots, h\}$.
The trade-off between the stickiness violation probability
and packet-level latency can be analyzed in a similar way as for
Scheme (III) by substituting $p_i^*$ into~\eqref{perf3} and noticing
that the stickiness violation probability is $p_h^*$. \\

\subsubsection*{Scheme (V): random flow assignment with threshold-based transfer to the least-loaded server}

In this scheme, an arriving flow is assigned to a randomly selected
server.
However, if the selected server already has $h$~flows, then the flow
(or a randomly selected flow associated with that server) is
instantly diverted to a server with the minimum number of flows.

%Let $i^* = \min\{i: s_{i + 1} < 1\}$ be the minimum number of flows
%across all servers, in the sense that the fraction of servers
%with less than $i^*$~flows is negligible, while the fraction
%of servers with exactly $i^*$~flows is strictly positive.

Let $m = \min\{i: s_{i + 1} < 1\}$ be the minimum number of flows
across all servers, in the sense that the fraction of servers
with less than $m$~flows is negligible, while the fraction
of servers with exactly $m$~active flows is strictly positive.
Note that $s_m = 1$ by definition.

In order to determine the probabilities $q_{i-1}(s)$ in~(\ref{fp1}),
we need to distinguish two cases.\\

\noindent $\bullet$ \textbf{Case~(i):} $m < h$, and thus $s_h < 1$.

In this case, flows at servers with a total of $m$~flows complete
at rate $m (s_m - s_{m+1}) / \beta$.
Flows that are initiated at servers that have already $h$~flows,
will be dispatched to these servers before any flow can be assigned
to a server with $m$~flows.

We need to distinguish two sub-cases, depending on whether
$\lambda s_h$ is larger than $m (1 - s_{m+1}) / \beta$ or not.
If $\lambda s_h \leq m (1 - s_{m+1}) / \beta$, i.e.,
$\rho s_h \leq m (1 - s_{m+1})$, then $q_{m-1}(s) = s_h$,
$q_{i-1}(s) = s_{i-1} - s_i$ for all $i = m + 1, \dots, h$,
and $q_{i-1}(s) = 0$ for all $i = 1, \dots, m - 1$ and $i \geq h + 1$.
If $\rho s_h > m (1 - s_{m+1})$, then
\[
\rho q_{m-1}(s) = m (1 - s_{m+1}), \mbox{ i.e.,}
\lambda q_{m-1}(s) = m (1 - s_{m+1}) / \beta,
\]
amd
\[
\begin{split}
\lambda q_m(s) &= \lambda (1 - s_{m+1}) + \tilde\lambda \\
&=
\lambda s_h + (\lambda - m / \beta) (1 - s_{m+1}),
\end{split}
\]
with $\tilde\lambda = \lambda s_h - m (1 - s_{m+1}) / \beta$,
\[
q_{i-1}(s) = s_{i-1} - s_i,
\hspace*{.4in} i = m + 2, \dots, h,
\]
and $q_{i-1}(s) = 0$ for all $i = 1, \dots, m - 1$ and $i \geq h + 1$.\\

\noindent $\bullet$ \textbf{Case~(i):} $m \geq h$, and thus $s_h = 1$.

Like in the previous case, flows at servers with a total of $m$~flows
complete at rate $m (1 - s_{m+1})$.
Arriving flows will be dispatched to these servers before any flow
can be assigned to a server with $m$~flows.

We need to distinguish two sub-cases, depending on whether
$\lambda$ is larger than $m (1 - s_{m+1}) / \beta$ or not.
If $\lambda \leq m (1 - s_{m+1}) / \beta$, i.e.,
$\rho \leq m (1 - s_{m+1})$, then $q_{m-1}(s) = 1$,
and $q_{i-1}(s) = 0$ for all $i \neq m$.
If $\rho > m (1 - s_{m+1})$, then
\[
\rho q_{m-1}(s) = m (1 - s_{m+1}), \mbox{ i.e.,}
\lambda q_{m-1}(s) = m (1 - s_{m+1}) / \beta,
\]
and
\[
\lambda q_m(s) = \tilde\lambda = \lambda - m (1 - s_{m+1}) / \beta, 
\]
and $q_{i-1}(s) = 0$ for all $i \neq m, m + 1$.\\

Having determined the expressions for $q_{i-1}(s)$, we can solve the
fixed point(s) of Equations (\ref{fp1}) explicitly.

\begin{theorem}\label{thm:scheme-V}
Assume $\rho < h$. Then the fixed point under Scheme (V) is given by
\begin{equation}\label{eq:fix-v}
p_i^* =
\begin{cases}
0, & i<i^*~\text{or}~i>h \cr
\frac{\rho p^*_h}{\rho - i^*} \left[\frac{h! \rho^{i^* - 1}}{(i^* - 1)! \rho^h}
\frac{\rho}{i^*} - 1\right], &i=i^* \cr
\frac{h!}{\rho^{h - i} i!} p^*_h, & i = i^* + 1, \cdots, h,\cr
\end{cases}
\end{equation}
with
\[
p_h^* = \left[\frac{\rho}{\rho - i^*}
\left[\frac{h! \rho^{i^* - 1}}{(i^* - 1)! \rho^h} \frac{\rho}{i^*} - 1\right] +
\sum_{i = i^* + 1}^{h} \frac{h!}{\rho^{h - i} i!}\right]^{- 1},
\]
and $i^* = \min\{i: \frac{\rho^i}{i!} > \frac{\rho^h}{h!}\}$.
\end{theorem}

\begin{proof}
See Appendix \ref{ap:V}.
\end{proof}

The fixed point \eqref{eq:fix-v} shows that the fraction of servers
with less than $i^*$~flows or more than $h$~flows is negligible
in the mean-field limit.
In other words, through the threshold-based load transfer,
the variation in the number of flows at each of the servers is
effectively limited to the range $\{i^*, \dots, h\}$.
The trade-off between the stickiness violation probability
and packet-level latency can be analyzed in a similar way as for
Scheme (III) by substituting $p_i^*$ into~\eqref{perf3} and noticing
that the stickiness violation probability is~$p_h^*$.

\section{Numerical evaluation}
\label{sec:simulation-flow-level}

In this section we numerically evaluate the performance of the various
flow-level load balancing schemes considered in the previous two sections.

\subsection{Simulation settings}

\begin{figure}[]
%\centering
\subfigure[flow-level JSQ policy]{\label{fig:perfect1-1}\includegraphics[width=42mm,height=36mm]{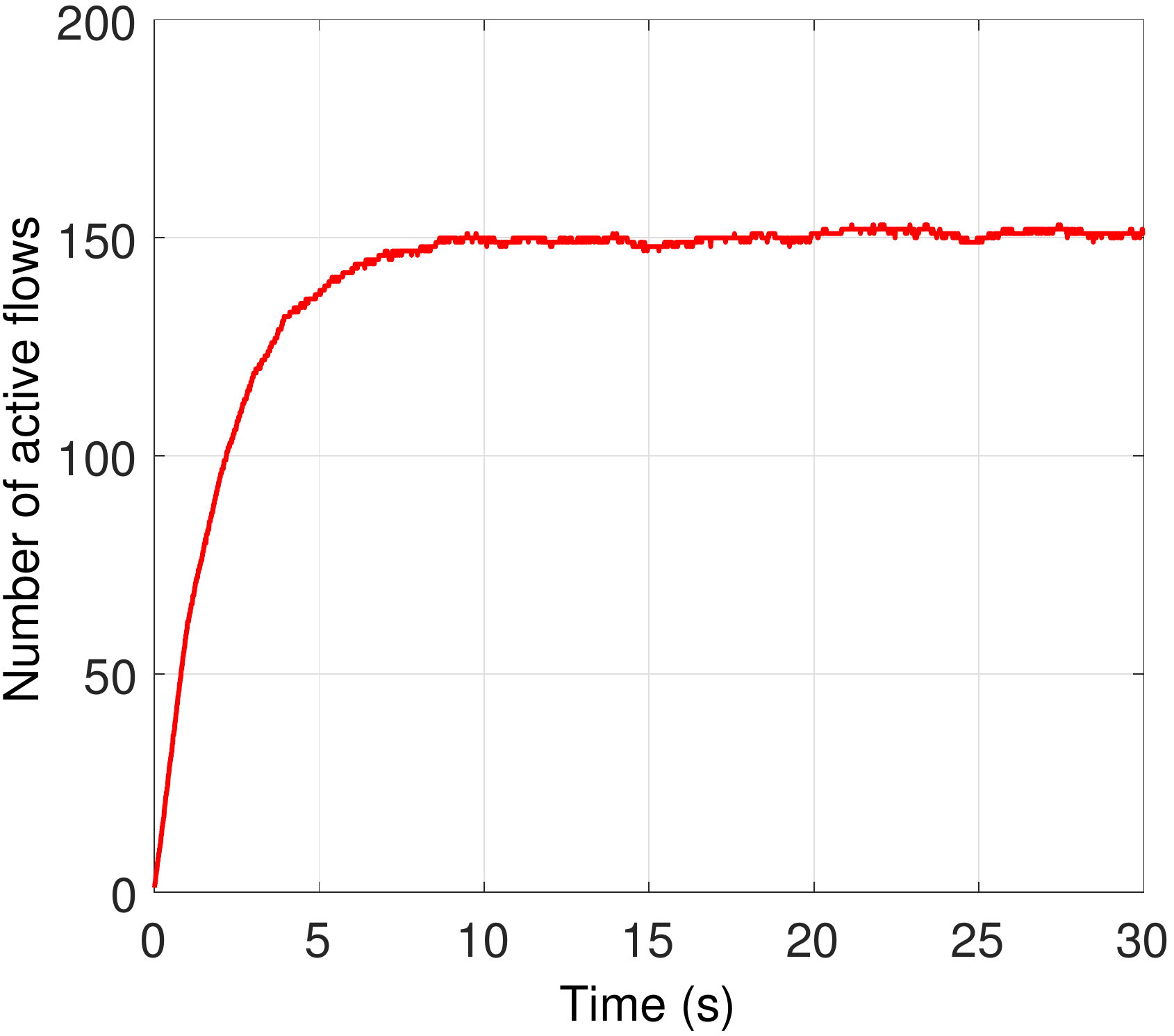}\includegraphics[width=42mm,height=35mm]{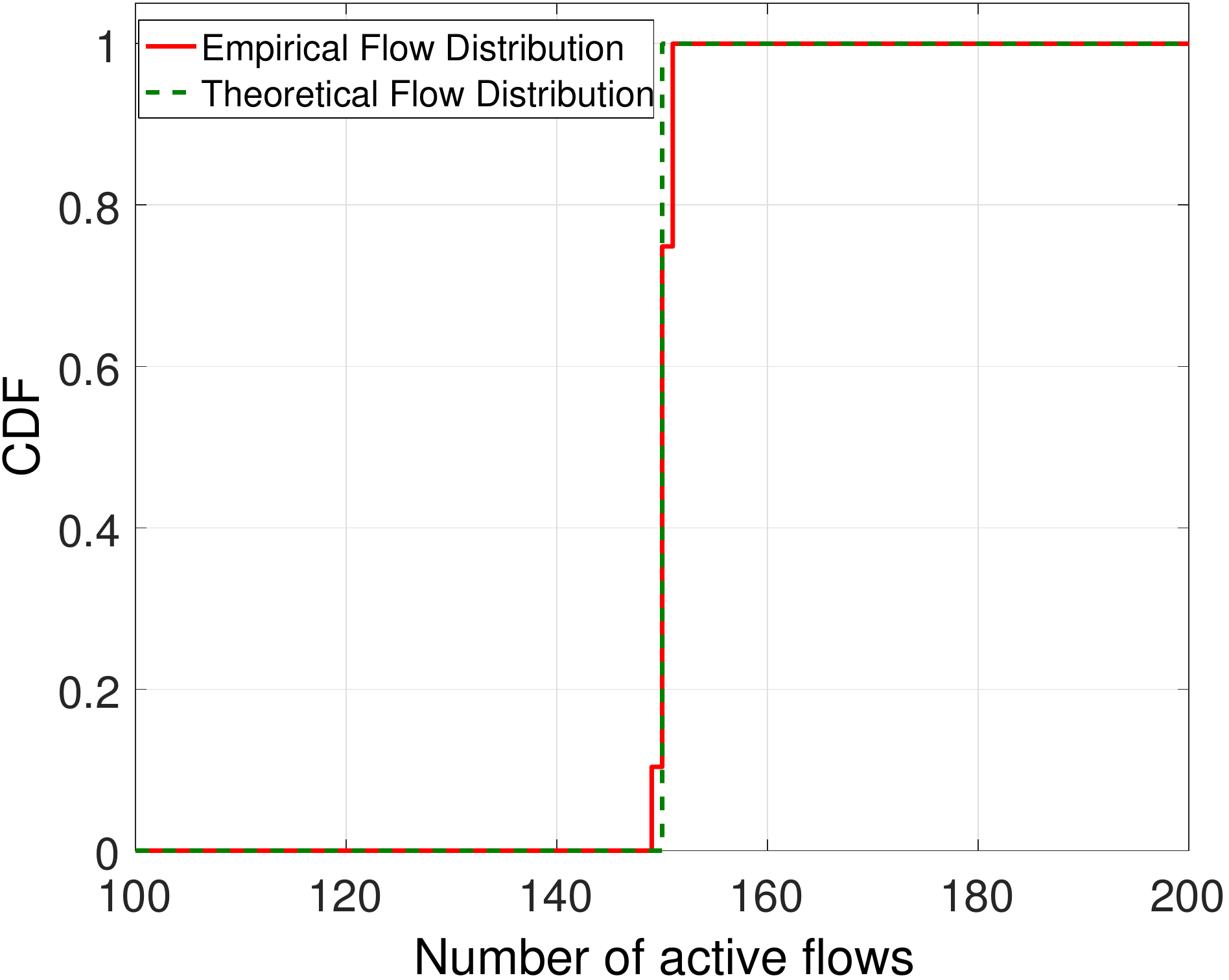}}
\subfigure[flow-level power-of-$d$ policy ($d=2$)]{\label{fig:power-1}\includegraphics[width=42mm,height=36mm]{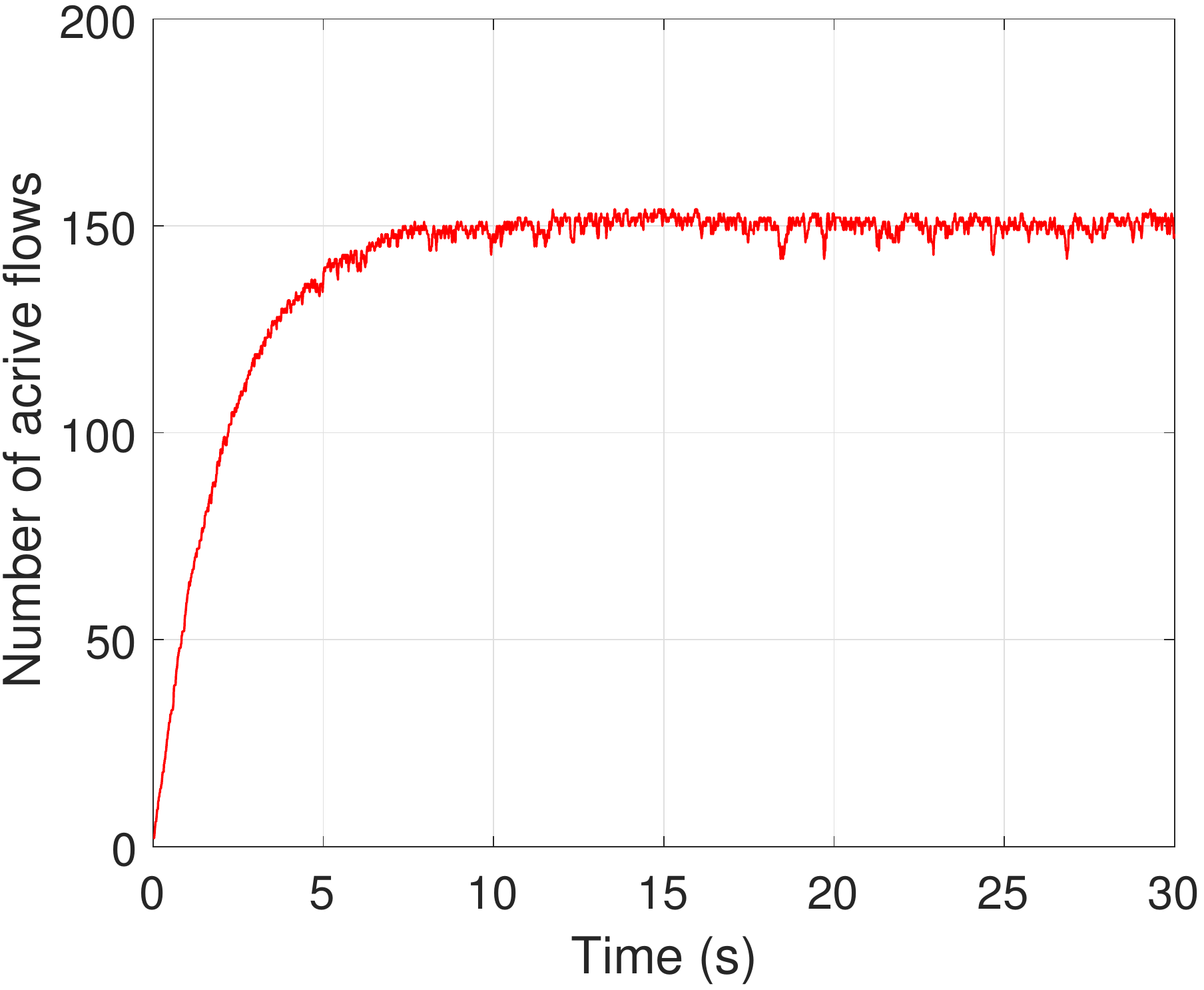}\includegraphics[width=42mm,height=35mm]{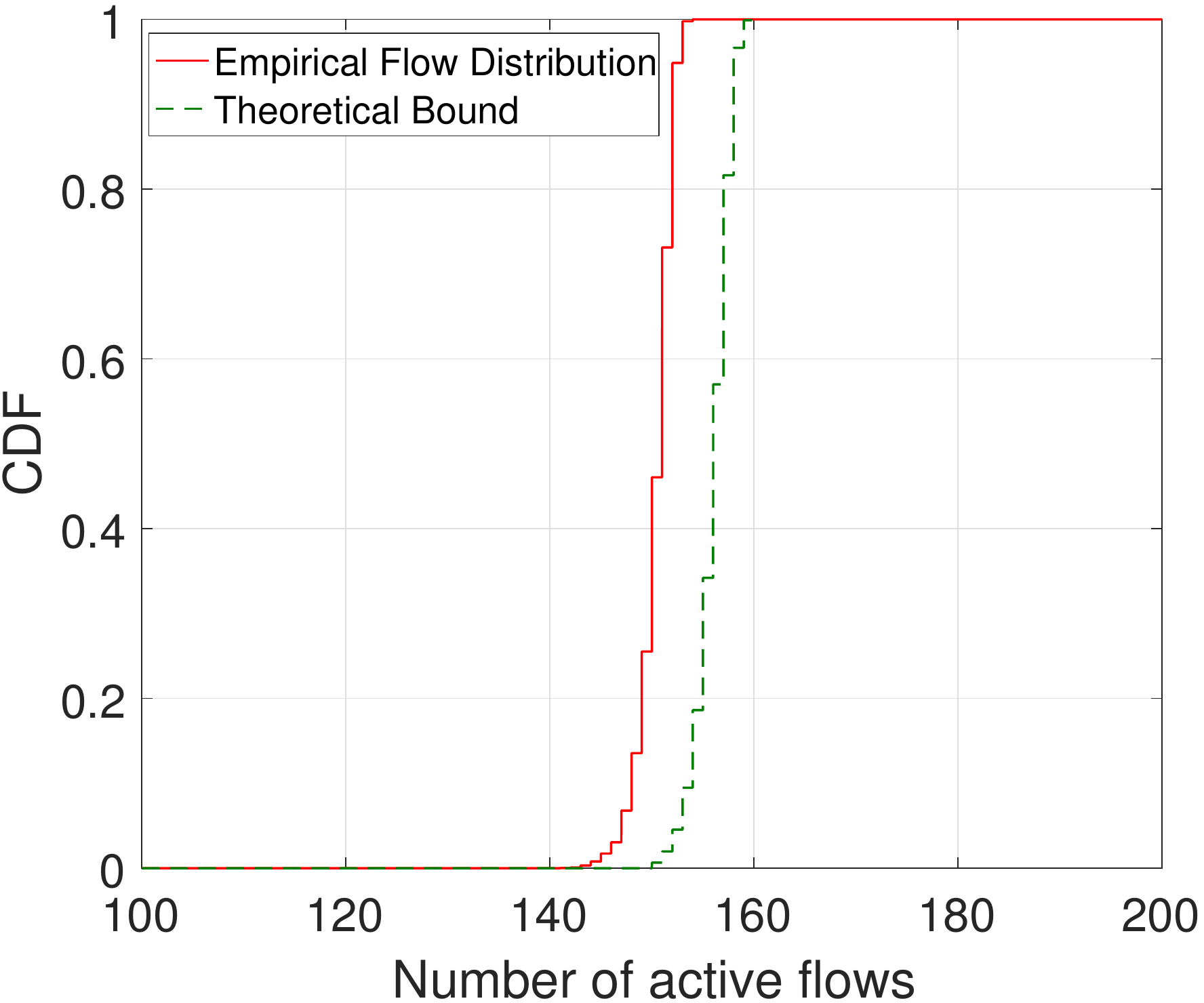}}
\subfigure[$l=0,~h=\infty$ (flow-level randomized load balancing)]{\label{fig:perfect1-2}\includegraphics[width=42mm,height=36mm]{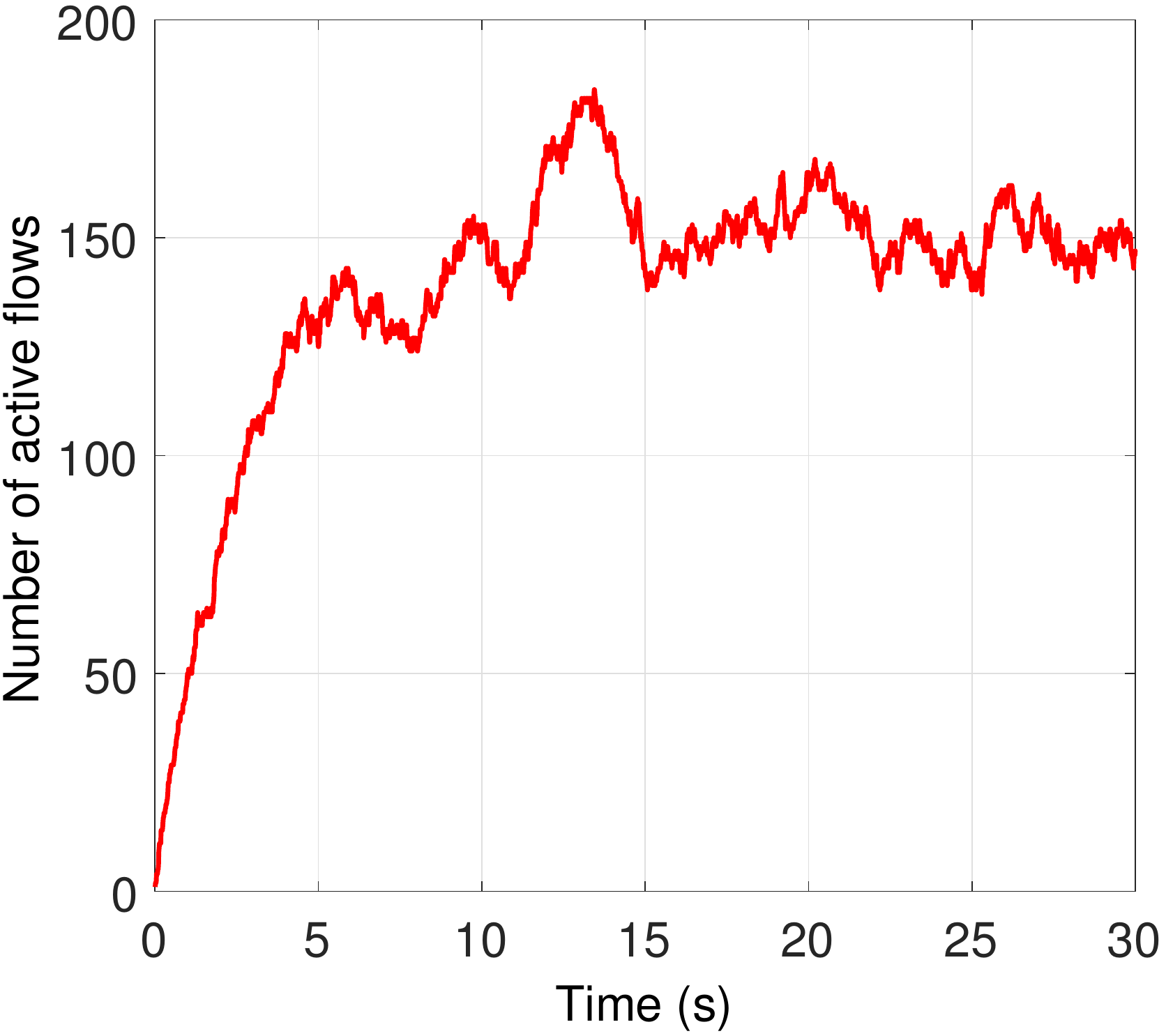}\includegraphics[width=42mm,height=35mm]{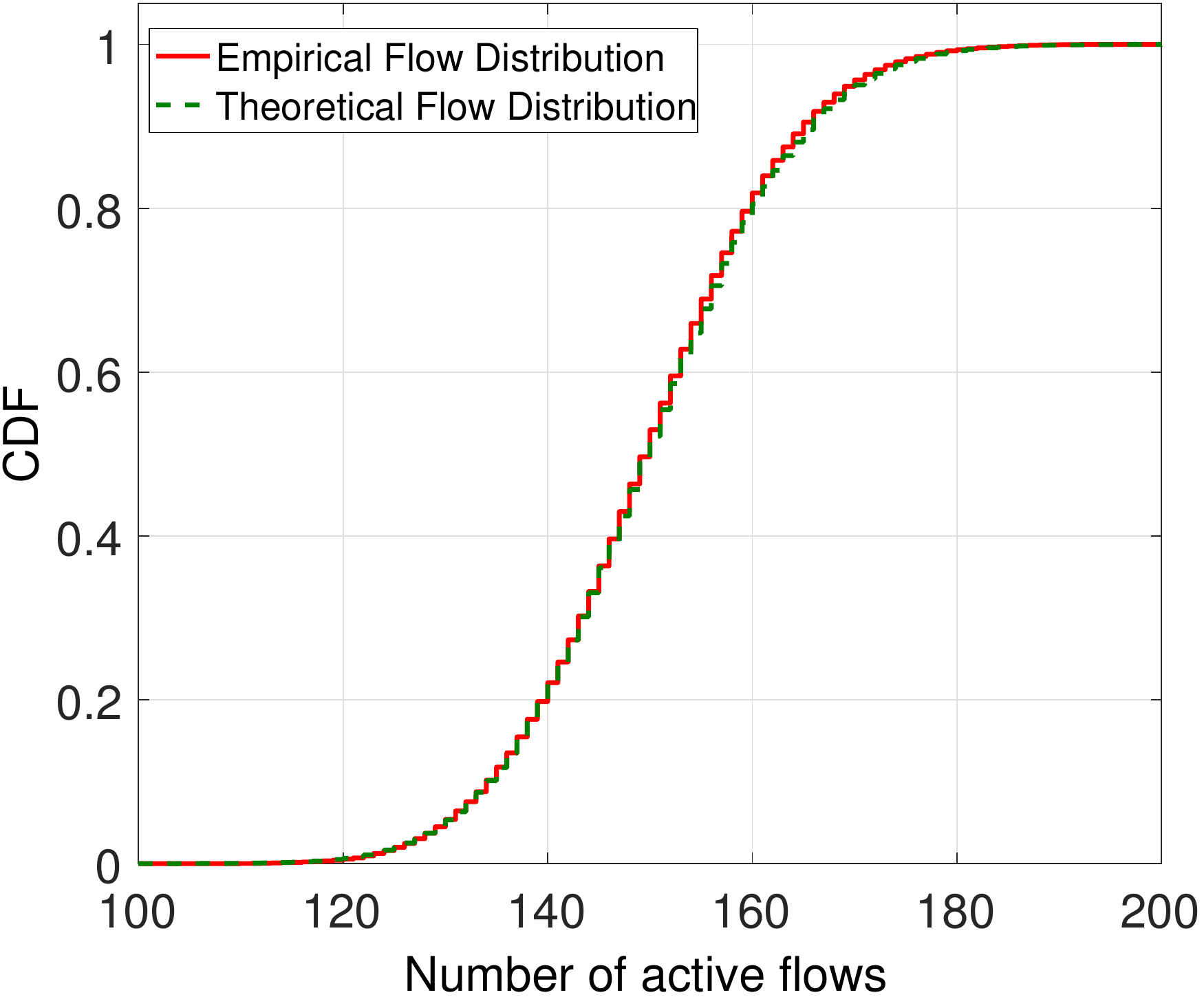}}
\subfigure[$l=k^*=150,~h=k^*+1=151$]{\label{fig:perfect1-3}\includegraphics[width=42mm,height=36mm]{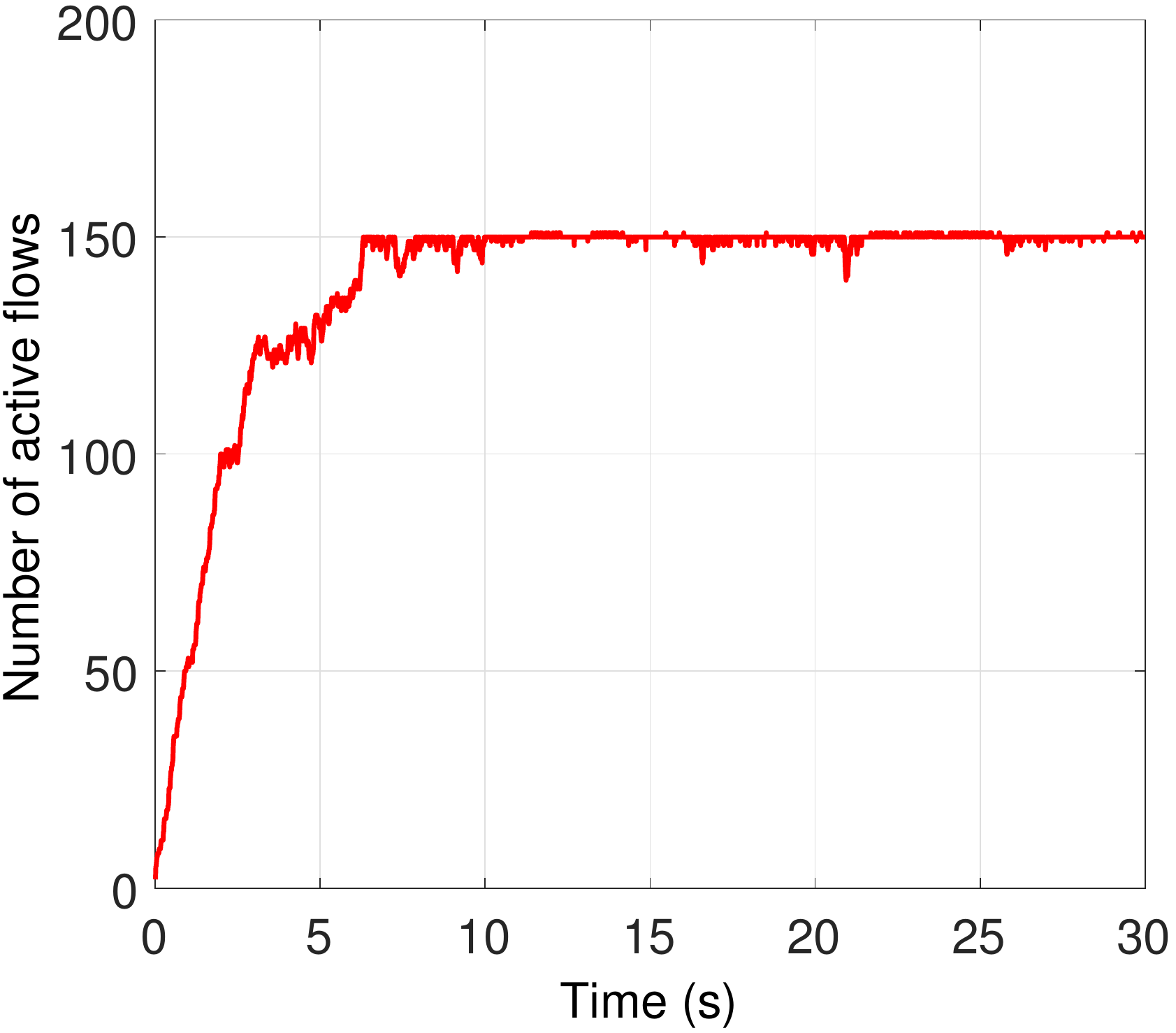}\includegraphics[width=42mm,height=34mm]{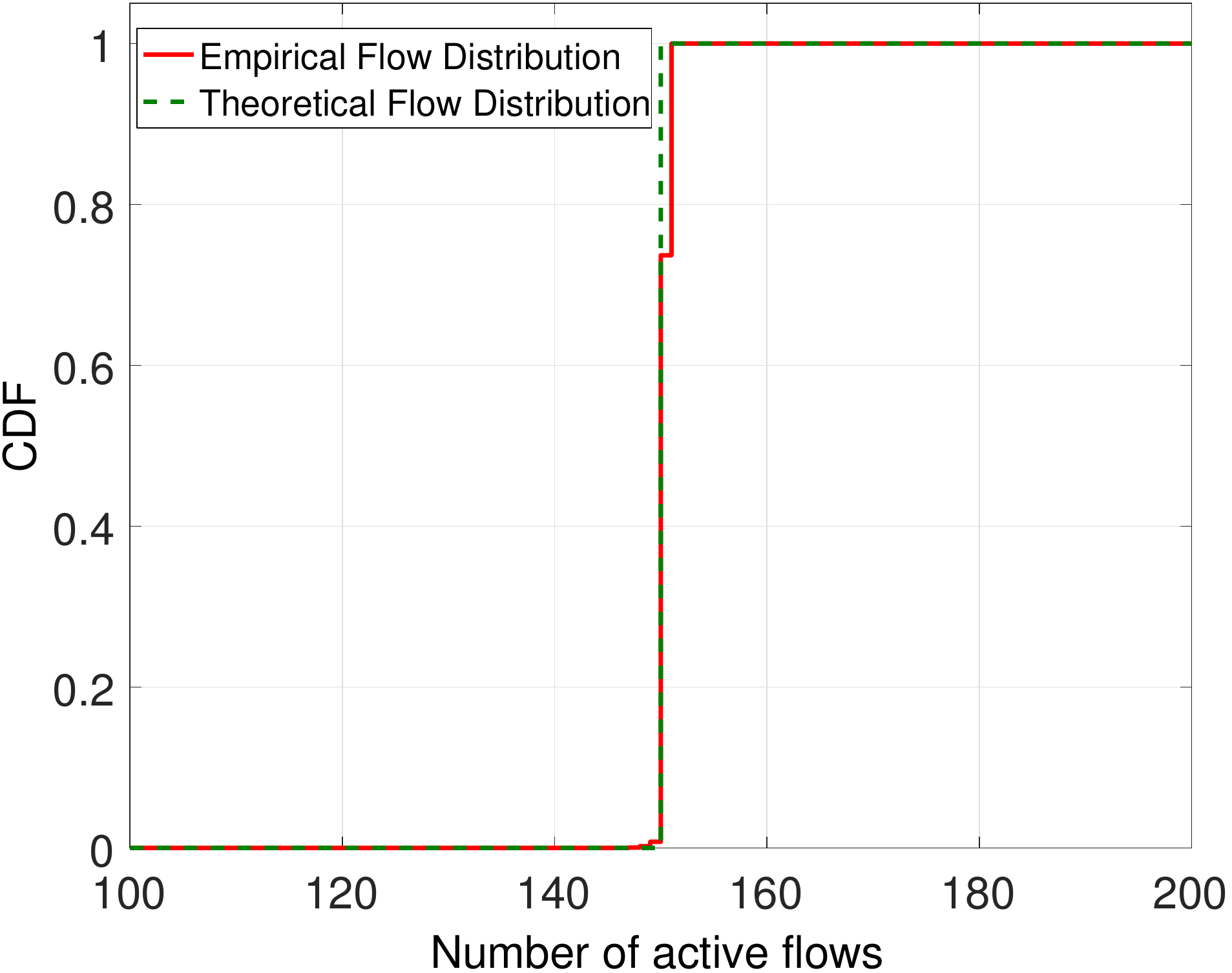}}
\subfigure[$l=k^*-10=140,~h=k^*+10=160$]{\label{fig:perfect1-4}\includegraphics[width=42mm,height=36mm]{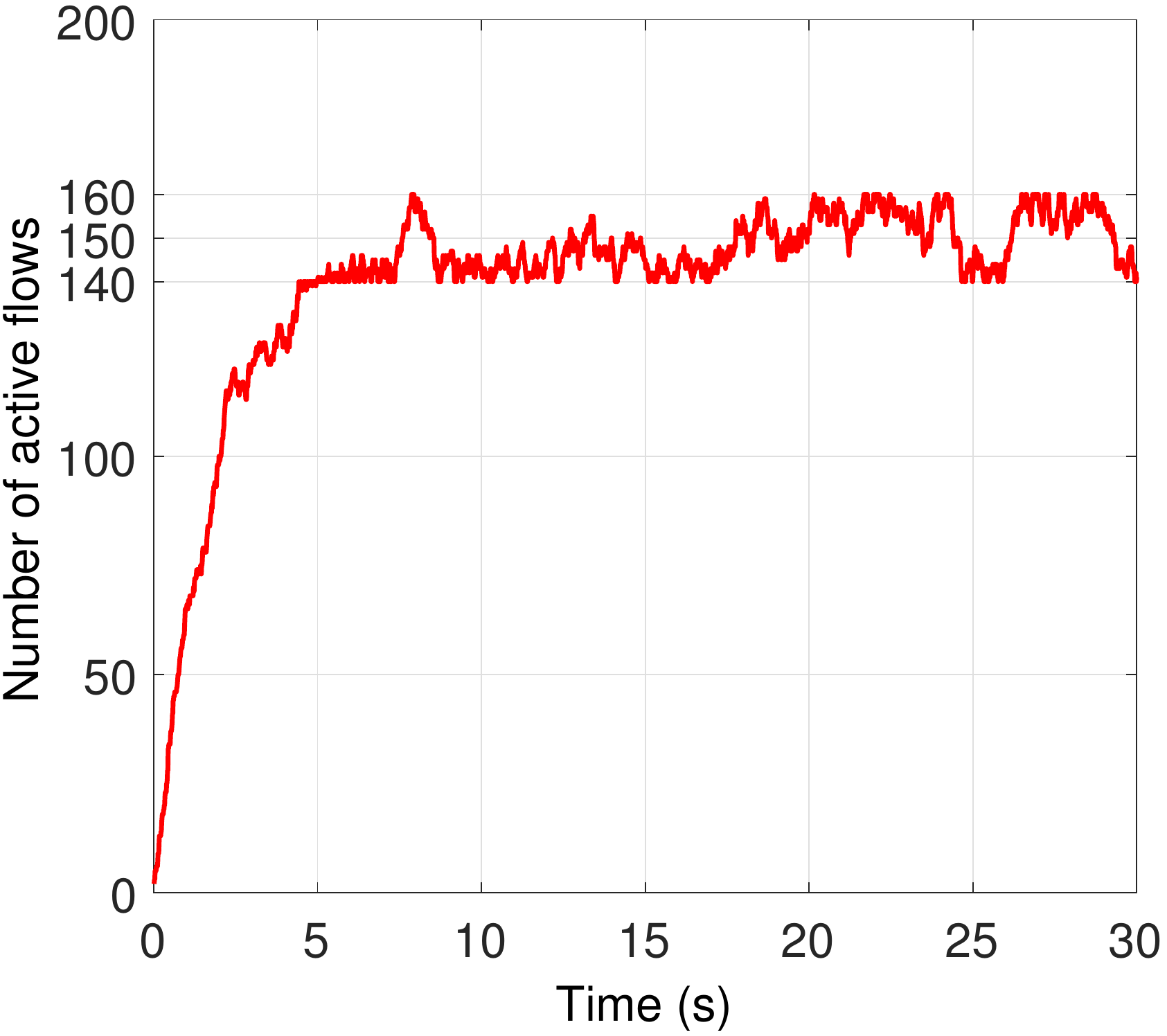}\includegraphics[width=42mm,height=35mm]{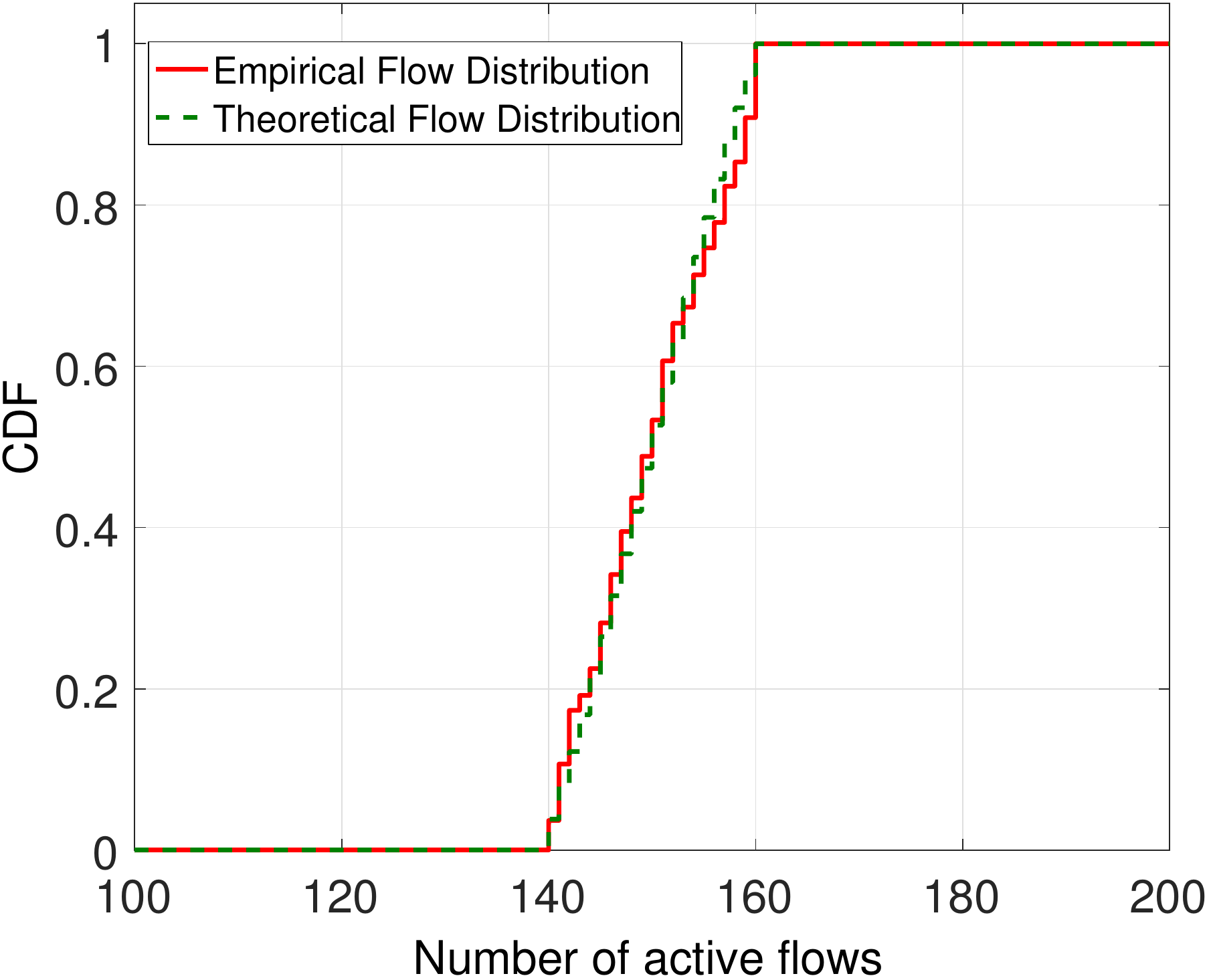}}\vspace{-2mm}
%\subfigure[$l=k^*+10,~h=k^*+15$]{\label{fig:perfect1-5}\includegraphics[width=42mm,height=31mm]{invite1-6}\includegraphics[width=42mm,height=30mm]{invite2-6}}
%\subfigure[$l=k^*-15,~h=k^*-10$]{\label{fig:perfect1-6}\includegraphics[width=42mm,height=31mm]{invite1-7}\includegraphics[width=42mm,height=30mm]{invite2-7}}
\caption{Variation in the number of active flows at a particular server
over time (left) and empirical vs.~theoretical flow distributions (right).
(a): flow-level power-of-$d$ policy;
(b): flow-level JSQ policy
(c)--(e): flow-level pull-based scheme for various threshold values~$l$
and~$h$.}
\label{fig:perfect1}
\end{figure}

\begin{figure}[]
\begin{center}
\includegraphics[width=3in]{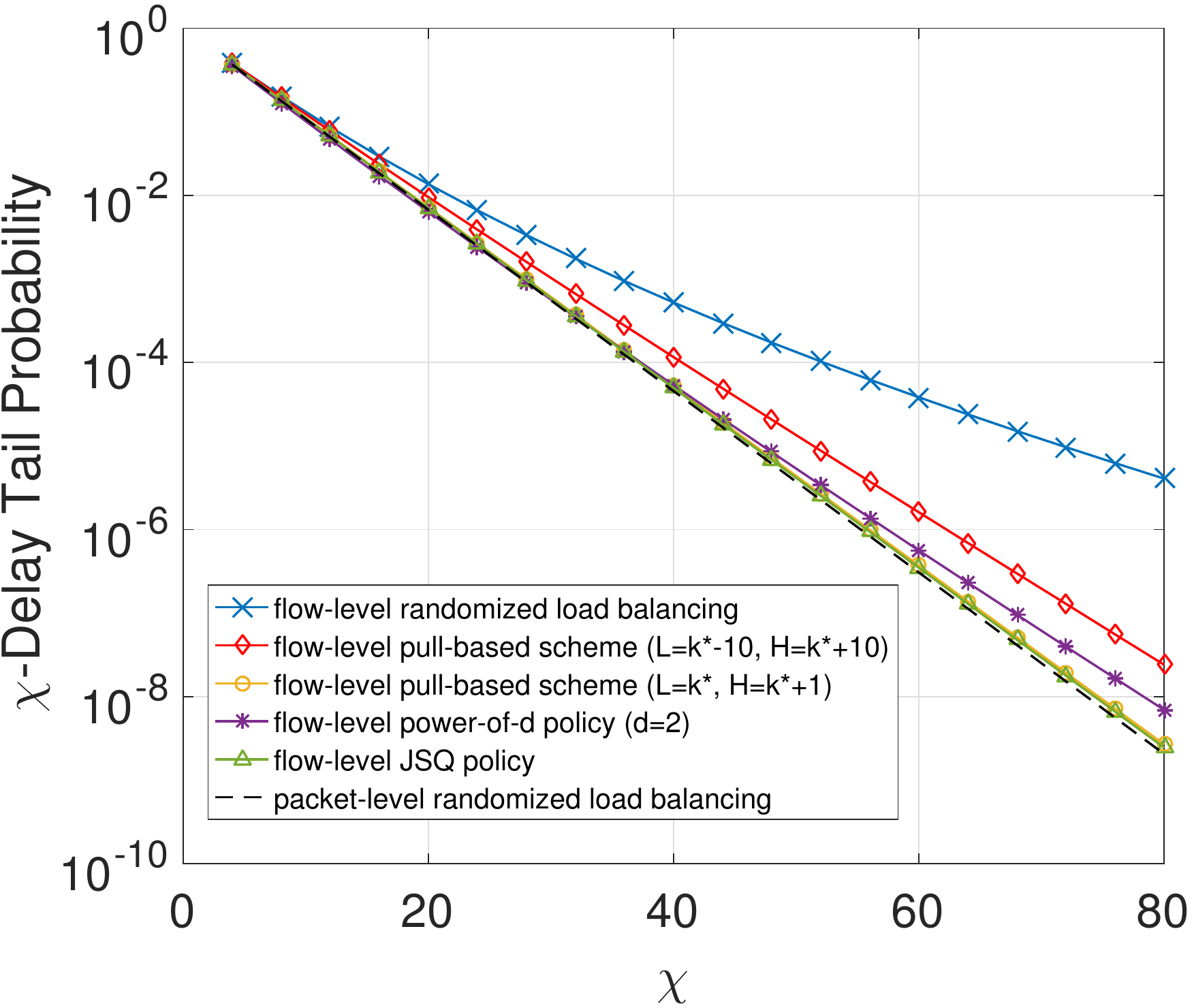}\vspace{-2mm}
\caption{Packet-level $\chi$-delay tail probability under various load
balancing schemes:
(1) Flow-level power-of-$d$ policy (and flow-level JSQ policy);
(2) Flow-level pull-based scheme;
(3) Packet-level randomized load balancing.}
\label{fig:perfect3}
\end{center}
\end{figure}

\begin{figure}[]
%\centering
\subfigure[Scheme (III) with $h=160$]{\label{fig:sacrifice1-1}\includegraphics[width=42mm,height=35mm]{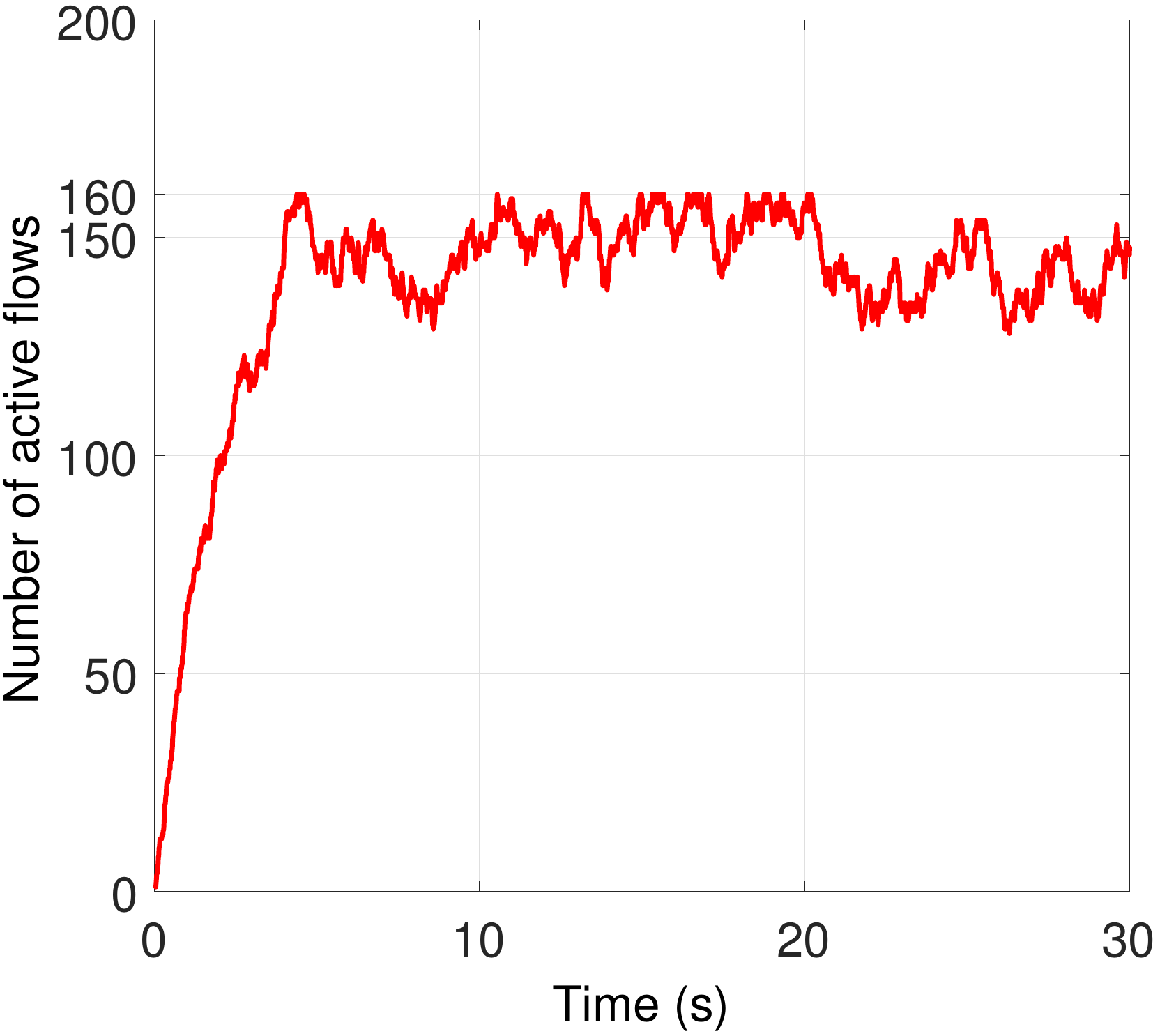}\includegraphics[width=42mm,height=35mm]{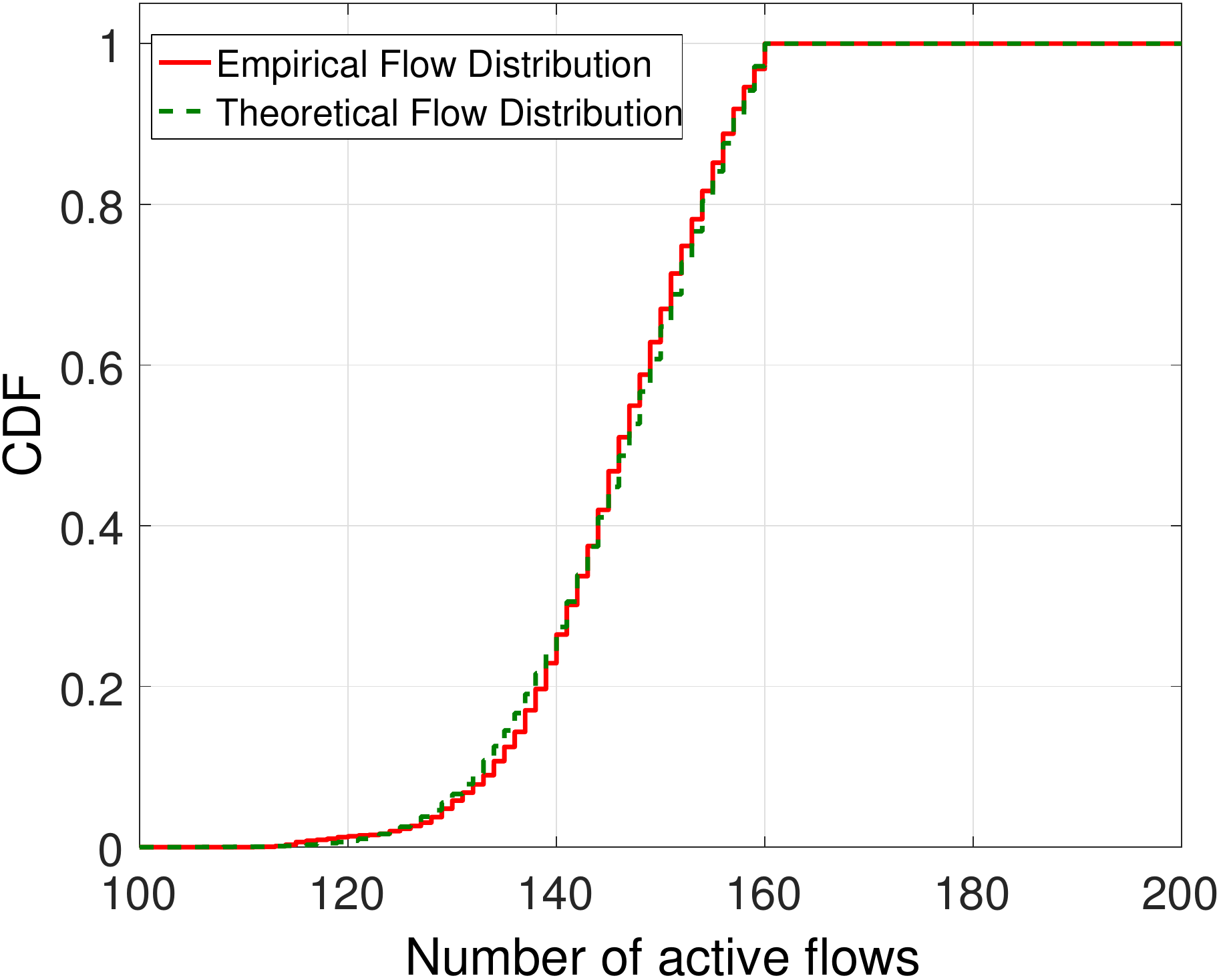}}
\subfigure[Scheme (IV) with $l=140, h=160$]{\label{fig:sacrifice1-2}\includegraphics[width=42mm,height=36mm]{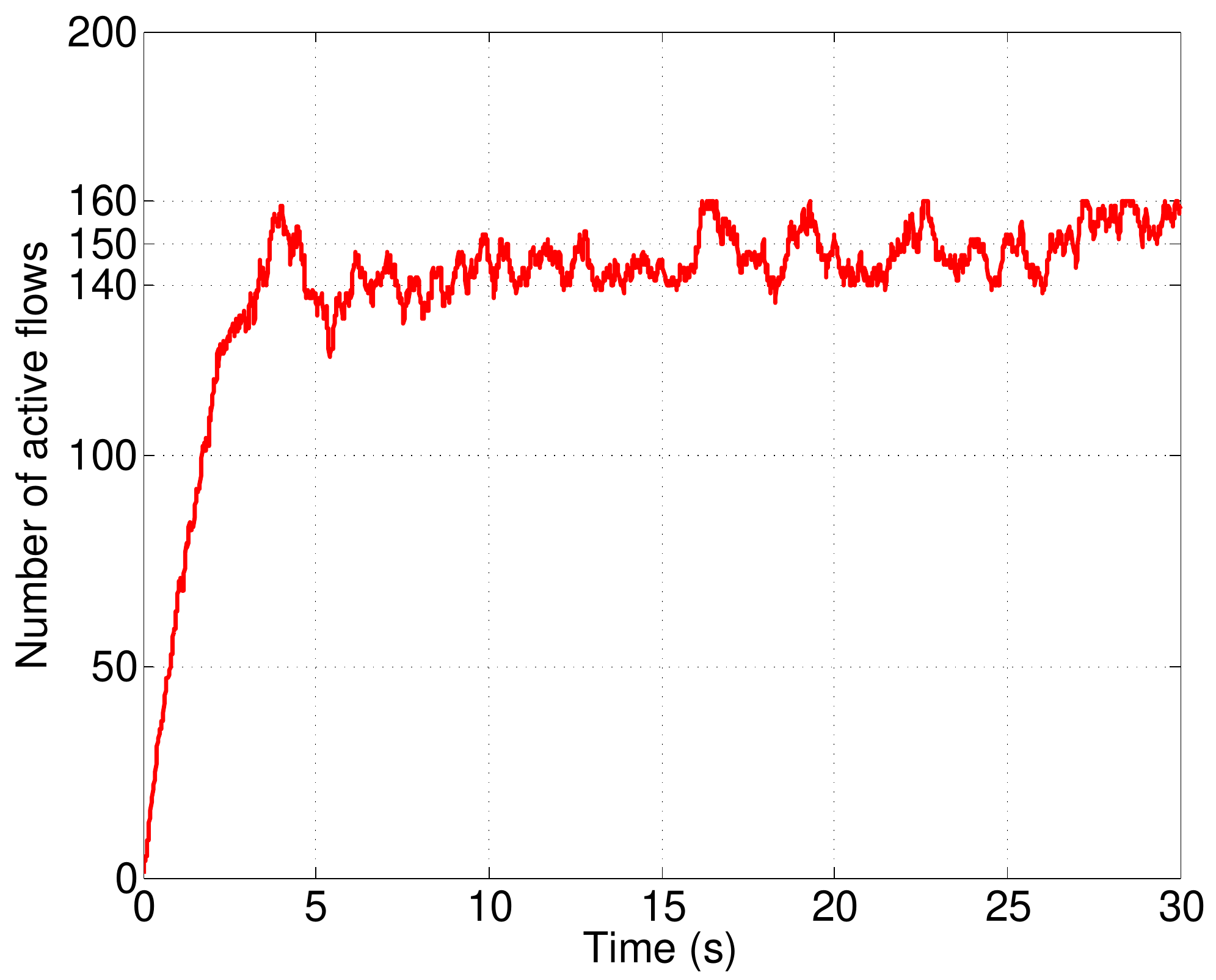}\includegraphics[width=42mm,height=35mm]{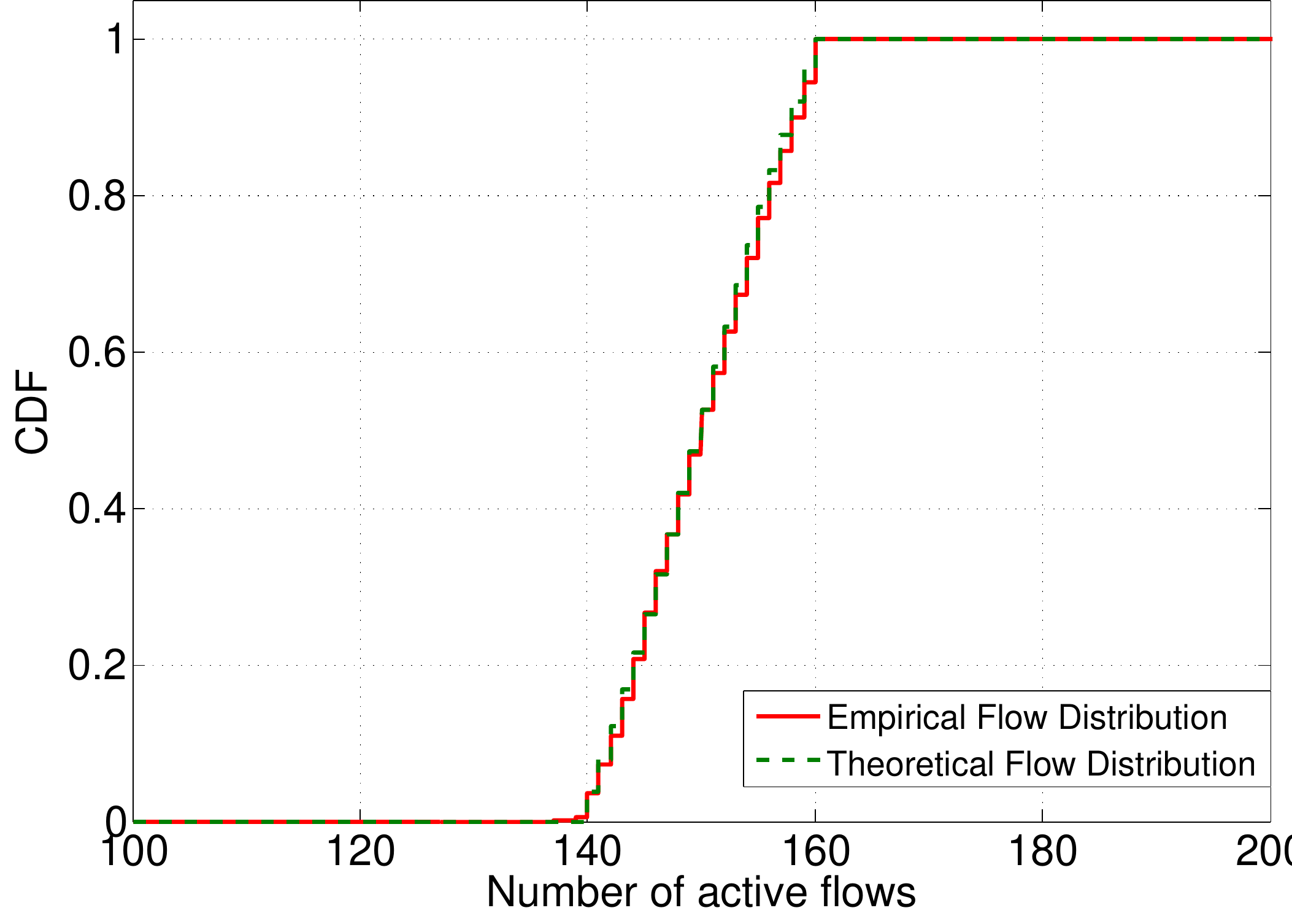}}
\subfigure[Scheme (V) with $h=160$]{\label{fig:sacrifice1-3}\includegraphics[width=42mm,height=36mm]{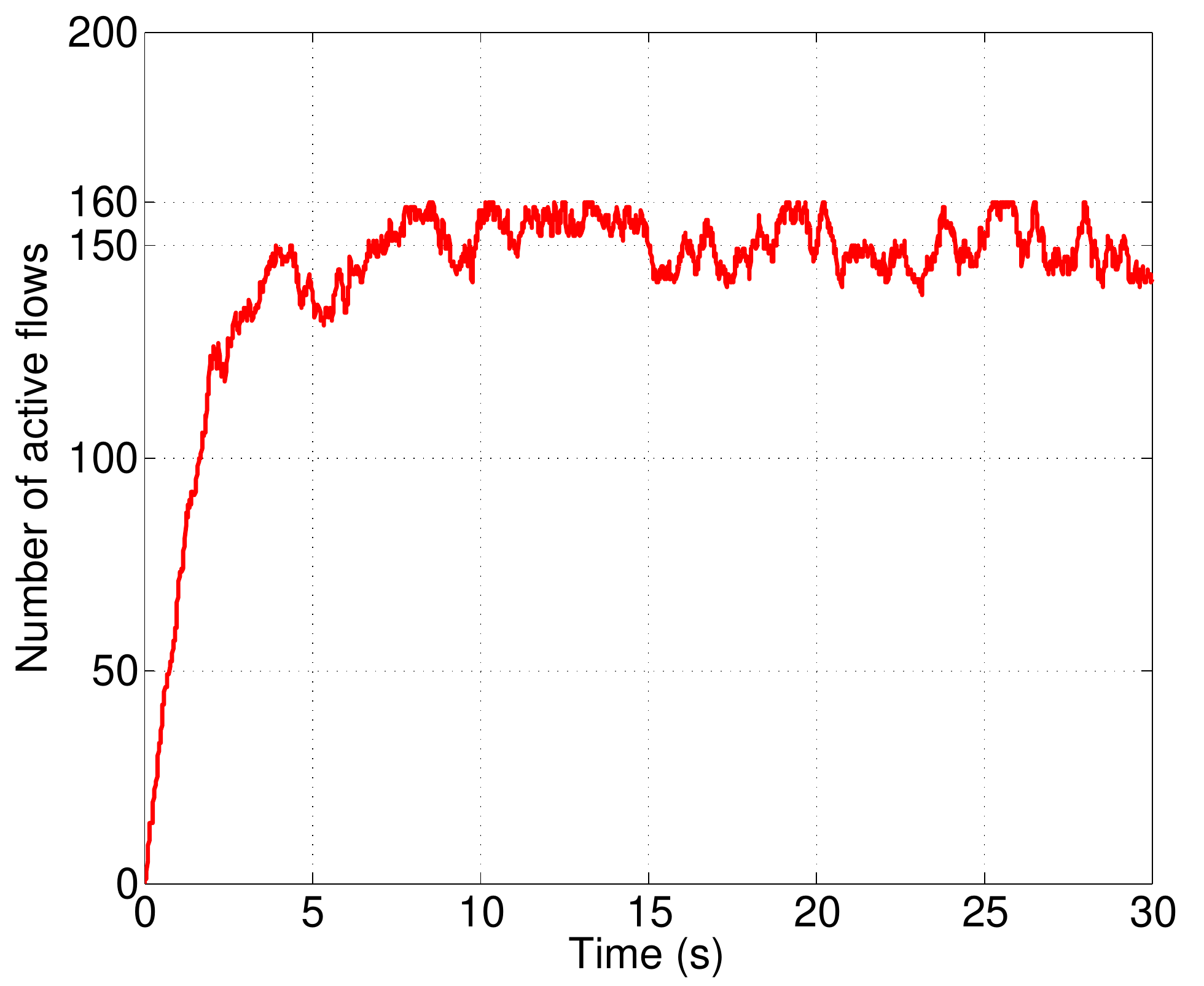}\includegraphics[width=42mm,height=35mm]{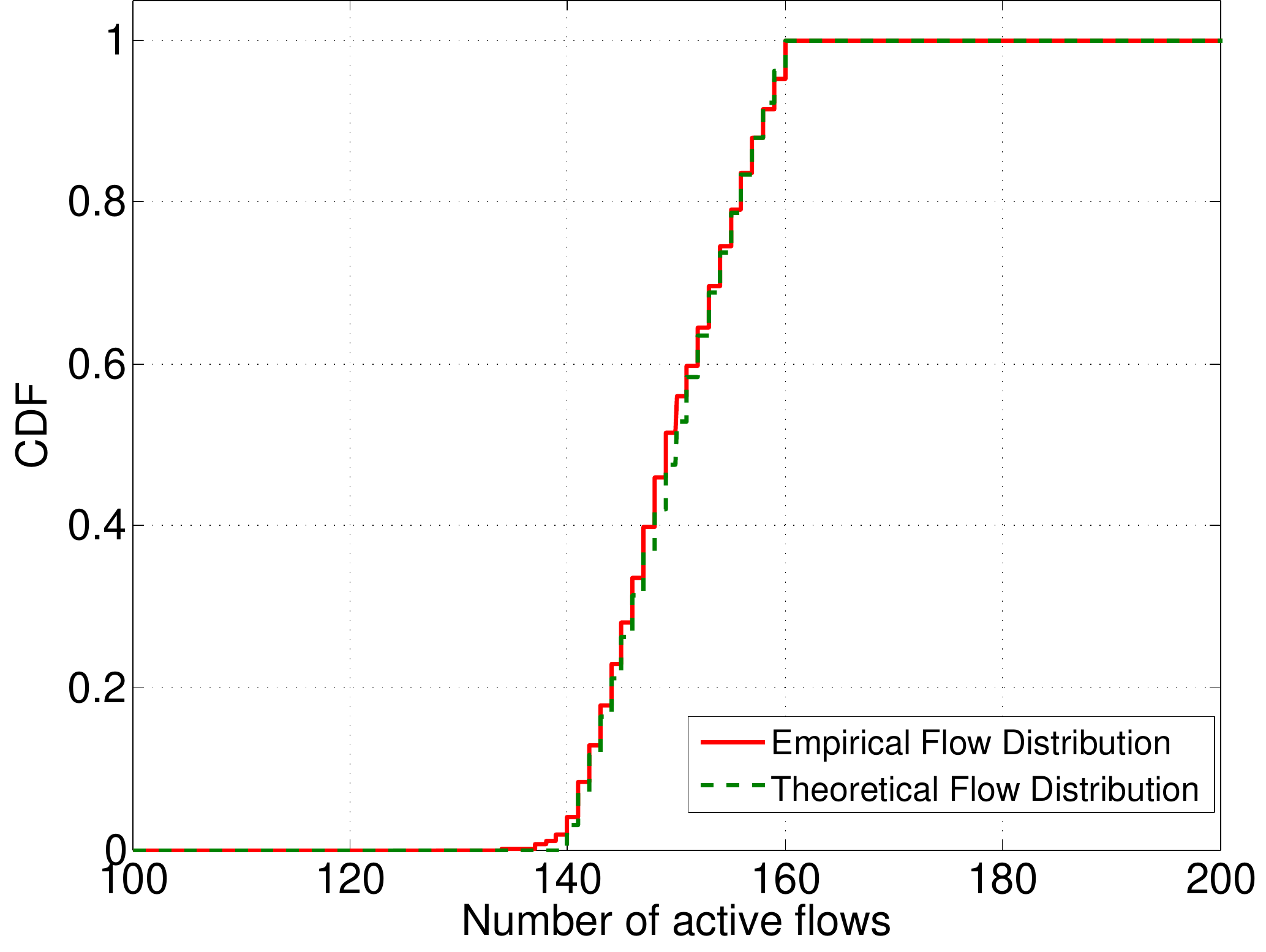}}\vspace{-2mm}
\caption{Variation in the number of active flows at a particular server
over time (left) and empirical vs.~theoretical flow distributions (right).
(a): Scheme (III);
(b): Scheme (IV);
(c): Scheme (V).}
\label{fig:sacrifice1}
\end{figure}

Simulations are conducted for a system with $n = 500$ servers.
Flows are initiated as a Poisson process with rate $\lambda = 100$
(per second) per server and flow durations are exponentially
distributed with mean $\beta = 1.5$ (seconds).
The average number of active flows at each server in stationarity is
$\rho = \lambda \beta = 150$.
The packet arrival rate is $\nu = 100$ (per second) per active flow
and the service rate at each server is $\mu = 20000$ packets (per second),
which implies that each server can handle up to $200$ concurrent flows.
Hence, the average system-wide utilization is $0.75$.

Simulating a system of the above size at the packet level is
prohibitively demanding for the time scale of flow dynamics.
Therefore, we adopt a hybrid analytical/simulation approach to
evaluate the packet-level performance.
Specifically, we first conduct flow-level simulations to obtain the
empirical flow distribution in terms of the probabilities~$p_i$,
and then substitute this into Equation~\eqref{perf3} to derive the
packet-level latency distribution.
%The simulation is run for a sufficiently long period of time that
%a steady-state regime is reached.

\subsection{Simulation results}

We first evaluate the performance of \textbf{Schemes (I) and (II)}
which both preserve perfect stickiness.

Figure~\ref{fig:perfect1} illustrates the variation over time in the
number of active flows at a typical server as well
as the corresponding stationary flow distributions under the two schemes.
Specifically, Figure~\ref{fig:perfect1-1} shows the performance
of the flow-level JSQ policy.
It is observed that this scheme perfectly stabilizes the flow population
at $\rho = 150$ (with little variation in the steady-state regime).
Figure~\ref{fig:power-1} shows the performance of the flow-level
power-of-$d$ policy with $d=2$.
The number of active flows is roughly kept around $\rho=150$,
though the variation is larger than that under the flow-level JSQ policy. 
Figures \ref{fig:perfect1-2}--\ref{fig:perfect1-4} illustrate the
performance of Scheme (II), i.e., the pull-based flow assignment scheme.
When $l = 0$ and $h = \infty$ (Figure~\ref{fig:perfect1-2}),
this scheme behaves like flow-level randomized load balancing where
the variation in the flow population is much greater than that for the
flow-level JSQ policy.
When $l = k^* = 150$ and $h = k^*+1 = 151$ (Figure~\ref{fig:perfect1-3}),
the flow-level pull-based scheme effectively keeps the number of active
flows around $\rho = 150$, achieving similar performance as the
flow-level JSQ policy whereas the flow variation is somewhat larger.
When $l = k^*-10 = 140$ and $h = k^*+10 = 160$
(Figure~\ref{fig:perfect1-4}), the flow population is effectively
limited to the range $[140,160]$, validating the theoretical analysis.
%However, if the thresholds~$l$ and~$h$ are not appropriately set,
%then the invite/disinvite scheme may lose its effectiveness.
%As is observed in Figures~\ref{fig:perfect1-5} and~\ref{fig:perfect1-6},
%the flow population at each server is no longer restricted to the range
%$[l,h]$.
It can further be observed that the empirical flow distributions
validate the theoretical results.
Note that we only have a theoretical bound for the flow-level
power-of-$d$ policy instead of the exact flow distribution.

Figure~\ref{fig:perfect3} compares the packet-level performance of
three load balancing schemes: flow-level power-of-$d$ scheme (and in particular flow-level JSQ policy), flow-level
pull-based scheme, and \textbf{packet-level} randomized load balancing.
There are several important observations.
First, the flow-level JSQ policy and the flow-level pull-based
scheme with $l = k^*$, $h = k^*+1$ achieve the best packet-level
latency performance among flow-level load balancing schemes.
This is expected since the two schemes yield the most balanced flow
distribution.
Under the flow-level power-of-$d$ scheme with $d=2$ and the flow-level
pull-based scheme with $l = k^*-10=140$ and $h = k^*+10=160$, the packet-level
latency performances are slightly worse but still reasonably good.
By comparison, the flow-level randomized load balancing scheme yields
the worst packet-level latency performance.
The second important observation is that even the best flow-level load
balancing scheme is outperformed by the simplest packet-level
randomized load balancing scheme.
This demonstrates the significant penalty for preserving strict stickiness.

Next, we numerically examine the performance of \textbf{Schemes (III),
(IV) and (V)} which may sacrifice stickiness for improvement in 
packet-level delay performance.

Figure~\ref{fig:sacrifice1} illustrates the variation over time in the
number of active flows at a typical server  as
well as the corresponding stationary flow distributions under the three schemes.
Specifically, Figure~\ref{fig:sacrifice1-1} shows the performance
of Scheme~(III), and indicates that
the empirical and the theoretical stationary distributions closely match.
Moreover, the number of active flows at each server is effectively
kept below the shedding threshold $h = 160$.
The corresponding results for Schemes (IV) and (V) are presented
in Figures~\ref{fig:sacrifice1-2} and~\ref{fig:sacrifice1-3},
respectively, and are qualitatively similar.
In case of Schemes (IV) and~(V), the number of active flows at each
server is effectively constrained in the range $[l,h]$ and $[i^*,h]$
in steady state, respectively.
%while in case of Scheme (V) the number of active flows at each
%server is effectively constrained in the range $[m,h]$
%(where $m = 140$ is derived from~\eqref{eq:derive_m}).

Figures \ref{fig:flow-shedding3}--\ref{fig:flow-min3} plot the
stickiness violation probability as a function of the value of the
upper threshold~$h$ for Schemes (III)-(V), respectively,
and indicate that the empirical values match the theoretical curves,
validating the analytical results.

\begin{figure*}[]
\begin{center}
\subfigure[Scheme (III)]{\label{fig:flow-shedding3}
\includegraphics[width=2.2in,height=2in]{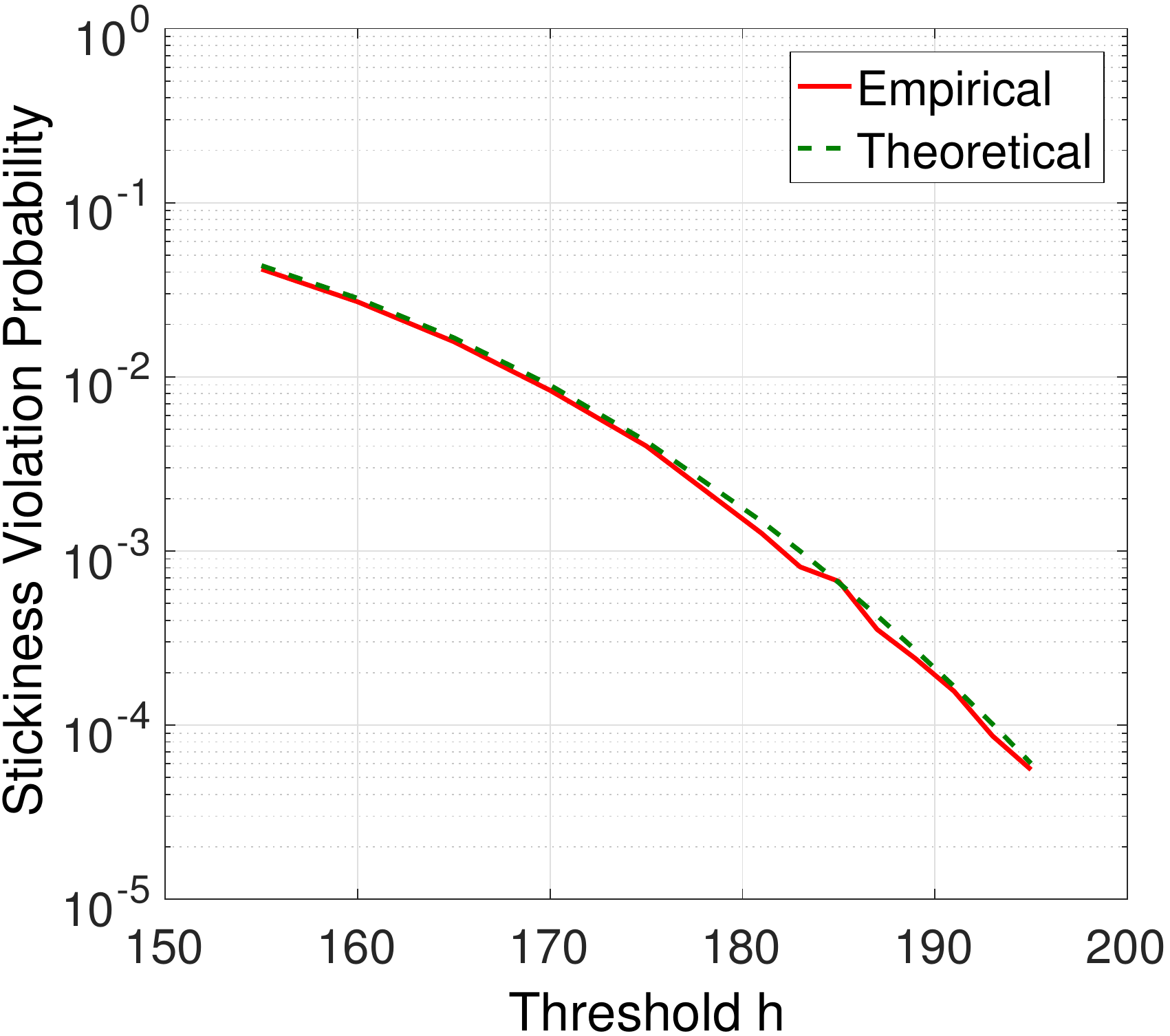}}
\subfigure[Scheme (IV)]{\label{fig:flow-invite3}
\includegraphics[width=2.2in,height=2in]{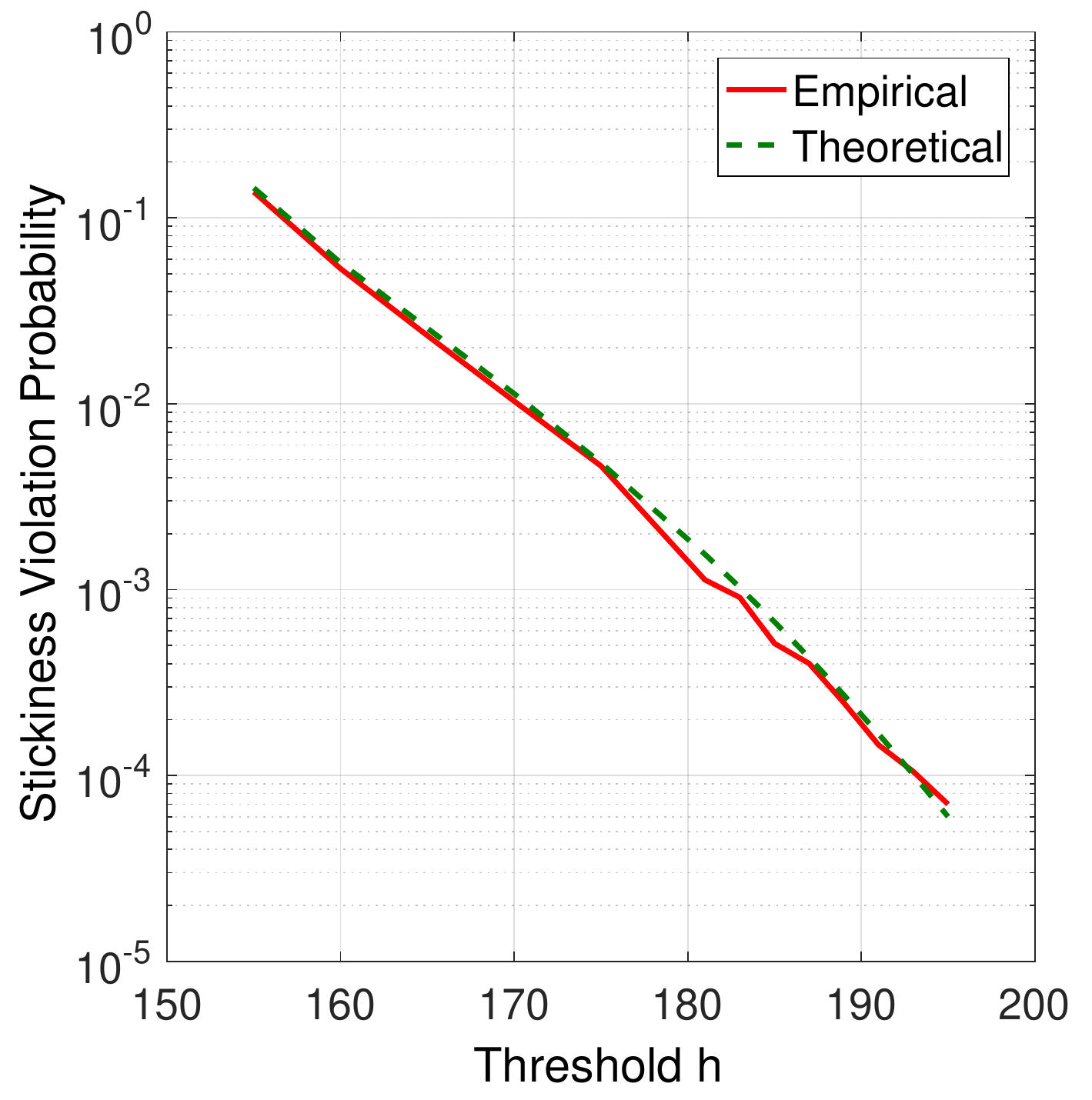}}
\subfigure[Scheme (V)]{\label{fig:flow-min3}
\includegraphics[width=2.2in,height=2in]{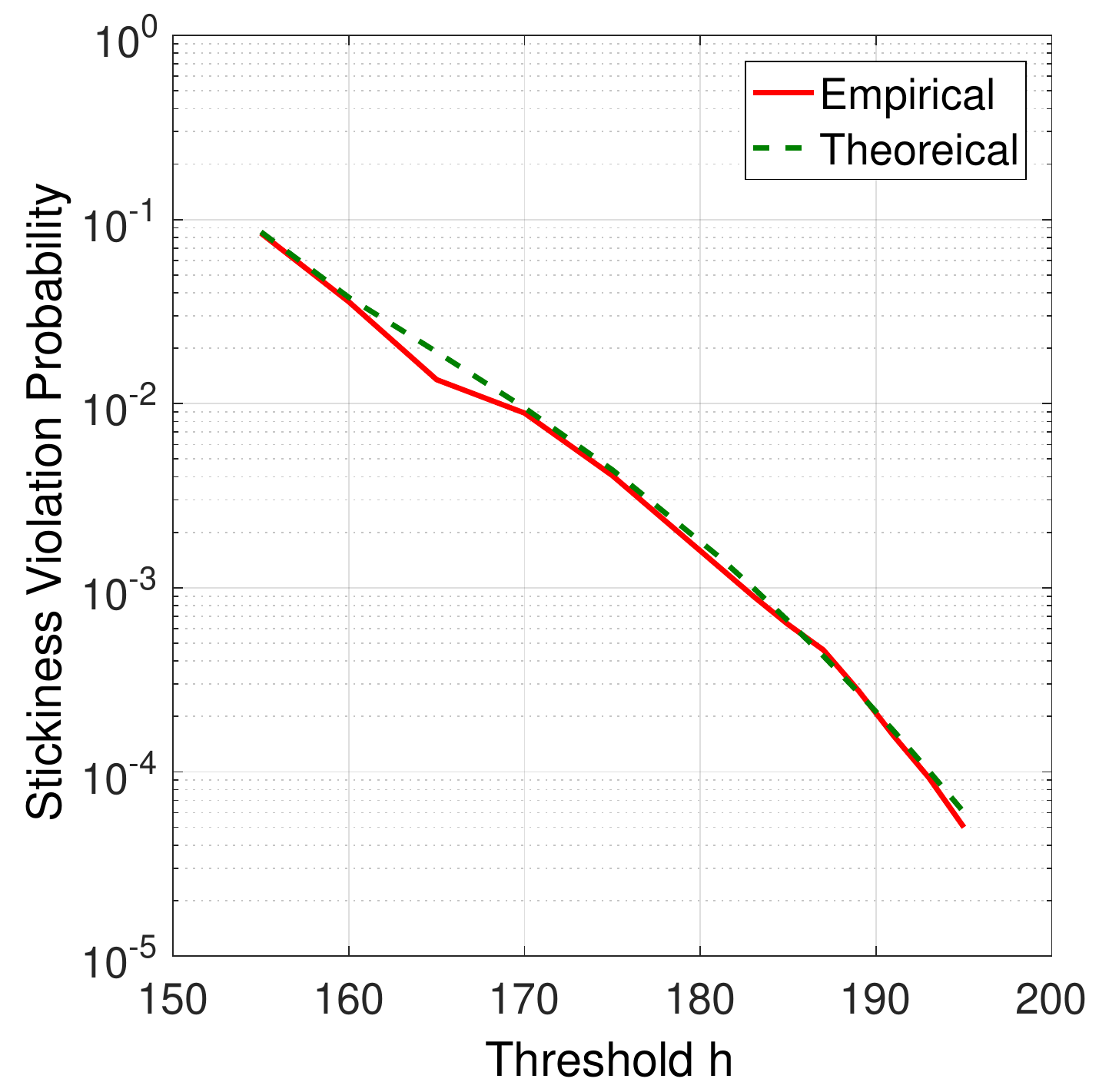}}\vspace{-2mm}
\caption{Stickiness violation probability for Schemes (III)-(V)
as a function of the upper threshold~$h$.}
\label{fig:sacrifice2}\vspace{-5mm}
\end{center}
\end{figure*}

\begin{figure*}[]
\begin{center}
\subfigure[Scheme (III)]{\label{fig:flow-shedding4}
\includegraphics[width=2.2in,height=2in]{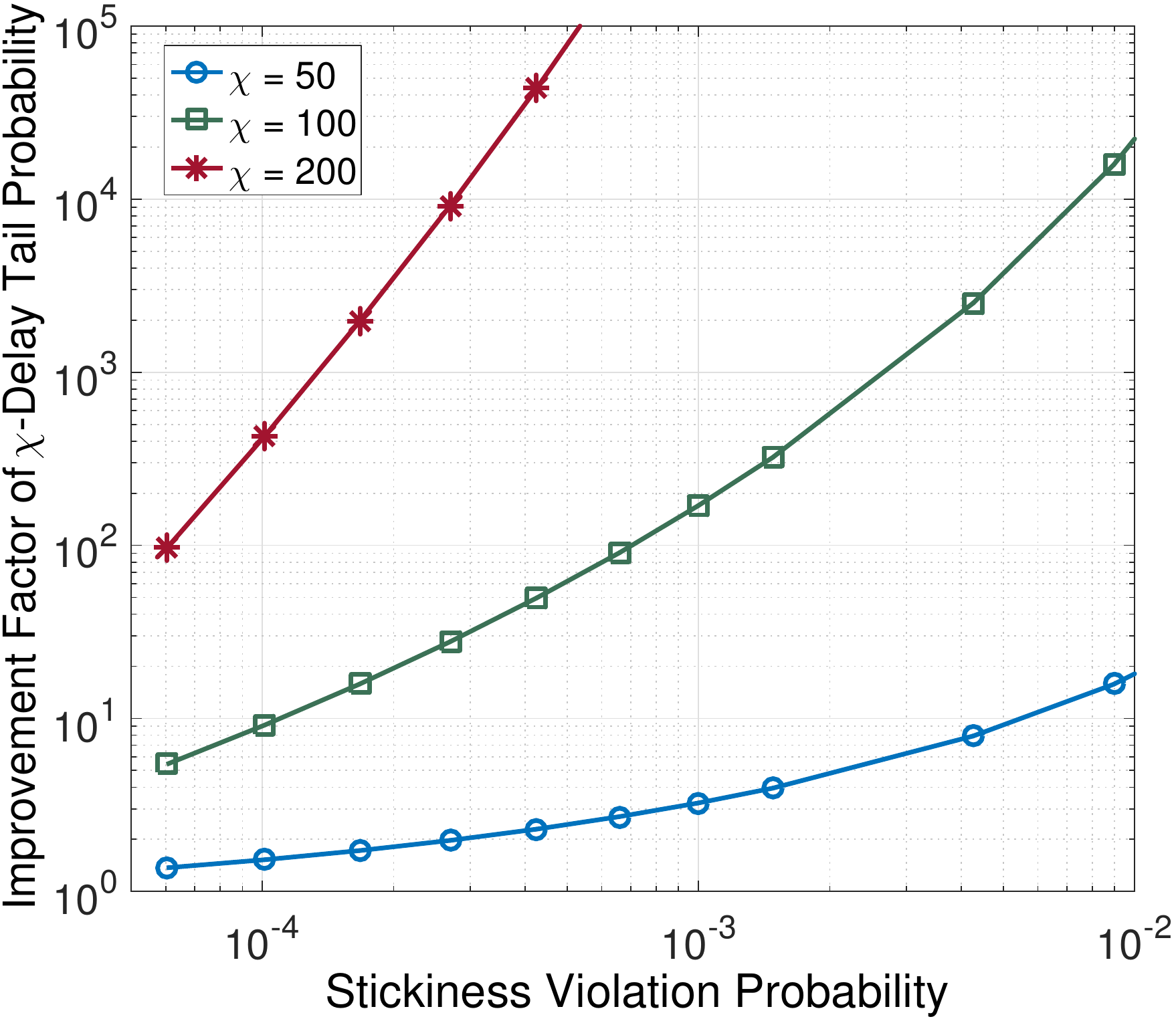}}
\subfigure[Scheme (IV)]{\label{fig:flow-invite4}
\includegraphics[width=2.2in,height=2in]{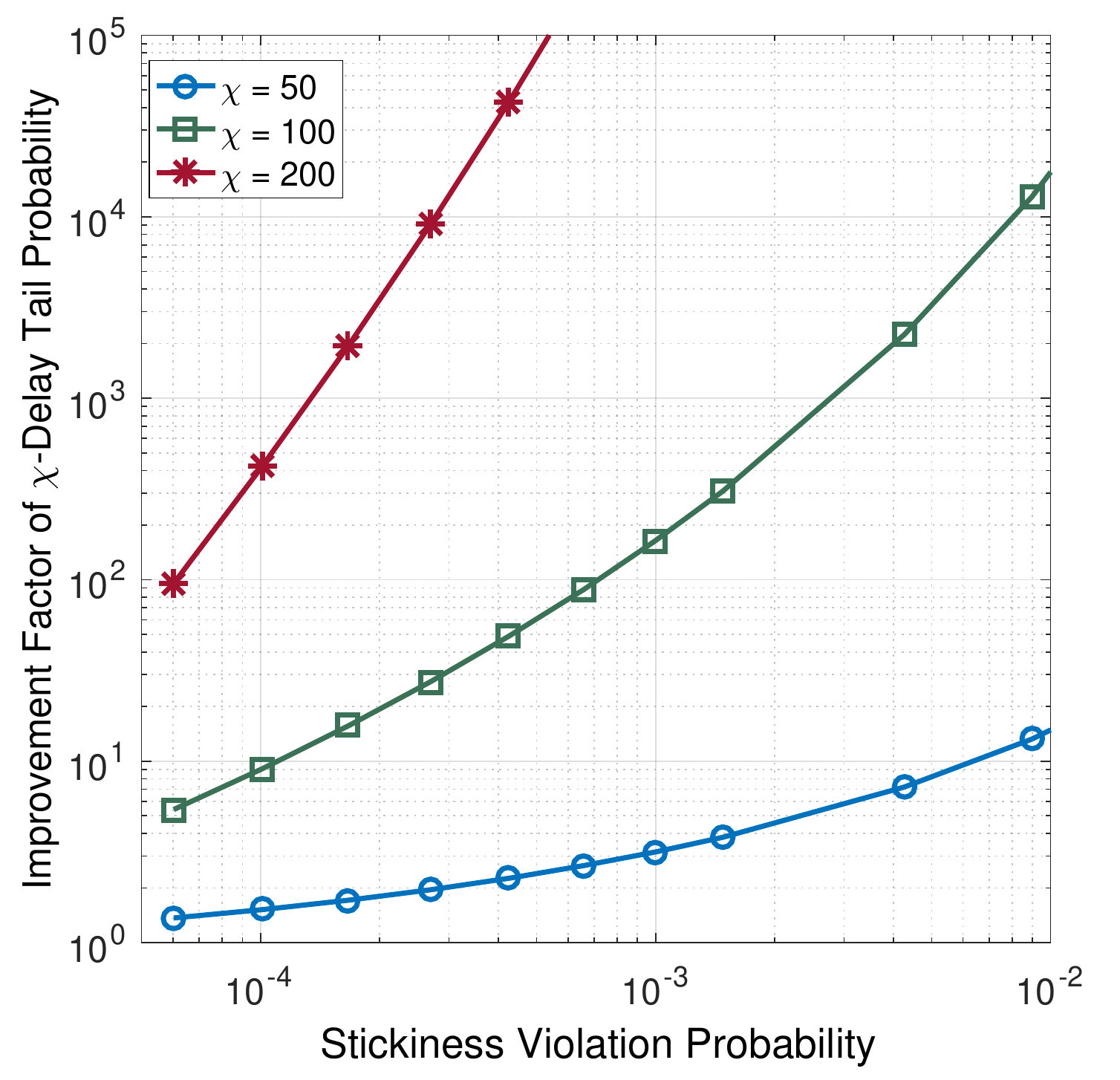}}
\subfigure[Scheme (V)]{\label{fig:flow-min4}
\includegraphics[width=2.2in,height=2in]{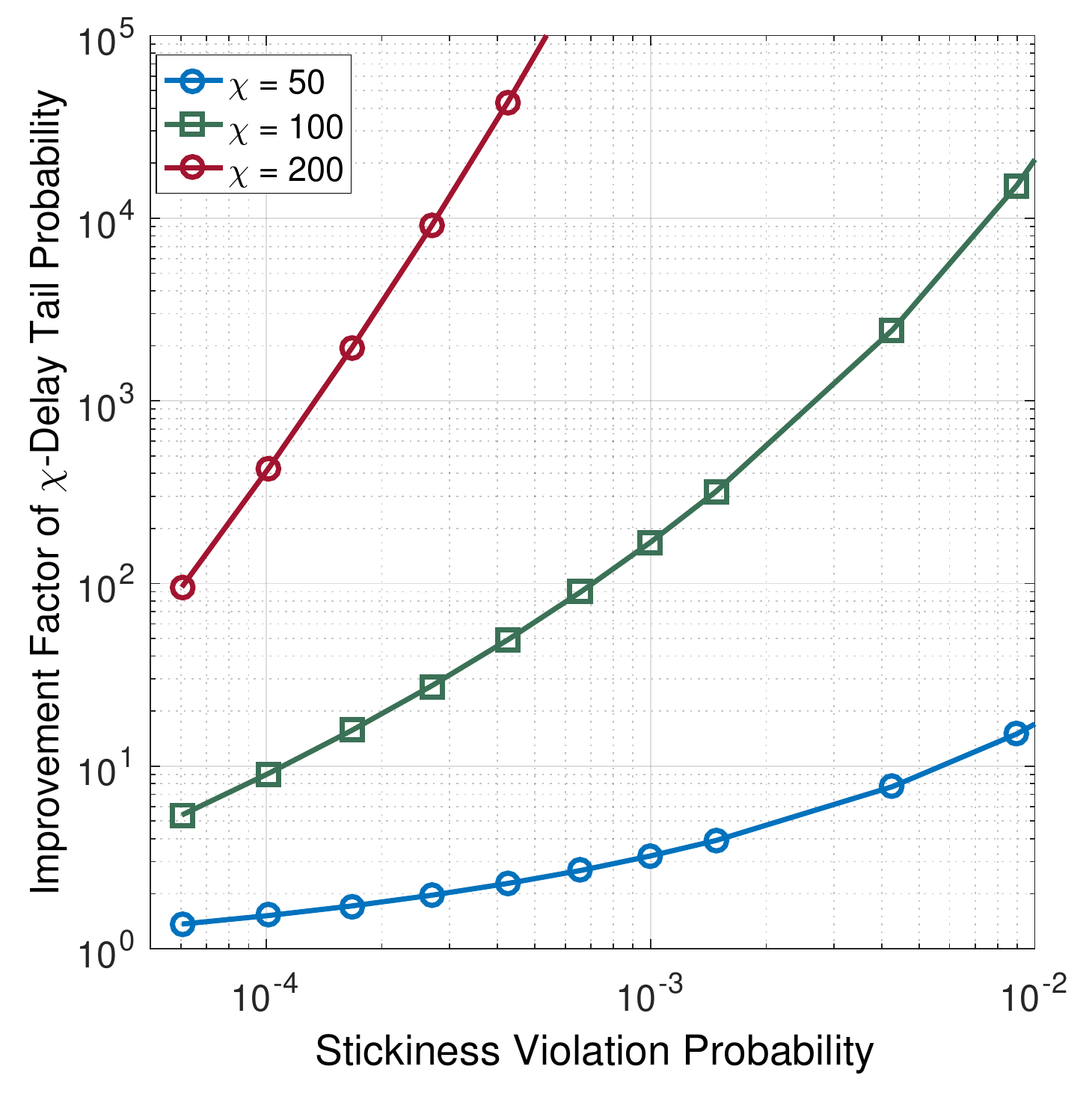}}\vspace{-2mm}
\caption{Stickiness violation versus improvement in packet-level latency
for Schemes (III)-(V).}
\label{fig:sacrifice3}\vspace{-5mm}
\end{center}
\end{figure*}

Figure \ref{fig:flow-shedding4} illustrates the trade-off between
the stickiness violation probability and the packet-level latency
performance, and shows qualitatively similar results for Schemes
(III), (IV) and (V).
To illuminate the benefits from stickiness violation, the packet-level
latency performance is measured by the improvement factor of the
$\chi$-delay tail probability as compared to the scenario with the
upper threshold $h = \infty$ (which preserves strict stickiness).
\textbf{We observe that relaxing the strict stickiness requirement by even
a minimal amount is highly effective in clipping the tail of the packet latency distribution.}
For example, for $\chi = 200$, a stickiness violation probability
as low as $\epsilon = 6 \times 10^{- 5}$ yields a reduction in the
delay tail probability by a factor~$100$.

\begin{figure}[]
\begin{center}
\includegraphics[width=64mm]{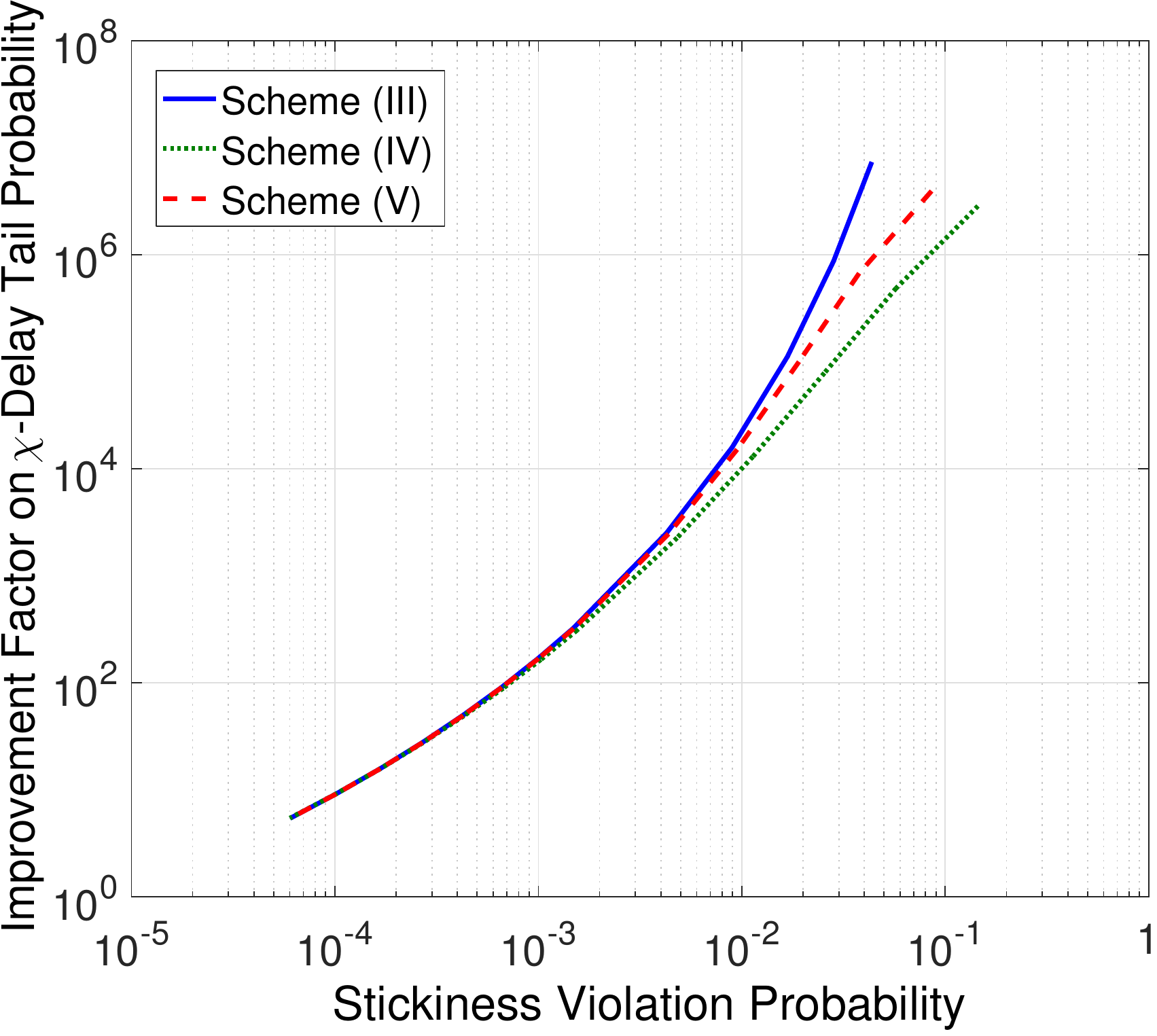}\vspace{-2mm}
\caption{Comparison among Schemes (III), (IV) and (V) with respect
to the trade-off between stickiness violation and packet-level latency
($\chi = 100$).}
\label{fig:flow-compare}\vspace{-5mm}
\end{center}
\end{figure}

%\begin{figure}[ht!]
%\begin{center}
%\includegraphics[width=64mm]{flow-shedding4}\vspace{-2mm}
%\caption{Stickiness violation versus improvement in packet-level
%latency for Scheme (III).}
%\label{fig:flow-shedding4}\vspace{-5mm}
%\end{center}
%\end{figure}

%The three schemes will be explicitly compared below in terms of the
%trade-off between stickiness violation and packet-level latency.
Figure~\ref{fig:flow-compare} shows the comparison among Schemes (III),
(IV) and (V) in terms of the trade-off between stickiness violation
and packet-level latency.
It is observed that the three schemes have quite similar performance
when the stickiness violation probability is relatively small.
However, as the stickiness violation probability grows, Scheme (III)
tends to outperform Schemes (IV) and (V).
This may be explained from the fact that Scheme (III) {\em discards\/}
overloaded flows, while the other two schemes {\em transfer\/}
overloaded flows.
Moreover, Scheme~(V) performs better than Scheme (IV), at the expense of higher
communication overhead since it requires load information
from all servers whenever a flow transfer occurs.

\section{Bin-based load balancing scheme}
\label{sec:bin-level}

In the previous sections we examined the trade-off between the flow
stickiness violation probability and the packet-level latency
performance in a scenario where load balancing can be performed at the
granularity level of individual flows.
However, flow-level load balancing is not scalable in practice since
a flow table needs to be maintained to record the assignment of all
active flows in the system.
In large-scale deployments with massive numbers of flows,
the associated flow table may quickly become unmanageable.
In this section, we propose a scalable bin-based load balancing scheme
that explicitly accounts for flow stickiness.
Simulation results (see Subsection~\ref{sec:bin_sim}) show that this
scheme achieves a good trade-off between stickiness violation
and packet-level performance. 

\subsection{Description of Bin-based Load Balancing}
In this subsection we describe the bin-based load balancing scheme. As discussed earlier, an efficient load balancing scheme should achieve a good balance among the following three aspects.
\begin{itemize}
\item{\textbf{Scalability.} The scheme must be able to support high packet forwarding rates, and hence should only involve minimal complexity per packet, in terms of both computation and communication overhead. Thus the size of a packet forwarding table should not be too large, nor should the amount of state information required to configure and dynamically adapt the forwarding rules be too large.}
\item{\textbf{Stickiness.} The scheme should support flow stickiness and ideally have the tunability for the degree of stickiness.}
\item{\textbf{Packet-level Performance.} The scheme should evenly distribute the traffic load across the available servers so as to optimize some relevant performance metrics in terms of packet delay, such as tail statistics and mean values (see Section \ref{sec:metric-packet} for detailed metrics).}
\end{itemize}

\begin{figure}[]
\center
\includegraphics[width=3.2in]{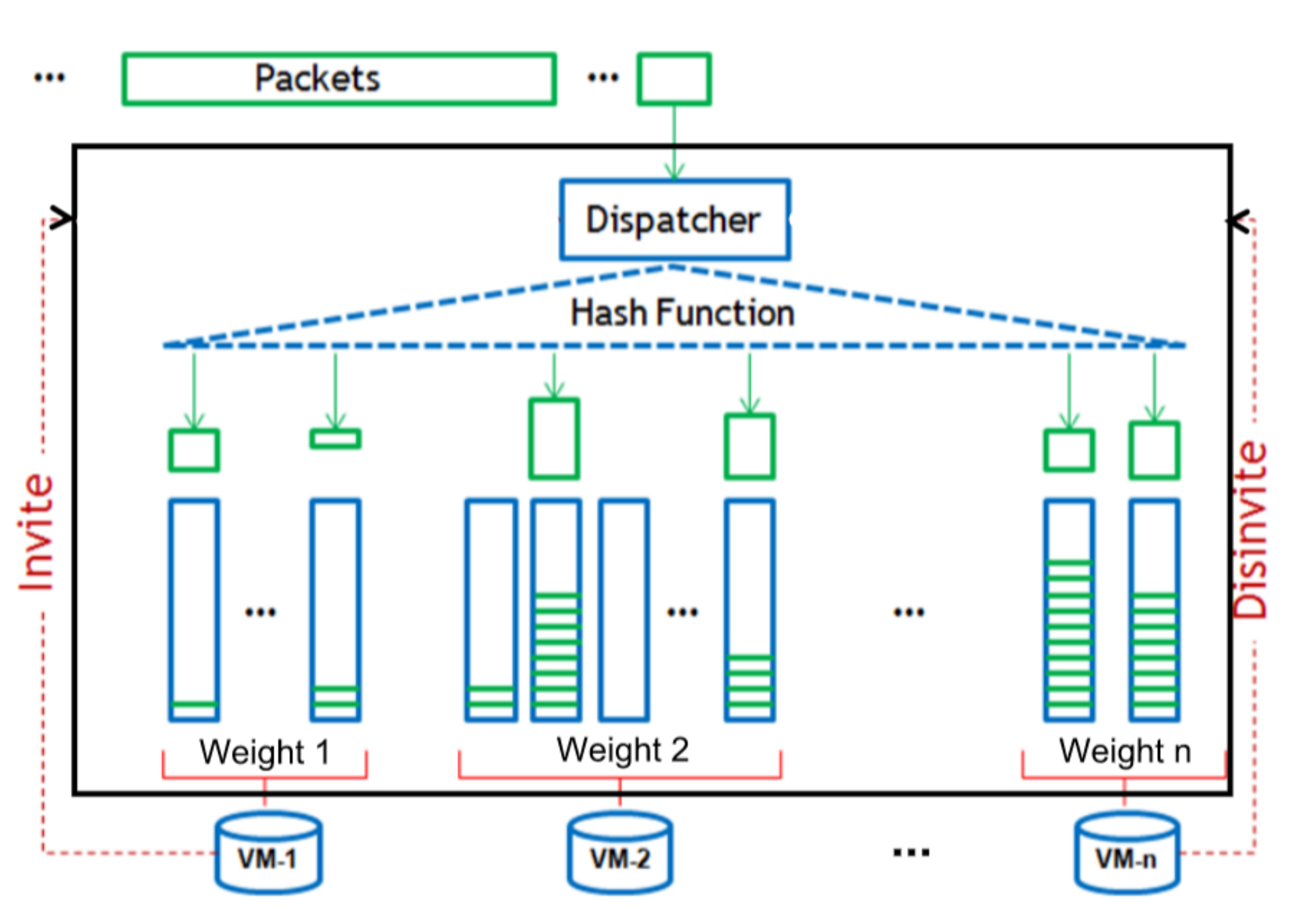}\vspace{-2mm}
\caption{Schematic of bin-based load balancing scheme}
\label{fig:system2}
\end{figure}

\begin{figure*}[]
\centering
\subfigure[Flow population variation over time, where $m=5n=2500$ bins are used.]{\label{fig:bin-variation1}\includegraphics[width=2.3in,height=2in]{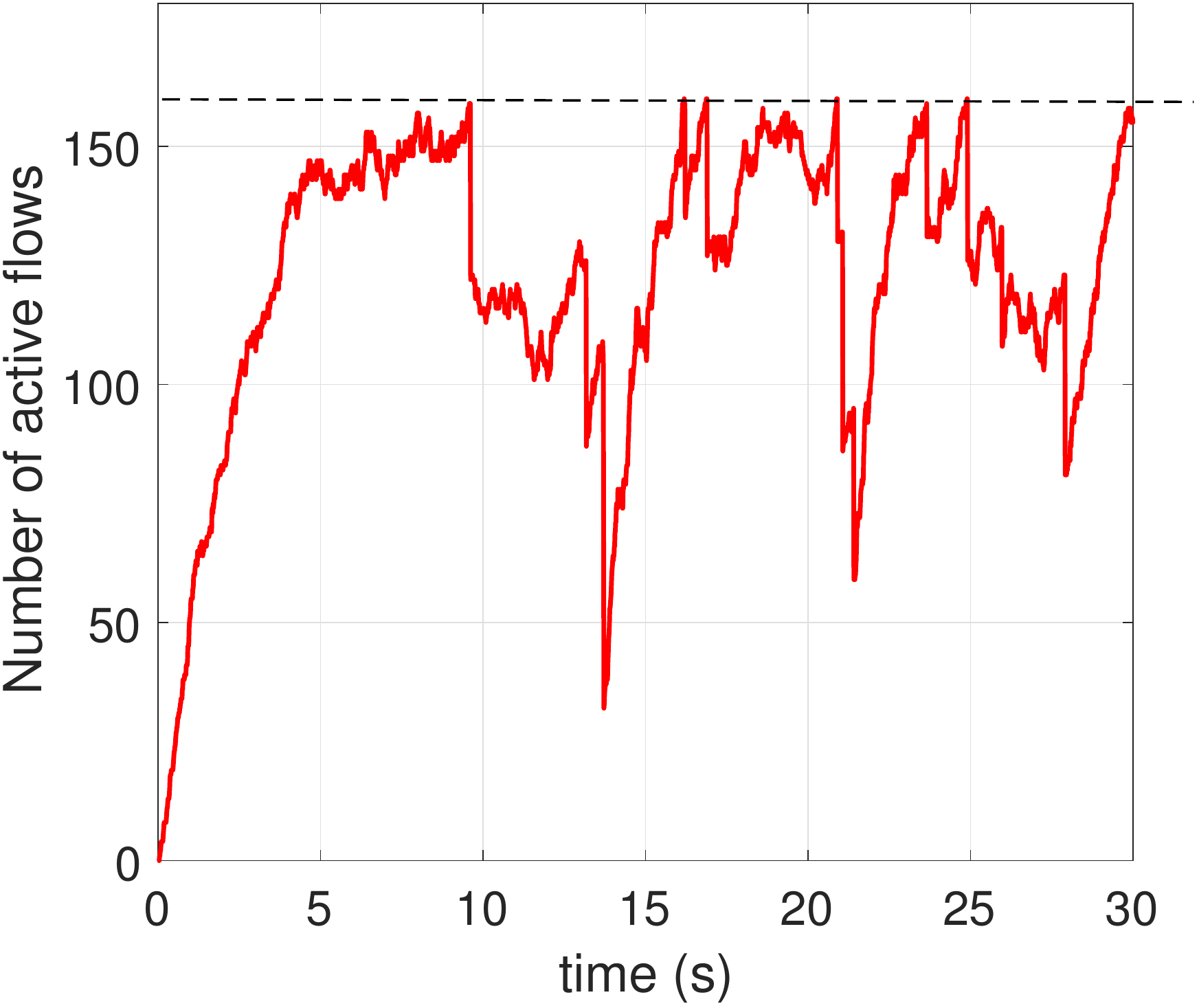}}
\subfigure[Flow population variation over time, where $m=10n=5000$ bins are used.]{\label{fig:bin-variation2}\includegraphics[width=2.3in,height=2in]{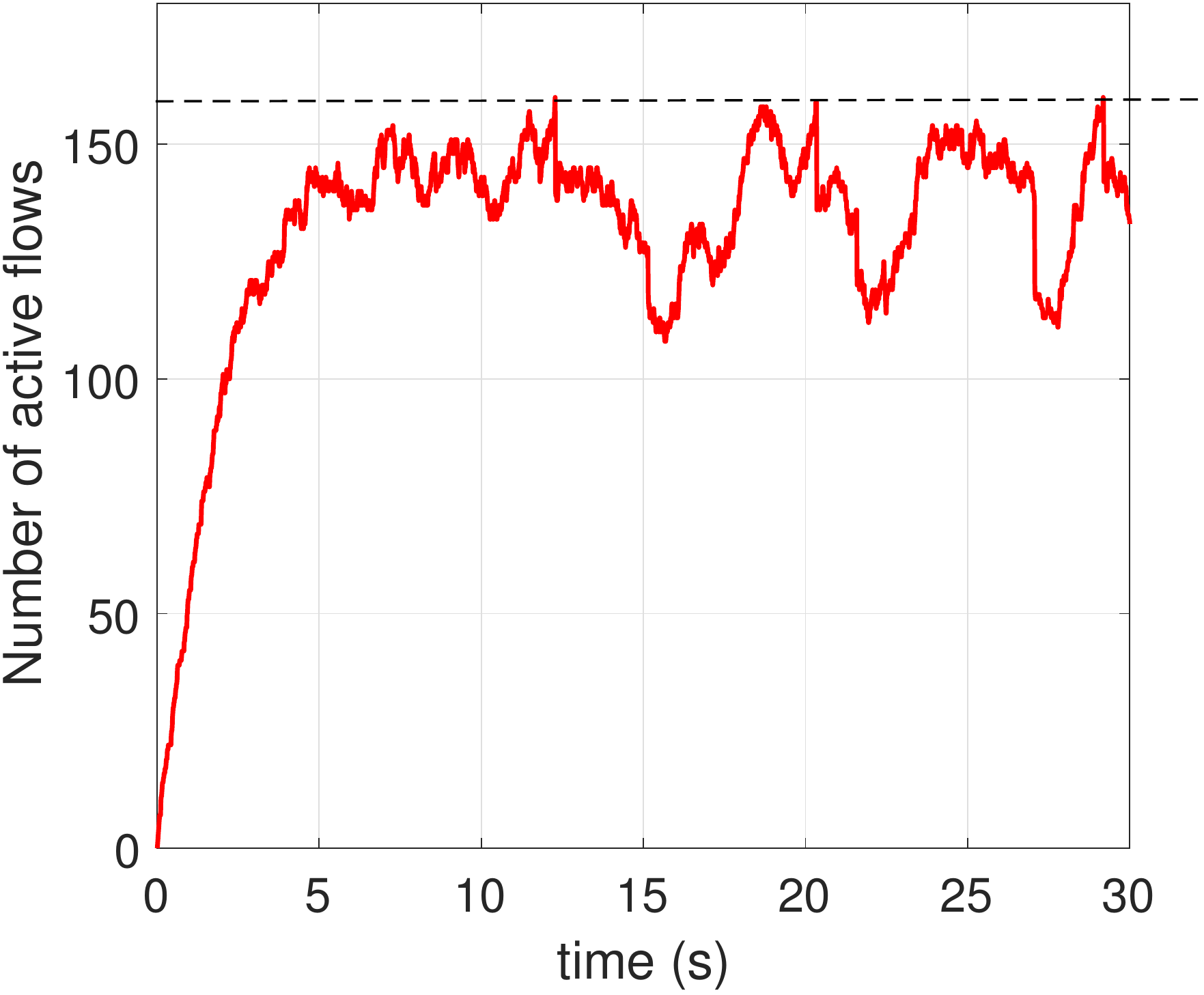}}
\subfigure[CDF of flow population]{\label{fig:bin-CDF}\includegraphics[width=2.3in,height=2in]{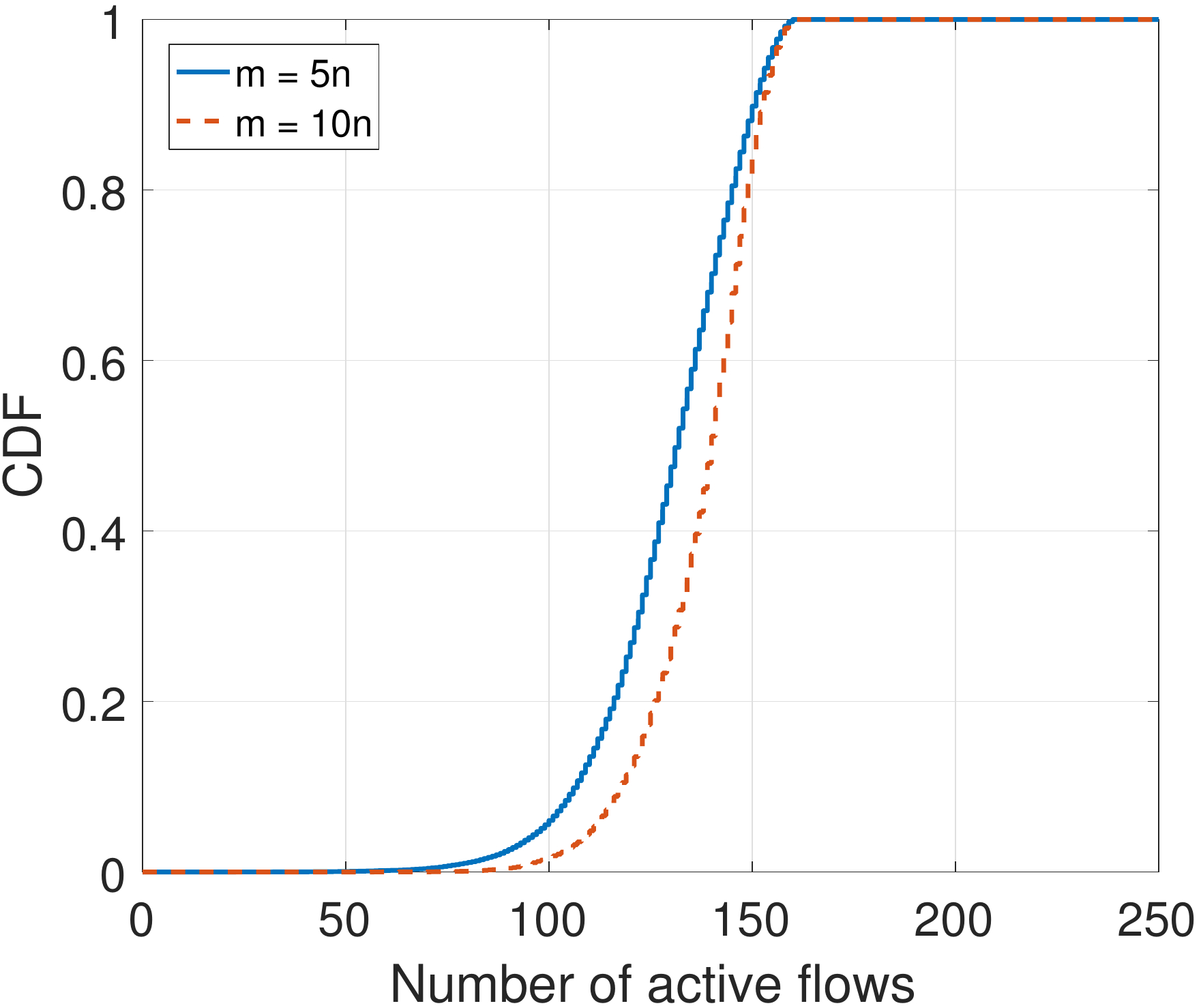}}
\caption{Variation in the number of active flows at a server over time and corresponding CDF of flow population, where $l=140$ and $h=160$.}
\label{fig:bin-variation}
\end{figure*}

In view of the above-specified requirements, we propose a bin-based load
balancing scheme, as is illustrated in Fig.~\ref{fig:system2}.
The key feature of the dispatching mechanism consists of a collection
of virtual 'bins' as an intermediary between incoming packet flows and servers.
Each incoming packet goes through the following two steps in order
to decide its destination server.

\begin{itemize}
\item[(1)]{\textbf{Packets to Bins.} The dispatcher maps each arriving packet  to some bin according to a \emph{static hash function} based on the five-tuple in the packet header. 
%Such a static mapping from packets to bins allows extremely fast packet forwarding rate and provides flow stickiness since packets belonging to the same flow are hashed to the same bin.
}

\item[(2)]{\textbf{Bins to Servers.} After determining the corresponding bin for each packet, the load balancing entity then looks up the \emph{bin table} which records the mapping from bins to servers. If a packet is hashed to bin $j$ and bin $j$ is associated with server $k$, then the packet is forwarded to server $k$ for processing. In contrast to the static hashing from packets to bins, the association of bins with servers is dynamically managed according the following pull-based bin re-allocation rule.}

\item[]{\textbf{Pull-based bin re-allocation rule.}  Each server maintains a simple load estimate  and reports status information to the load balancer. When the load at a server reaches the upper threshold~$h$,
a disinvite message is sent from the server to the dispatcher;
the disinvite message is revoked as soon as the load
drops below the level~$h$ again. When the load at a server falls below the lower threshold~$l$, an invite message is sent from the server to the dispatcher;
the invite message is retracted as soon as the load at the server reaches the level~$l$ again. Moreover, when the load at a server exceeds the threshold $h$, a randomly selected bin is deallocated from the server, and re-allocated to an arbitrary server with an outstanding invite message, if any. Otherwise, the bin is re-allocated to an arbitrary server without
an outstanding disinvite message, if any. If all servers have outstanding disinvite messages, then the bin is re-allocated to a randomly selected server.}
\end{itemize}

\noindent \textbf{Stickiness-delay trade-off.}
In the above bin-based mechanism, all packets belonging to the same
flow are hashed to the same bin.
If the mapping from bins to servers remains unchanged, then all
packets belonging to the same flow are forwarded to the server,
ensuring perfect flow stickiness.
However, in this case the above bin-based scheme becomes flow-level
randomized load balancing which delivers poor performance.
On the other hand, if a bin is re-allocated, then any existing active
flow in the bin will lose stickiness, in return for the improvement
of packet-level performance.
Note that the stickiness violation probability is determined by the
frequency of bin re-allocations and the number of active flows
in a re-allocated bin;
these quantities can be tuned by setting the thresholds~$l$ and~$h$
as well as the number of bins.
Thus, the bin-based scheme can achieve any desired trade-off between
stickiness violation and packet-level performance.

\noindent \textbf{Scalability issues.}
For each packet, the above bin-based scheme only involves one static
hashing computation (from packets to bins) and one simple table lookup
(from bins to servers).
The communication overhead in terms of load status reports is decoupled
from packet arrivals and occurs only occasionally if the thresholds~$l$
and~$h$ are properly set.
As a result, the bin-based load balancing scheme can support very high
packet forwarding rates.
The only bottleneck lies in the size of the bin table which is
proportional to the number of bins.
Simulation results (see Subsection~\ref{sec:bin_sim}) suggest that using
$m=10n$ bins (where $n$ is the number of servers) is sufficient to
achieve a good trade-off between stickiness violation and packet-level
performance.

\noindent \textbf{Heterogeneous scenarios.}
Although we focus on the scenario with homogeneous server capacities,
the bin-based scheme is well suited for heterogeneous scenarios
where servers may have different packet processing rates.
In particular, through the dynamic adjustment of bin assignment,
the number of bins associated with each server will ultimately
stabilize at a level proportional to its processing rate,
even if the bin assignment is incorrectly configured at the beginning.
Moreover, to speed up convergence of the bin adjustment process,
we can re-allocate a bin when it becomes ``almost empty''
(e.g., contains few flows).
However, this operation increases the stickiness violation probability,
leading to another interesting trade-off between stickiness violation
and convergence speed, which is however beyond the scope of the present
paper.

Unfortunately it turns out to be difficult to analytically derive the
stationary distribution of the flow population under the bin-based
scheme due to the complex mutual interaction between flow dynamics
and bin dynamics.
Hence we rely on simulation experiments to evaluate its performance.

\subsection{Numerical Evaluation}\label{sec:bin_sim}

The simulation setting is the same as in Section~\ref{sec:simulation-flow-level} and omitted for brevity. In the following, we focus on the influence of the two thresholds $l$ and $h$ as well as the number of bins $m$. As mentioned above, these parameters determine the balance among scalability, flow stickiness violation and packet-level performance. For simplicity, the load of a server is measured by the number of active flows at that server (in practice it is more convenient to measure average server utilization).

Figure \ref{fig:bin-variation} illustrates the variation over time in the number of active flows and the corresponding flow population distribution at a typical server under the bin-based scheme, where we set $l=140$ and $h=160$ (in the number of active flows). It is observed that the number of active flows is effectively kept below $h=160$. However, unlike those flow-level load balancing schemes in Section \ref{sec:flow-level-lb}, the flow population is not well kept above the lower threshold $l=140$. This implies a more imbalanced flow population distribution (and thus worse delay performance) under the bin-based scheme than under those flow-level load balancing schemes. Moreover, the more bins are used, the more balanced the flow population distribution is. In fact, as the number of bins~$m$ grows large, the bin-based scheme reduces to the flow-level load balancing scheme (IV) (see Section \ref{flow-perfect-stickiness}). However, using more bins leads to a larger bin table, which shows a trade-off between scalability and packet-level performance.

Figure \ref{fig:bin-number-violation} plots the stickiness violation probability as a function of the value of the bin re-allocation threshold $h$. It is observed that using more bins contributes to a lower stickiness violation probability. However, as mentioned earlier, using more bins implies a larger bin table. Thus, there is a trade-off between scalability and stickiness violation.

\begin{figure}[]
\center
\includegraphics[width=2.8in]{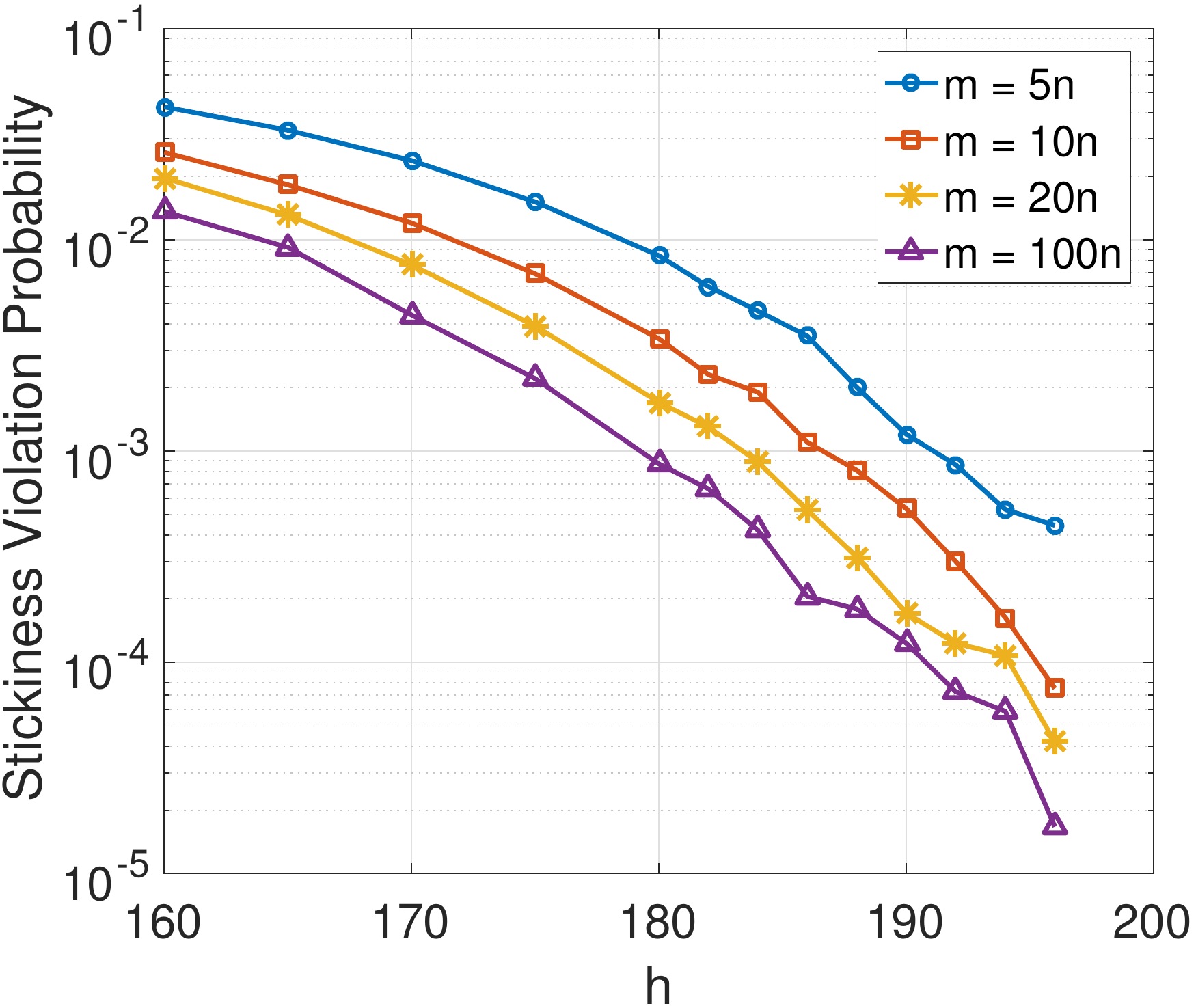}\vspace{-3mm}
\caption{Stickiness Violation Probability under different number of bins, where $m$ is the number bins and we set $l=140$.}
\label{fig:bin-number-violation}
\end{figure}

Figure \ref{fig:bin-number-trade-off} illustrates the trade-off curve between stickiness violation and packet-level performance, for various numbers of bins. There are two important observations. First, the trade-off curve becomes better as we increase the number of bins, meaning that a larger improvement in packet-level delay can be achieved with the same amount of stickiness violation.  However, as noted earlier, using more bins leads to worse scalability. Moreover, the improvement brought by increasing bins is diminishing as $m$ grows large. Considering modern data centers with tens of thousands of servers, we claim that using $m=10n$ (where $n$ is the number of servers) may be a good choice for balancing scalability, stickiness violation and packet-level performance. The second observation is that the performance benefits brought by relaxing the stickiness requirement is still significant even under the bin-based scheme. For example, when $m=10n$ bins are used, a stickiness violation probability as low as $\epsilon = 7\times 10^{-5}$ yields a reduction in the $\chi$-delay tail probability by a factor 100 (where $\chi=200$).

\begin{figure}[]
\center
\includegraphics[width=2.8in]{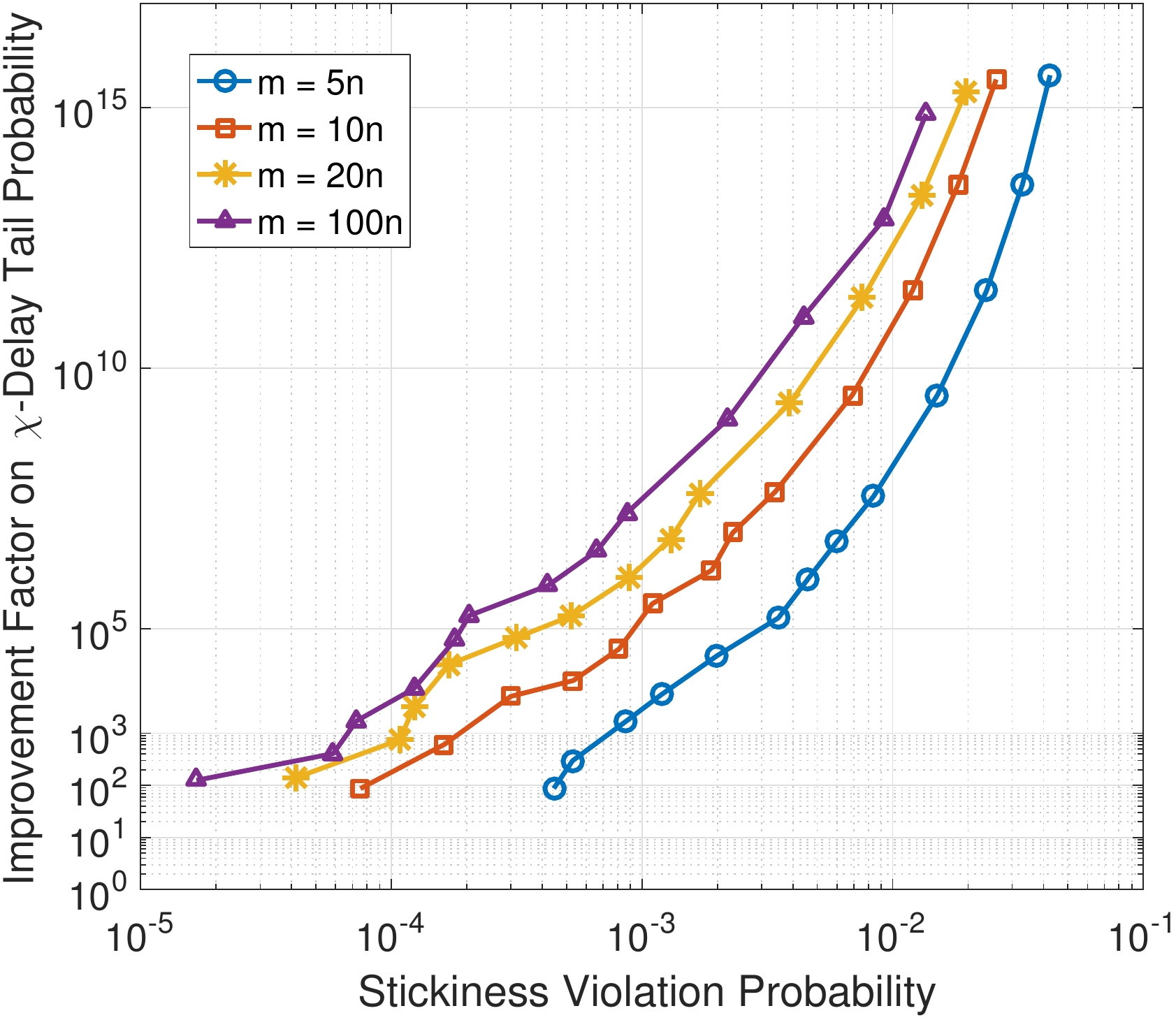}\vspace{-3mm}
\caption{Trade-off between stickiness violation and packet-level delay performance, where $m$ is the number bins and we set $l=140$.}
\label{fig:bin-number-trade-off}
\end{figure}

%Finally, we apply the bin-based scheme in a heterogeneous scenario where servers have different packet processing capacities. 

\section{Conclusion}
\label{conc}

We have investigated the fundamental trade-off between flow stickiness
violation and packet-level delay performance for load balancing.
Our theoretical and simulation results show that a stringent flow
stickiness requirement carries a significant penalty in terms
of packet-level delay performance.
Moreover, relaxing the stickiness requirement by a minuscule amount
is highly effective in clipping the tail of the latency distribution.
We further propose a bin-based load balancing scheme which achieves
a good balance among scalability, flow stickiness violation
and packet-level delay performance.

\newpage

\appendix

\section{Mean-field limits and fixed points for flow-level load
balancing schemes}

%\subsection{Random flow assignment with possible blocking (load shedding)}
%
%In this case, an arriving flow is assigned to a randomly selected server.
%However, if the selected server already has $h$~active flows,
%then the flow is immediately terminated and discarded.
%
%We simply have
%\[
%q_{i-1}(s) = s_{i-1} - s_i,
%\hspace*{.4in} i = 1, \dots, h,
%\]
%and $q_{i-1}(s) = 0$ for all $i \geq h + 1$.
%This yields a system of linear differential equations which can be
%explicitly solved.
%In particular, the fixed point is determined by
%
%Of course, in this case the numbers of active flows at the various
%servers evolve as in $n$~independent Erlang-B loss systems each with
%offered traffic~$\rho$ and capacity for $h$~flows.
%
%\[
%\lambda (s_{i-1} - s_i) = \frac{i}{\beta} (s_i - s_{i+1}),
%\hspace*{.4in} i = 1, \dots, h,
%\]
%which may be rewritten as
%\[
%\lambda p_{i-1} = \frac{i}{\beta} p_i,
%\hspace*{.4in} i = 1, \dots, h,
%\]
%yielding
%\[
%p_i = p_0 \frac{\rho^i}{i!},
%\hspace*{.4in} i = 1, \dots, h,
%\]
%with
%\[
%p_0 = \left[\sum_{i = 0}^{h} \frac{\rho^i}{i!}\right]^{- 1},
%\]
%with $\rho = \lambda \beta$.
%This agrees with the known exact truncated Poisson stationary
%distribution in this case: the numbers of active flows at the various
%servers evolve as in $n$~independent Erlang-B loss systems each
%with load~$\rho$ and capacity for $h$~flows.

\subsection{Scheme (II): pull-based flow assignment}
\label{ap:II}

In order to determine the fixed point of Equations~(\ref{fp1})
for Scheme (II), we distinguish three cases, depending on whether
$\rho < l$, $l \leq \rho \leq h$, or $\rho > h$.\\

\noindent $\bullet$ \textbf{Case $\rho < l$}.

Since in stationarity the average number of active flows per server
is~$\rho$, there must be many servers with less than $l$~flows,
and hence there must be many invite messages available at all times.
This corresponds to case~(i) above, and yields the fixed-point equation:
\[
\rho \frac{s_{i-1} - s_i}{1 - s_l} = i (s_i - s_{i+1}),
\hspace*{.4in} i = 1, \dots, l,
\]
with $s_{l+1} = 0$, which may be equivalently written as
\[
\rho \frac{p_{i-1}}{1 - p_l} = i p_i
\hspace*{.4in} i = 1, \dots, l.
\]

Summing the above equations over $i = 1, \dots, l$, we obtain
\[
\frac{\rho}{1 - p_l} \sum_{i = 0}^{l - 1} p_i = \sum_{i = 0}^{l} i p_i.
\]
This shows that the normalization condition $\sum_{i = 0}^{l} p_i = 1$
is equivalent to $\sum_{i = 0}^{l} i p_i = \rho$, reflecting that the
average number of active flows per server in stationarity equals~$\rho$.

Rewriting the above equations, we obtain
\[
p_i = \frac{\sigma}{i} p_{i-1},
\hspace*{.4in} i = 1, \dots, l,
\]
with
\begin{equation}
\sigma = \frac{\rho}{1 - p_l},
\label{sigmaiII}
\end{equation}
yielding
\[
p_i = p_0 \frac{\sigma^i}{i!},
\hspace*{.4in} i = 0, \dots, l,
\]
with
\[
p_0 = \left[\sum_{i = 0}^{l} \frac{\sigma^i}{i!}\right]^{- 1}.
\]

We observe that the probabilities~$p_i$ correspond to the
stationary occupancy distribution of an Erlang loss system
with load~$\sigma$ and capacity~$l$.
In particular,
\begin{equation}
p_l = \frac{\frac{\sigma^l}{l!}}{\sum_{i = 0}^{l} \frac{\sigma^i}{i!}} =
\Erl(\sigma; l),
%\frac{\frac{1}{l!} \left(\frac{\rho}{1 - p_l}\right)^l}
%{\sum_{i = 0}^{l} \frac{1}{i!} \left(\frac{\rho}{1 - p_l}\right)^i}
\label{plII1}
\end{equation}
where $\Erl(a; h)$ denotes the blocking probability in an Erlang
loss system with load~$a$ and capacity~$h$.
The relation $\rho = (1 - p_l) \sigma$ thus implies that the amount
of carried traffic is~$\rho$, which is consistent with the fact that
the average number of active flows per server in stationarity
equals~$\rho$.
(Denoting the blocked traffic by $\Delta\rho = \sigma - \rho$,
this may also be written in the form
\[
p_l = \frac{\Delta\rho}{\rho + \Delta\rho} = 1 - \frac{\rho}{\sigma},
\]
with
\[
\Delta\rho = p_l \sigma = (\rho + \Delta\rho) \Erl(\rho + \Delta\rho; l),
\]
which reveals a connection with an Erlang loss system with retrials.)
We deduce that $\sigma$ is the offered traffic volume in an Erlang loss
system with capacity~$l$ for which the carried traffic equals $\rho < l$.
Since the carried traffic in such a system is a strictly increasing
continuous function of the offered traffic, drops to~$0$ as the
offered traffic vanishes, and approaches~$l$ as the offered traffic
tends to infinity, we may conclude that $\sigma$ exists and is unique.\\

%Specifically, Equations~(\ref{sigmaiII}), (\ref{plII1}) and~(\ref{plII2})
%imply that $\sigma$ is a root of the equation
%\[
%\frac{\sum_{i = 0}^{l} i \frac{\sigma^i}{i!}}
%{\sum_{i = 0}^{l} \frac{\sigma^i}{i!}} =
%\frac{\sigma \sum_{i = 0}^{l - 1} i \frac{\sigma^i}{i!}}
%{\sum_{i = 0}^{l} \frac{\sigma^i}{i!}} =
%\sigma [1 - \Erl(\sigma; l)] = \rho,
%\]

\noindent $\bullet$ \textbf{Case $l \leq \rho < h$}.

In this case, there are both invite and dis-invite messages generated
in stationarity, but the former are instantly used, while the latter
naturally disappear once the number of flows at the corresponding
server drops below~$h$ again.

This corresponds to case~(ii) above, and yields the fixed-point equation:
\[
\tilde\lambda \beta \frac{s_{i-1} - s_i}{1 - s_h} = i (s_i - s_{i+1}),
\hspace*{.4in} i = l + 1, \dots, h,
\]
with $s_l = 1$ and $s_i = 0$ for all $i \geq h + 1$,
which may equivalently be written as
\[
\tilde\lambda \beta \frac{p_{i-1}}{1 - p_h} = i p_i,
\hspace*{.4in} i = l + 1, \dots, h.
\]

Summing the above equations over $i = l + 1, \dots, h$, we obtain
\[
\frac{\rho - l p_l}{1 - p_h} \sum_{i = l}^{h - 1} p_i =
\sum_{i = l}^{h} i p_i - l p_l.
\]
This shows that the normalization condition $\sum_{i = l}^{h} p_i = 1$
is equivalent to $\sum_{i = l}^{h} i p_i = \rho$, reflecting that the
average number of active flows per server in stationarity equals~$\rho$.

Rewriting the above equations, we obtain
\begin{equation}
p_i = \frac{\sigma}{i} p_{i-1},
\hspace*{.4in} i = l + 1, \dots, h,
\label{piII1}
\end{equation}
with
\begin{equation}
\sigma = \frac{\tilde\lambda \beta}{1 - p_h} =
\frac{(\lambda - l p_l / \beta) \beta}{1 - p_h} =
\frac{\rho - l p_l}{1 - p_h},
\label{sigmaiiII}
\end{equation}
yielding
\begin{equation}
p_i = p_l \frac{\sigma^{i - l}}{i (i - 1) \dots (l+1)} =
p_l \frac{\sigma^{i - l} l!}{i!},
\hspace*{.4in} i = l, \dots, h,
\label{piII2}
\end{equation}
with
\begin{equation}
p_l = \left[\sum_{i = l}^{h} \frac{\sigma^{i - l} l!}{i!}\right]^{- 1}.
%\left[\sum_{i = 0}^{h - l} \frac{\sigma^i l!}{(i + l)!}\right]^{- 1}
\label{plII}
\end{equation}
%\[
%p_l = \left[\sum_{i = l}^{h} \frac{\left(\frac{\rho - l p_l}{1 - p_h}\right)^{i - l} l!}{i!}\right]^{- 1} =
%\left[\sum_{i = 0}^{h - l} \frac{\left(\frac{\rho - l p_l}{1 - p_h}\right)^i l!}{(i + l)!}\right]^{- 1}.
%\]

The probabilities $p_i$ may be interpreted as stationary occupancy
distribution of a modified Erlang loss system with load~$\sigma$
and capacity~$h$, where departing users are replaced by dummy users
to prevent the number of users from falling below~$l$.
Specifically, Equations~(\ref{sigmaiII}), (\ref{piII2}) and~(\ref{plII})
imply that $\sigma$ is a root of the equation
\[
\begin{split}
\sum_{i = h}^{l} i p_i &= 
\frac{\sum_{i = l}^{h} i \frac{\sigma^{i - l}}{i!}}
{\sum_{i = l}^{h} \frac{\sigma^{i - l}}{i!}} =
%\frac{\sum_{i = 1}^{h} \frac{\sigma^{i - l}}{(i - 1)!}}
%{\sum_{i = l}^{h} \frac{\sigma^{i - l}}{i!}} =
\frac{\sigma \sum_{i = l - 1}^{h - 1} \frac{\sigma^{i - l}}{i!}}
{\sum_{i = l}^{h} \frac{\sigma^{i - l}}{i!}} \\
&=
\sigma \left[1 + \frac{\frac{\sigma^{- 1}}{(l - 1)!} - \frac{\sigma^{h - l}}{h!}}
{\sum_{i = l}^{h} \frac{\sigma^{i - l}}{i!}}\right] \\
&=
\sigma [1 - p_h + \LRE(\sigma; l; h)] \\
&=
\sigma [1 - p_h] + l p_l \\
&= \rho,
\end{split}
\]
with
\[
\LRE(\sigma; l; h) = \frac{l p_l}{\sigma} = \frac{\mu l p_l}{\lambda}
\]
representing the average number of dummy users created per regular
user in the above-described modified Erlang loss system.
Thus, $\sigma$ equals the offered traffic volume in such a system
for which the carried traffic equals $\rho \in [l, h)$.
As before, the carried traffic is a strictly increasing continuous
function of the offered traffic, drops to~$l$ when the offered traffic
vanishes, and tends to~$h$ as the offered traffic goes to infinity,
and hence we may conclude that $\sigma$ exists and is unique.

In case $h = l + 1$, i.e., $l = k^*$, this yields $p_l = h - \rho$
and $p_h = \rho - l$, i.e., $s_l = 1$ and $s_h = \rho - l$.

In case $l = 0$, we obtain
\[
p_i = p_0 \frac{\sigma^i}{i!},
\hspace*{.4in} i = 0, \dots, h,
\]
with
\[
p_0 = \left[\sum_{i = 0}^{h} \frac{\sigma^i}{i!}\right]^{- 1},
\]
and
\[
\sigma = \frac{\rho}{1 - p_h}.
\]

In particular,
\[
p_h = \frac{\frac{\sigma^h}{h!}}{\sum_{i = 0}^{h} \frac{\sigma^i}{i!}},
%\frac{\frac{1}{h!} \left(\frac{\rho}{1 - p_h}\right)^h}
%{\sum_{i = 0}^{h} \frac{1}{i!} \left(\frac{\rho}{1 - p_h}\right)^i}
\]
or equivalently,
\[
p_h = \Erl(\sigma; h), %\Erl(\frac{\rho}{1 - p_h}; h)
\]

Denoting $\Delta\rho = \sigma - \rho$, this may also be written
in the form
\[
p_h = \frac{\Delta\rho}{\rho + \Delta\rho} = 1 - \frac{\rho}{\sigma},
\]
with
\[
\Delta\rho = p_h \sigma = (\rho + \Delta\rho) \Erl(\rho + \Delta\rho; h),
\]
which reveals a connection with an Erlang loss system with retrials.

In particular, in case $l = 0$ and $h = \infty$, we find
\[
p_i = \ee^{- \rho} \frac{\rho^i}{i!},
\hspace*{.4in} i \geq 0.
\]
This reflects a far stronger property:
the numbers of active flows at the various servers are in fact
independent and Poisson distributed with parameter~$\rho$ in the
pre-limit system for any~$n$.\\

\noindent $\bullet$ \textbf{Case $\rho \geq h$}.

In this case, a server that sees the number of active flows drop from~$h$
to $h - 1$ and revokes its outstanding dis-invite message, will immediately
be assigned a newly arriving flow, and re-issue its dis-invite message.
In stationarity all servers have outstanding dis-invite messages.

This corresponds to case~(iii) above, and yields the fixed-point equation:
\[
\tilde\lambda (s_{i-1} - s_i) = i (s_i - s_{i+1}) / \beta,
\hspace*{.4in} i \geq h + 1,
\]
with $s_h = 1$, which may equivalently be written as
\[
\tilde\lambda p_{i-1} = i p_i / \beta,
\hspace*{.4in} i \geq h + 1.
\]

Summing the above equations over $i \geq h + 1$, we obtain
\[
(\rho - h p_h) \sum_{i = h}^{\infty} p_i = \sum_{i = h}^{\infty} i p_i - h p_h.
\]
This shows that the normalization condition $\sum_{i = h}^{\infty} p_i = 1$
is equivalent to $\sum_{i = h}^{\infty} i p_i = \rho$, reflecting that the
average number of active flows per server in stationarity equals~$\rho$.

Rewriting the above equations as
\begin{equation}
p_i = \frac{\sigma}{i} p_{i-1},
\hspace*{.4in} i \geq h + 1,
\label{piII3}
\end{equation}
with
\begin{equation}
\sigma = \tilde\lambda \beta =
(\lambda - h p_h / \beta) \beta = \rho - h p_h,
\label{sigmaiiiII}
\end{equation}
we obtain
\[
p_i = p_h \frac{\sigma^{i - h}}{i (i - 1) \dots (h+1)} =
p_h \frac{\sigma^{i - h} h!}{i!},
\hspace*{.4in} i \geq h,
\]
with
\begin{equation}
p_h = \left[\sum_{i = h}^{\infty} \frac{\sigma^{i - h} h!}{i!}\right]^{- 1}.
%\left[\sum_{i = 0}^{\infty} \frac{\sigma^i h!}{(i + h)!}\right]^{- 1}
\label{phII}
\end{equation}
%\[
%p_h = \left[\sum_{i = h}^{\infty} \frac{(\rho - h p_h)^{i - h} h!}{i!}\right]^{- 1} =
%\left[\sum_{i = 0}^{\infty} \frac{(\rho - h p_h)^i h!}{(i + h)!}\right]^{- 1}.
%\]

As before, the probabilities $p_i$ may be interpreted as the
stationary occupancy distribution of a modified Erlang system
where departing users are replaced by dummy users to prevent the
number of users from falling below~$h$.
Specifically, Equations~(\ref{sigmaiiiII}) and~(\ref{phII}) imply
that $\sigma$ is a root of the equation
\[
\begin{split}
\sum_{i = h}^{\infty} i p_i =
\frac{\sum_{i = h}^{\infty} i \frac{\sigma^{i-h}}{i!}}
{\sum_{i = h}^{\infty} \frac{\sigma^{i-h}}{i!}} &=
\frac{\frac{1}{h!} + \sigma \sum_{i = h}^{\infty} \frac{\sigma^{i-h}}{i!}}
{\sum_{i = h}^{\infty} \frac{\sigma^{i-h}}{i!}} \\
&=
\sigma [1 + \LRE(\sigma; h)] = \rho,
\end{split}
\]
with
\[
\LRE(\sigma; h) = \frac{h p_h}{\sigma} = \frac{\mu h p_h}{\lambda}
\]
representing the average number of dummy users created per regular
user in the above-described modified Erlang system.
Thus, $\sigma$ equals the offered traffic volume in such a system
for which the carried traffic equals $\rho > h$.
As before, the carried traffic is a strictly increasing continuous
function of the offered traffic, drops to~$h$ when the offered traffic
vanishes, and grows without bound when the offered traffic goes to
infinity, and hence we may conclude that $\sigma$ exists and is unique.

\subsection{Scheme (IV): random flow assignment with threshold-based
flow transfer to a server with an invite message}
\label{ap:IV}

In order to determine the fixed point of Equations~(\ref{fp1})
for Scheme (III), we distinguish three cases, depending on whether
$\rho < l$, $l \leq \rho \leq h$, or $\rho > h$.\\

\noindent $\bullet$ \textbf{Case $\rho < l$}.

Since in stationarity the average number of active flows per server
is~$\rho$, there must be many servers with less than $l$~flows.
This corresponds to case~(i) above, and yields the fixed-point equation:
\[
\rho (s_{i-1} - s_i) \frac{1 - s_l + s_h}{1 - s_l} = i (s_i - s_{i+1}),
\hspace*{.4in} i = 1, \dots, l,
\]
\[
\rho (s_{i-1} - s_i) = i (s_i - s_{i+1}),
\hspace*{.4in} i = l + 1, \dots, h,
\]
with $s_{h+1} = 0$, which may equivalently be written as
\[
\rho p_{i-1} \frac{1 - s_l + s_h}{1 - s_l} = i p_i,
\hspace*{.4in} i = 1, \dots, l,
\]
\[
\rho p_{i-1} = i p_i,
\hspace*{.4in} i = l + 1, \dots, h.
\]

Summing the above equations over $i = 1, \dots, h$, we obtain
\[
\rho \left(\frac{1 - s_l + s_h}{1 - s_l} \sum_{i = 1}^{l} p_{i - 1} +
\sum_{i = l + 1}^{h} p_{i-1}\right) = \sum_{i = 1}^{h} i p_i,
\]
which may be simplified to
\[
\rho \sum_{i = 0}^{h} p_i = \sum_{i = 0}^{h} i p_i.
\]
This shows that the normalization condition $\sum_{i = 0}^{h} p_i = 1$
is equivalent to $\sum_{i = 0}^{h} i p_i = \rho$, reflecting that the
average number of active flows per server in stationarity equals~$\rho$.

Rewriting the above equations, we obtain
\[
p_i = \frac{\sigma}{i} p_{i-1},
\hspace*{.4in} i = 1, \dots, l,
\]
with
\begin{equation}
\sigma = \frac{(1 - s_l + s_h) \rho}{1 - s_l},
\label{sigmaiIII}
\end{equation}
and
\[
p_i = \frac{\rho}{i} p_{i-1},
\hspace*{.4in} i = l + 1, \dots, h,
\]
yielding
\begin{equation}
p_i = p_0 \left\{ \begin{array}{ll} \frac{\sigma^i}{i!} & i = 0, \dots, l, \\
\frac{\sigma^l \rho^{i - l}}{i!} & i = l + 1, \dots, h \end{array} \right.
\label{piIII1}
\end{equation}
with
\begin{equation}
p_0 = \left[\sum_{i = 0}^{l} \frac{\sigma^i}{i!} +
\sum_{i = l + 1}^{h} \frac{\sigma^l \rho^{i - l}}{i!}\right]^{- 1}.
\label{p0III}
\end{equation}

Equations~(\ref{sigmaiIII})--(\ref{p0III}) imply that $\sigma$ is
a root of the equation
\[
\begin{split}
H(\sigma) &= \sum_{i = 0}^{h} i p_i \\
&=
\frac{\sum_{i = 0}^{l} i \frac{\sigma^i}{i!} +
\sum_{i = l + 1}^{h} i \frac{\sigma^l \rho^{i - l}}{i!}}
{\sum_{i = 0}^{l} \frac{\sigma^i}{i!} +
\sum_{i = l + 1}^{h} \frac{\sigma^l \rho^{i - l}}{i!}} \\
&=
\frac{\sigma \sum_{i = 0}^{l - 1} \frac{\sigma^i}{i!} +
\rho \sum_{i = l}^{h - 1} \frac{\sigma^l \rho^{i - l}}{i!}}
{\sum_{i = 0}^{l} \frac{\sigma^i}{i!} +
\sum_{i = l + 1}^{h} \frac{\sigma^l \rho^{i - l}}{i!}} \\
&=
\sigma [1 - s_l + s_h] + \rho s_l \\
&= \rho.
\end{split}
\]
It is easily verified that the function $H(\sigma)$ is strictly
increasing and continuous, drops to~$0$ as $\sigma \downarrow 0$,
and tends to a value above~$h$ as $\sigma \to \infty$,
and hence we may conclude that $\sigma$ exists and is unique.\\

\noindent $\bullet$ \textbf{Case $l \leq \rho < h$}.

This corresponds to case~(ii) above, and yields the fixed-point equation: for $i = l, l + 1, \dots, h$
\[
\begin{split}
\lambda (s_{i-1} - s_i) + \tilde\lambda \tilde{q}_{i-1}(s) &=
\lambda (s_{i-1} - s_i) + \tilde\lambda \frac{s_{i-1} - s_i}{1 - s_h} \\
&=
\left(\lambda + \frac{\tilde\lambda}{1 - s_h}\right) (s_{i-1} - s_i) \\
&=
i (s_i - s_{i+1}) / \beta,
\end{split}
\]
with $s_l = 1$ and $s_i = 0$ for all $i \geq h + 1$, which may
equivalently be written as
\[
\left(\lambda + \frac{\tilde\lambda}{1 - p_h}\right) p_{i-1} =
i p_i / \beta,
~i = l + 1, \dots, h.
\]

Thus we obtain
\begin{equation}
p_i = \frac{\sigma}{i} p_{i-1},
~i = l + 1, \dots, h,
\label{piIII2}
\end{equation}
with
\begin{equation}
\begin{split}
\sigma &=
\lambda + \frac{\tilde\lambda}{1 - p_h} \\
&=
\frac{\lambda (1 - p_h) + \lambda s_h  - l p_l / \beta}{1 - p_h} \\
&= \frac{(\lambda - l p_l / \beta) \beta}{1 - p_h} \\
&=
\frac{\rho - l p_l}{1 - p_h}.
\label{sigmaiiIII}
\end{split}
\end{equation}

Now observe that Equations~(\ref{piIII2}) and~(\ref{sigmaiiIII})
coincide with the corresponding Equations~(\ref{piII1})
and~(\ref{sigmaiiII}) for Scheme~(II).
Hence the fixed point for Scheme (III) is identical,
and may be interpreted in a similar fashion.\\

\noindent $\bullet$ \textbf{Case $\rho \geq h$}.

In this case, a server that sees the number of flows drop from~$h$
to $h - 1$ will immediately be assigned a newly arriving flow.

This corresponds to case~(iii) above, and yields the fixed-point equation:
\[
\tilde\lambda (s_{i-1} - s_i) = i (s_i - s_{i+1}) / \beta,
\hspace*{.4in} i \geq h + 1,
\]
with $s_h = 1$, which may equivalently be written as
\[
\tilde\lambda p_{i-1} = i p_i / \beta,
\hspace*{.4in} i \geq h + 1.
\]
Thus we obtain
\begin{equation}
p_i = \frac{\sigma}{i} p_{i-1},
\hspace*{.4in} i \geq h + 1,
\label{piIII3}
\end{equation}
with
\begin{equation}
\sigma = \tilde\lambda \beta^2 = (\lambda - h p_h / \beta) \beta = \rho - h p_h.
\label{sigmaiiiIII}
\end{equation}

Just like observed above, Equations~(\ref{piIII3})
and~(\ref{sigmaiiiIII}) coincide with the corresponding
Equations~(\ref{piII3}) and~(\ref{sigmaiiiII}) for Scheme~(II).
Hence the fixed point for Scheme (IV) is identical,
and may again be interpreted in a similar fashion.

\subsection{Scheme (V): random flow assignment with threshold-based
flow transfer to the least-loaded server}
\label{ap:V}

In order to determine the fixed point of Equations~(\ref{fp1})
for Scheme~(V) in case $\alpha = 0$, we distinguish two cases,
depending on whether $\rho < h$ or $\rho \geq h$.\\

\noindent $\bullet$ \textbf{Case $\rho < h$}.

This corresponds to case~(i) above, and yields the fixed-point equation:
\[
\begin{split}
\lambda q_m(s) &= \lambda (s_m - s_{m+1}) + \tilde\lambda \\
&=
\lambda s_h + (\lambda - \frac{m}{\beta}) (s_m - s_{m+1}) \\
&=
\frac{m + 1}{\beta} (s_{m+1} - s_{m+2}), 
\end{split}
\]
i.e.,
\[
\rho s_h + (\rho - m) (s_m - s_{m+1}) = (m + 1) (s_{m+1} - s_{m+2}),
\]
and
\[
\lambda q_{i-1}(s) = \lambda (s_{i-1} - s_i) = \frac{i}{\beta} (s_i - s_{i+1}),
~i = m + 2, \dots, h,
\]
with $s_i = 1$ for all $i \leq m$ and $s_i = 0$ for all $i \geq h + 1$,
which may equivalently be written as
\[
\rho p_h + (\rho - m) p_m = (m + 1) p_{m+1},
\]
and
\[
\rho p_{i-1} = i p_i,
~i = m + 2, \dots, h.
\]

Summing the above equations over $i = m + 1, \dots, h$, we obtain
\[
\rho \sum_{i = m}^{h} p_i = \sum_{i = 1}^{h} i p_i.
\]
This shows that the normalization condition $\sum_{i = m}^{h} p_i = 1$
is equivalent to $\sum_{i = m}^{h} i p_i = \rho$, reflecting that the
average number of active flows per server in stationarity equals~$\rho$.

Rewriting the above equations, we obtain
\[
p_i = \frac{\rho}{i} p_{i-1},
~i = m + 2, \dots, h,
\]
yielding
\[
p_h = \frac{\rho^{h - i}}{(i + 1) (i + 2) \dots h} p_i =
\frac{\rho^{h - i} i!}{h!} p_i,
~i = m + 1, \dots, h,
\]
or equivalently,
\[
p_i = \frac{h!}{\rho^{h - i} i!} p_h,
~i = m + 1, \dots, h,
\]
and in particular,
\[
p_{m + 1} = \frac{h!}{\rho^{h - m - 1} (m + 1)!} p_h.
\]
Thus,
\[
\begin{split}
p_m &= \frac{(m + 1) p_{m+1} - \rho p_h}{\rho - m} \\
&=
\frac{\rho p_h}{\rho - m} \left[\frac{h! \rho^m}{m! \rho^h - 1}\right] \\
&=
\frac{\rho p_h}{\rho - m} \left[\frac{h! \rho^{m - 1}}{(m - 1)! \rho^h}
\frac{\rho}{m} - 1\right].
\end{split}
\]

Note that $\frac{\rho^m}{m!} > \frac{\rho^h}{h!}$ is necessary
and sufficient for $p_m > 0$, ensuring
\[
\tilde\lambda = \lambda p_h - \frac{m}{\beta} p_m < \lambda p_h.
\]

Also, if $\frac{\rho^{m - 1}}{(m - 1)!} > \frac{\rho^h}{h!}$, then
\[
p_m > \frac{\rho p_h}{\rho - m} \left[\frac{\rho}{m} - 1\right] =
\frac{\rho p_h}{m},
\]
which would imply
\[
\tilde\lambda = \lambda p_h - \frac{m}{\beta} p_m < 0.
\]
Thus, we must have $\frac{\rho^{m - 1}}{(m - 1)!} < \frac{\rho^h}{h!}$,
and hence the value of~$m$ is uniquely determined as
\[
\min\{i: \frac{\rho^i}{i!} > \frac{\rho^h}{h!}\}.
\]\\

\noindent $\bullet$ \textbf{Case $\rho \geq h$}.

This corresponds to case~(ii) above, and yields the fixed-point equation:
\[
\begin{split}
\lambda q_m(s) = \tilde\lambda &= \lambda - \frac{m}{\beta} (s_m - s_{m+1}) \\
&=
\frac{m + 1}{\beta} (s_{m+1} - s_{m+2}), 
\end{split}
\]
i.e.,
$
\rho - m (s_m - s_{m+1}) = (m + 1) (s_{m+1} - s_{m+2}),
$
with $s_i = 1$ for all $i < m$ and $s_i = 0$ for all $i \geq m + 2$,
which may equivalently be written as
\[
\rho - m p_m = (m + 1) p_{m+1},
\]
with $p_i = 0 $ for all $i \neq m, m + 1$.

This yields $p_m = m + 1 - \rho$, $p_{m+1} = \rho - m$,
with $m = k^* = \lfloor \rho \rfloor$.

\end{document}